\documentclass[11pt]{amsart}
\usepackage{graphicx}
\usepackage{amscd}
\usepackage{amsmath}
\usepackage{amsfonts}
\usepackage{amssymb}
\usepackage{dsfont}
\usepackage{setspace}
\usepackage{enumerate}         
\usepackage{color}
\usepackage{url}
\usepackage{amsthm}
\usepackage{hyperref}
\usepackage{mathrsfs}
\usepackage{bm}
\usepackage{xy}
\usepackage{color}
\usepackage{subfig}
\usepackage{bbm}
\usepackage{etoolbox}
\usepackage{array}
\usepackage{eurosym}
\allowdisplaybreaks[4]

\usepackage{natbib}

\usepackage{geometry}
\usepackage{tikzit}

\tikzstyle{red dot}=[fill=red, draw=black, shape=circle]
\tikzstyle{green dot}=[fill=green, draw=black, shape=circle]
\tikzstyle{blue dot}=[fill=blue, draw=black, shape=circle]
\tikzstyle{emptyRectanle}=[fill=white, draw=black, shape=rectangle]
\tikzstyle{emptyCircle}=[fill=white, draw=black, shape=circle]
\tikzstyle{Big Square}=[fill=white, draw=black, shape=rectangle, minimum width=2cm, minimum height=1cm]

\tikzstyle{new edge style 0}=[->]
\tikzstyle{new edge style 1}=[-]
\tikzstyle{doubleArrow}=[<->]
\tikzstyle{new edge style 2}=[<->, draw=blue]

\theoremstyle{plain}
\newtheorem{theorem}{Theorem}[section]

\newtheorem{corollary}[theorem]{Corollary}

\newtheorem{lemma}[theorem]{Lemma}

\newtheorem{proposition}[theorem]{Proposition}

\newtheorem{definition}[theorem]{Definition}

\newtheorem{assumption}[theorem]{Assumption}

\theoremstyle{remark}
\newtheorem{remark}[theorem]{Remark}

\numberwithin{equation}{section}

\newcommand{\rsto}{]\!\kern-1.8pt ]}
\newcommand{\lsto}{[\!\kern-1.7pt [}

\renewcommand{\labelenumi}{\rm{(\roman{enumi})}}

\numberwithin{equation}{section}

\newcommand{\Ind}[1]{\mathrm{1}_{\left\{#1\right\}}}

\newcommand{\GG}{\mathbb{G}}

\newcommand{\RR}{\mathbb{R}}
\newcommand{\QQ}{\mathbb{Q}}

\newcommand{\PP}{\mathbb{P}}
\newcommand{\NN}{\mathbb{N}}
\newcommand{\EE}{\mathbb{E}}

\newcommand{\cU}{\mathcal{U}}

\newcommand{\cB}{\mathcal{B}}

\newcommand{\cG}{\mathcal{G}}

\newcommand{\cN}{\mathcal{N}}

\newcommand{\cP}{\mathcal{P}}

\newcommand{\Excond}[3]{\mathbb{E}^{#1}\left[\left.#2\right|#3\right]}  

\renewcommand{\cite}{\citet}

\makeatletter
\patchcmd{\@maketitle}
  {\ifx\@empty\@dedicatory}
  {\ifx\@empty\@date \else {\vskip3ex \centering\footnotesize\@date\par\vskip1ex}\fi
   \ifx\@empty\@dedicatory}
  {}{}
\patchcmd{\@adminfootnotes}
  {\ifx\@empty\@date\else \@footnotetext{\@setdate}\fi}
  {}{}{}
\newcommand{\subjclassname@JEL}{JEL Classification}
\makeatother

\begin{document}

\title[Ccy-HJM]{Cross-Currency Heath-Jarrow-Morton Framework in the Multiple-Curve Setting}

\author{Alessandro Gnoatto}
\address[Alessandro Gnoatto]{University of Verona, Department of Economics, \newline
\indent Via Cantarane 24, 37129 Verona, Italy}
\email{alessandro.gnoatto@univr.it}

\author{Silvia Lavagnini}
\address[Silvia Lavagnini]{BI Norwegian Business School, Department of Data Science and Analytics, \newline
\indent Nydalsveien 37, 0484 Oslo, Norway}
\email{silvia.lavagnini@bi.no}

\begin{abstract}
We provide a general HJM framework for forward contracts written on abstract market indices with arbitrary fixing and payment adjustments, and featuring collateralization in any currency denominations. In view of this, we first provide a thorough study of cross-currency markets in the presence of collateral and incompleteness. Then we give a general treatment of collateral dislocations by
describing the instantaneous cross-currency basis spreads by means of HJM models, for which we derive  appropriate drift conditions.  The framework obtained allows us to simultaneously cover forward-looking risky IBOR rates, such as EURIBOR, and backward-looking rates based on overnight rates, such as SOFR. Due to the discrepancies in market conventions of different currency areas created by the benchmark transition, this is pivotal for describing portfolios of interest-rate products that are denominated in multiple currencies. As an example of contract simultaneously depending on all the risk factors that we describe within our framework, we treat cross-currency swaps using our proposed abstract indices.
\end{abstract}

\keywords{HJM, FX, cross-currency basis, multiple curves,  benchmark transition, SOFR, collateral}
\thanks{{\em Acknowledgements.} The authors are grateful to the participants of the minisymposium on ``Interest rate modelling'' at the 2023 SIAM Conference on Financial Mathematics and Engineering in Philadelphia, the joint LMU TUM Oberseminar Finanz- und Verischerungsmathematik in Munich and the XXV Workshop on Quantitative Finance in Bologna where this work has been presented. The authors are also grateful to two anonymous referees for helpful comments.}
\subjclass[2010]{91G30, 91B24, 91B70. \textit{JEL Classification} E43, G12}

\date{\today}

\maketitle

\section{Introduction}
Benchmark reforms have introduced significant discrepancies among interest rate option markets of different currency areas. In the US market, for example, caps and floors are currently written on a compounded version of the secured overnight financing rate (SOFR), whereas in the EUR area, the unsecured EURIBOR rate is still the market standard underlying.  This poses a significant challenge when considering a portfolio of interest rate derivatives that are denominated in multiple currencies: 
a model suitable for such portfolio-wide calculations should be able to simultaneously describe forward-looking credit-sensitive rates on the one hand, and forward-looking and backward-looking overnight-based interest rates on the other hand.

In this paper, we provide a general Heath-Jarrow-Morton (HJM) framework to describe forward contracts written on abstract market indices. Our setting allows for indices with arbitrary fixing and payment adjustments and indices on any asset class, so to accommodate the benchmark transition. Moreover, it allows for multiple currencies, meaning that the cash flows from a contract and the collateralization may be denominated in arbitrary combinations of currencies. Thus we simultaneously extend the literature on multiple-curve valuation with collateral, on interest rate and on cross-currency modeling, and we define a solid foundation for analysing the benchmark transition.

The basic building block for any term-structure model is given by a family of term structures of zero-coupon bonds (ZCBs). In the present paper, each term structure of ZCBs is characterized by a base currency and by the choice of the collateral currency. In a framework with $L$ currency areas, for any fixed time horizon $T>0$, this entails to model zero-coupon bonds $B^{k_0,k_3}(\cdot,T)$ denominated in units of a currency $k_0$ and involving an exchange of collateral in a currency $k_3$, for each pair of currencies $1\le k_0, k_3 \le L$. In Section \ref{sec:case2} we derive the following pricing formula under the so-called domestic measure $\QQ^{k_0}$:
\begin{align*}
    B^{k_0,k_3}(t,T)=\mathbb{E}^{\QQ^{k_0}}\left[\left.e^{-\int_t^T (r^{c,k_0}_s+q^{k_0,k_3}_s) d s}\right|\cG_t\right]=e^{-\int_t^T (f^{c,k_0}_t(u)+q^{k_0,k_3}_t(u))du}, \quad \mbox{for }t\le T, \ 1\leq k_0,k_3\leq L,
    \end{align*}
where $\GG=\left(\cG_t\right)_{t\in [0,T]}$ is a suitable filtration introduced in the sequel. This result, obtained under an extension of the setting of \cite{gnoSei2021}, shows that to describe ZCBs one must define both a model for the instantaneous collateral forward rate of currency $k_0$, $f^{c,k_0}=(f^{c,k_0}_t(T))_{t\in[0,T]}$, and a model for the instantaneous cross-currency basis spread, $q^{k_0, k_3}=(q^{k_0, k_3}_t(T))_{t\in[0,T]}$, together with their associated short rates, $r^{c,k_0}=(r^{c,k_0}_t)_{t\geq 0}$ and $q^{k_0,k_3}=(q^{k_0,k_3}_t)_{t\geq 0}$, respectively. 
The resulting large parametric family of multiple term structures of zero-coupon bonds allows for the definition of forward measures in view of term-structure modeling. In particular, this leads to \textit{extended} forward measures, generalizing the approach of \cite{LyasMer2019}. A comprehensive description and analysis of the considered general cross-currency market can be found in Sections \ref{sec:MultiCurrTrading}-\ref{sec:measureChanges}

The main objective of the present work is to set the described cross-currency market ``in motion'' by means of a Heath-Jarrow-Morton framework,  \cite{hea92}. On a filtered probability space $\left(\Omega,\cG,\GG,\PP\right)$, we provide a general treatment of collateral dislocations by describing the instantaneous cross-currency basis spreads $(q_t^{k_0, k_3}(T))_{t \in[0, T]}$ according to 
\begin{align*}
    q_t^{k_0, k_3}(T)=q_0^{k_0, k_3}(T)+\int_0^t \alpha_s^{k_0, k_3}(T) d s+\int_0^t \sigma_s^{k_0, k_3}(T) d X_s,
\end{align*}
where $X=(X_t)_{t\geq 0}$ is an It\^o semimartingale. In doing so, we derive a new type of HJM drift condition for the dynamics above under a suitable forward measure, thus generalizing the cross-currency HJM framework of \cite{fushita09} to a semimartingale setting, see equation \eqref{eq:driftcondak3k0}.
Moreover, we extend the general multiple-curve HJM framework of \cite{Cuchiero2016} across two directions: we consider multiple-currencies thanks to our careful specification of the instantaneous cross-currency basis spreads $(q_t^{k_0, k_3}(T))_{t \in[0, T]}$, and we consider abstract indices as the target modeling quantities. In the single-currency approach of \cite{Cuchiero2016}, the focus is indeed limited to modelling forward-rate agreements (FRA) rates, i.e., using their notation, to modelling the quantities
\begin{align*}
    L_t(T, T+\delta)=\mathbb{E}^{\mathbb{Q}^{T+\delta}}\left[L_T(T, T+\delta) \mid \mathcal{G}_t\right], \qquad \mbox{for }t\le T,
\end{align*}
where $\mathbb{Q}^{T+\delta}$ denotes a $(T+\delta)$-forward measure with num\'eraire given by the overnight indexed swaps (OIS) bond $B(\cdot, T+\delta)$\footnote{This measure corresponds  in our notation to $\QQ^{T,k_0,k_0}$ with num\'eraire $B^{k_0,k_0}(\cdot, T+\delta)$, see Definition \ref{def:forwardmeasure}.}, and $L_T(T, T+\delta)$ is the spot IBOR rate at time $T$ for the interval $[T,T+\delta]$, for $\delta\ge 0$. We define instead abstract indices of the form
\begin{align*}   
I^{k_0,k_2,k_3}_t(T-\delta^f, T, T+\delta^p)=\frac{\Excond{\QQ^{k_0}}{\frac{B^{c,k_0,k_3}_t}{B^{c,k_0,k_3}_{T+\delta^p}}I^{k_2}_{T}(T-\delta^f, T, T+\delta^p)\mathcal{X}^{k_0,k_2}_{T+\delta^p}}{\cG_t}}{\Excond{\QQ^{k_0}}{\frac{B^{c,k_0,k_3}_t}{B^{c,k_0,k_3}_{T+\delta^p}}}{\cG_t}}, \qquad \mbox{for }t\le T,
\end{align*}
for some given $\delta^f,\delta^p\geq 0$ representing, respectively, the fixing and the payment adjustments, see Definition \ref{def:abstractForwardIndex}
and, in particular, equation \eqref{eq:genForwardRate}. In the equation above, the term $I^{k_2}_{T}(T-\delta^f, T, T+\delta^p)$ inside the conditional expectation is a spot abstract index denominated in currency $k_2$, with period starting in $T-\delta^f$, fixing in $T$ and payment in $T+\delta^p$. Moreover, the term $\mathcal{X}^{k_0,k_2}_{T+\delta^p}$ is the spot exchange rate between the currencies $k_0$ and $k_2$ at time $T+\delta^p$, and the cash account $B^{c,k_0,k_3}_t=e^{\int_0^t(r^{c,k_0}_s+q^{k_0,k_3}_s) d s}$ captures collateralization in multiple currencies.
By working with abstract indices, we encompass the case of classical IBOR rates, of new backward-looking indices based on overnight rates, such as SOFR, and other quantities such as inflation or commodity prices. This allows us to accommodate the situation where we simultaneously consider a market with standard forward-looking rates (e.g., EURIBOR for the EUR area or TIBOR for the JPY area) and backward-looking rates (e.g., SOFR-based rates for the USD area).

Due to the benchmark reforms in some jurisdictions, our framework is pivotal for the management of large portfolios of interest-rate products which are denominated in different currencies and are subject to different market conventions. This is a typical situation that is faced when computing xVA at the portfolio level by Monte Carlo simulations, see e.g. \cite{BiaGnoOli2019}. To illustrate this, let us recall the industry-standard formula for the credit valuation adjustment (CVA) computed over a netting set of contingent claims denominated in $L$ different currencies $k_2=1, \dots, L$, with cash flow streams $A^{n,k_2}$ and associated mark-to-market $S^{k_2}(A^{n,k_2})$, for $n=1,\ldots \mathtt{N}^{k_2}$ with $\mathtt{N}^{k_2}$ the number of streams for the currency $k_2$, with exchange of collateral $C^{k_3}$ in currency $k_3$, against a counterparty with hazard rate $\lambda=\left(\lambda_t\right)_{t\geq 0}$ and recovery rate $R$. For each $t \ge0$, this is given by the following formula:
\begin{align*}
    \mathrm{CVA}_t=(1-R)\mathbb{E}^{\mathbb{Q}^{k_0}}\left[\left.\int_t^Te^{-\int_t^u \left(r^{c,k_0}_s+\lambda_s\right) ds}\left(\sum_{k_2=1}^{L}\mathcal{X}^{k_0,k_2}_u\sum_{n=1}^{\mathtt{N^{k_2}}}S_u^{k_2}(A^{n,k_2})-\mathcal{X}^{k_0,k_3}_uC^{k_3}_u\right)^+\lambda_u du\right|\mathcal{G}_t\right].
\end{align*}
It is clear that, due to the netting agreement, the computation and risk management of CVA (and of xVA in general) is a non-linear portfolio-wide problem. This motivates the introduction of large scale simulation models simultaneously describing all risk factors affecting the evolution of the portfolio. It becomes also apparent that the model should be flexible enough to accommodate the different conventions prevailing in the different economies. This paper provides sound foundations for term-structure modeling in the current cross-currency market situation.

In the last part, we demonstrate the relevance of all the modeling quantities that we consider in the paper: as a test-bed for our framework, we describe cross-currency swap contracts using our proposed abstract indices, meaning that we can cover, for example, the situation of a legacy EURUSD cross-currency contract created before the LIBOR transition exchanging USD LIBOR against EURIBOR. According to the US LIBOR Act, market participants can choose the LIBOR fallback rate that they deem more appropriate. Hence for the USD leg, the agents may agree on the fallback proposed by the International Swaps and Derivatives Association (ISDA) which is based on SOFR, or they may choose an alternative benchmark such as AMERIBOR. Our framework is general enough to cover all the possible situations, so it represents the ideal setup for the valuation of a portfolio that typically combines legacy trades and new positions. We point out that cross-currency swaps are merely one of many possible examples of applications which are feasible within the context of the present framework. 

In the remaining part of the introduction, we review the most important market features that motivate the present work together with the existing literature.

\subsection{Violations of the covered interest rate parity} 
Consider a domestic agent endowed with an initial capital $\mathcal{N}^{k_0}$. 
At time $t\ge0$ the agent faces two investment alternatives. A first possibility would be to invest the initial capital for a horizon $\delta$ by lending on the ${k_0}$-unsecured market, thus earning the ${k_0}$-IBOR rate $L^{k_0}_t(t,t+\delta)$. Alternatively, the agent could enter at time $t$ into a $k_3$-foreign exchange forward with length $\delta$ and rate $\mathcal{X}^{k_0,k_3}_t(t+\delta)$. In this case, at time $t$ he/she would convert the amount $\mathcal{N}^{k_0}$ at the spot exchange rate $\mathcal{X}^{k_0, k_3}_t$ and lend the amount in foreign currency on the foreign unsecured market, where he/she would earn the unsecured $k_3$-IBOR rate $L^{k_3}_t(t,t+\delta)$. After the time $\delta$, the agent would then reconvert the amount using the foreign exchange forward rate $\mathcal{X}^{k_0,k_3}_t(t+\delta)$ agreed at time $t$. The combination of such an FX spot and an FX forward transaction is termed \emph{FX swap}. We say that the covered interest rate parity holds if the two strategies described deliver the same amount in domestic currency at the end of the period $\delta$. In particular, if the covered interest rate parity were to hold, then we would obtain the classical relation that links the market quote of the FX forward with the unsecured spot rates of the two currencies, namely
\begin{equation}
\label{eq:notirp}
\mathcal{X}^{k_0,k_3}_t(t+\delta)=\mathcal{X}^{k_0,k_3}_t\frac{1+\delta L_t^{k_0}(t,t+\delta)}{1+\delta L_t^{k_3}(t,t+\delta)}.
\end{equation}
Market data on FX swaps and FX forwards show however the systematic violations of \eqref{eq:notirp}.

We can similarly discuss the valuation of cross-currency swaps. These are long-term transactions that involve an exchange of cash flows between two agents, here denoted with $k_0$ for \emph{domestic} and $k_3$ for \emph{foreign}, over a schedule of dates, say $T_0, T_1, \dots, T_n$. Since the cash flows are indexed on the floating rates\footnote{Depending on the currency pairs involved, the floating rates could be secured or unsecured.} of two currencies, cross-currency swaps can be seen as a long-short position on two floating-rate bonds denominated in two different currencies. In particular, at time $T_0$ the two agents lend to each other the notional amounts $\mathcal{N}^{k_0}$ and $\mathcal{N}^{k_3}$ in domestic and foreign currency, respectively. Then, at each time $T_i$, $i=1, \dots, n$, the agents receive floating-rate interests for the notionals lent. In addition, at time $T_N$ the notionals are swapped back. If the covered interest rate parity were to hold, then the sum of the value of the two legs at time $t\le T_0$ should be zero in the absence of any adjustment. More precisely, taking the perspective of the domestic $k_0$ agent, we should observe that
\begin{equation}
\label{eq:notirp2}
\begin{aligned}
    0=&\mathcal{N}^{k_0}\left(-B^{k_0}(t,T_0)+\sum_{i=1}^{N} (T_i-T_{i-1})L^{k_0}_t(T_{i-1},T_i) B^{k_0}(t,T_i)+B^{k_0}(t,T_N)\right)\\
    &-\mathcal{X}^{k_0,k_3}_t\mathcal{N}^{k_3}\left(-B^{k_3}(t,T_0)+\sum_{i=1}^{N} (T_i-T_{i-1})L^{k_3}_t(T_{i-1},T_i) B^{k_3}(t,T_i)+B^{k_3}(t,T_N)\right),
\end{aligned}
\end{equation}
where $B^{k_0}(t,\cdot)$ and $B^{k_3}(t,\cdot)$ denote risk-free zero-coupon bonds in domestic and foreign currency, respectively.

However, when looking at market data, one observes that the covered interest parity is systematically violated. More precisely, the relations \eqref{eq:notirp} and \eqref{eq:notirp2} were approximately satisfied before the 2007 financial crisis. Since then, persistent violations have been observed, with the consequence that cross-currency forward values can not be reconstructed from unsecured funding rates. In particular, the relations \eqref{eq:notirp} and \eqref{eq:notirp2} must be adjusted by introducing the so-called \emph{cross-currency basis swap spread}. For cross-currency swaps against USD, for example, the market practice involves to introduce a spread over the floating rate for the non-USD leg of the contract. If ${k_0}$ corresponds to USD, then this means that for relation \eqref{eq:notirp2} to hold, we must substitute $L^{k_3}_t(T_{i-1},T_i)$ with $L^{k_3}_t(T_{i-1},T_i) + \mathcal{S}_0(T_N)$ for all $i=1, \dots, N$, where $\mathcal{S}_0(T_N)$ is the cross-currency basis swap spread which is a function of the contract's maturity $T_N$ and is set at the stipulation of the contract. Similarly, for relation \eqref{eq:notirp} to hold, we must substitute $L^{k_3}_t(t,t+\delta)$ with $L^{k_3}_t(t,t+\delta) + \mathcal{S}_0(t+\delta)$. We report in Figure \ref{fig:ccySwapSpreads} the time series for the cross-currency basis swap spreads $\mathcal{S}_0(T_N)$ for the currencies pairs EUR-USD and GBP-USD, and for several maturities $T_N$ ranging from one to thirty years.

\begin{figure}[t]
	\resizebox{1\textwidth}{!}{
		\begin{tabular}{@{}>{\centering\arraybackslash}m{0.5\textwidth}@{}>{\centering\arraybackslash}m{0.5\textwidth}@{}}
			 \textbf{\large EUR-USD} & \textbf{\large GBP-USD}\\
			\includegraphics[width=0.47\textwidth]{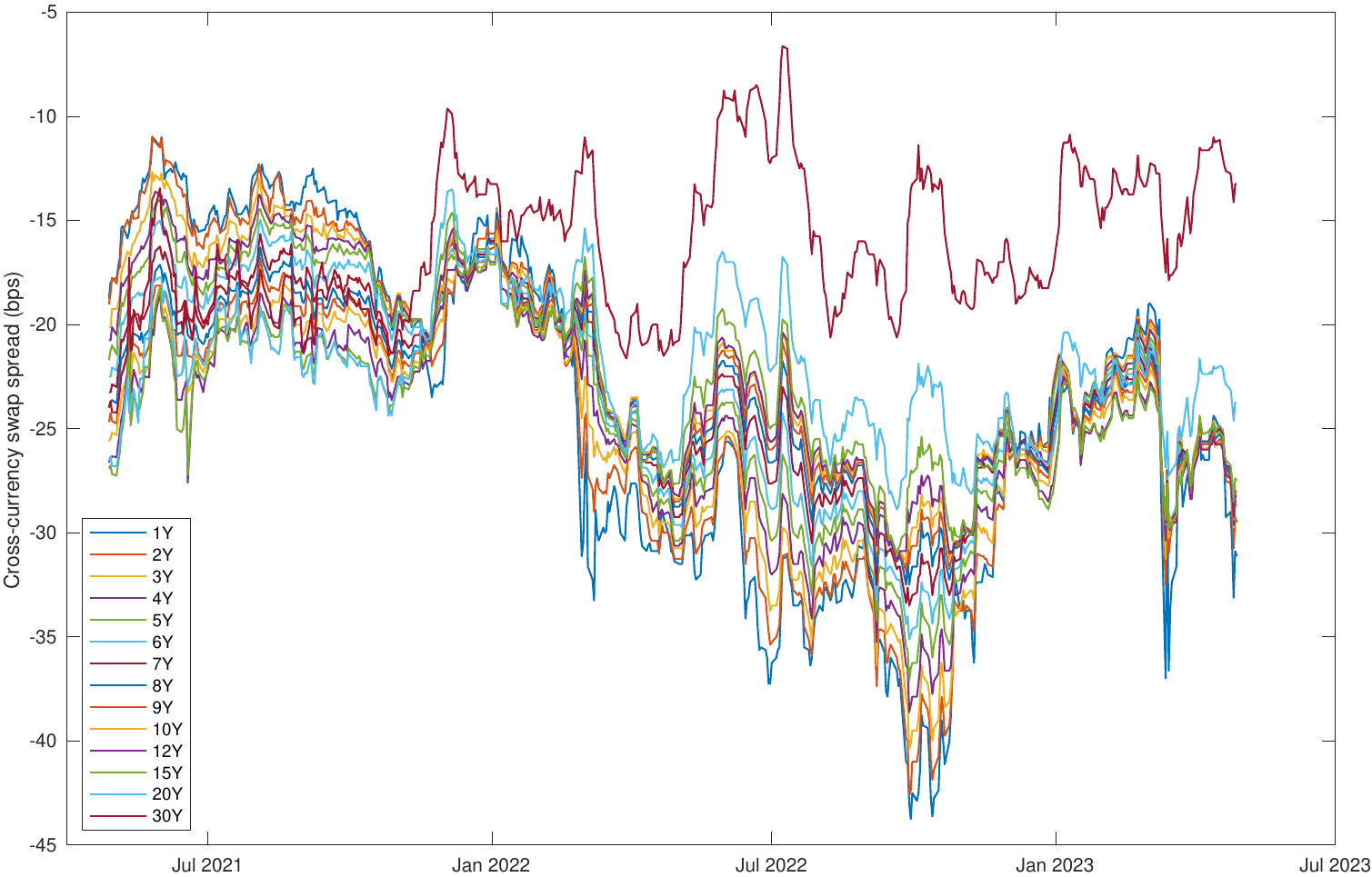} &	
			\includegraphics[width=0.47\textwidth]{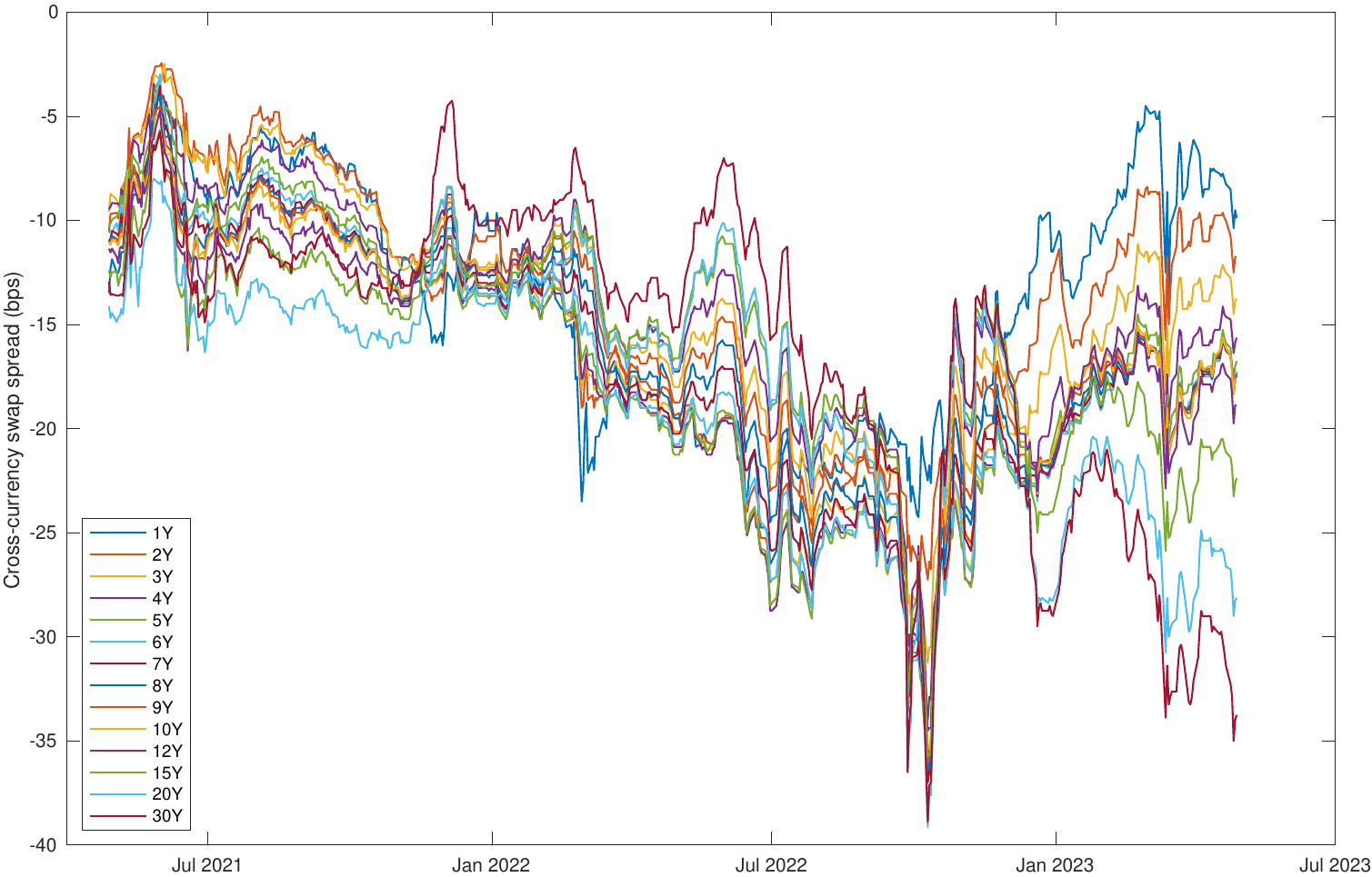}\\
		\end{tabular}
	}
    \caption{Cross-currency basis swap spread time series for the pair EUR-USD (left panel) and for the pair GBP-USD (right panel). Each curve corresponds to a different maturity.}
    \label{fig:ccySwapSpreads}
\end{figure}

To explain this phenomenon, we need to look into the nature of the contracts under consideration. In particular, the market quotes refer to perfectly collateralized instruments. In other words, the published quotes assume the existence of an ideal collateralization agreement (Credit Support Annex - CSA), in which the two agents exchange margin calls in continuous time so to perfectly annihilate  any outstanding credit exposure. On the other hand, the replication strategy involves IBOR rates. Hence, it is subject to (at least) the liquidity risk, since the lending activity is not supported by any guarantee (unsecured lending). In summary, there is a discrepancy between a perfectly collateralized derivative security and a wrongly postulated replication strategy using unsecured borrowing/lending. This highlights the importance of studying the cross-currency basis swap spread. This, however, has received limited coverage in the financial mathematics literature so far: it was analyzed in \cite{fuji10}, \cite{McCloud2013}, and \cite{mopa17}. We also mention \cite{ding2024}, that first appeared after the present work: here the authors study the pricing of cross-currency swaps and focus on contracts referencing overnight rates that we mention at the end of Section \ref{sec:example}. \textcolor{black}{In Remark \ref{rem:Ding} we relate their model with our general setting. Finally,} up to our knowledge, the only reference providing a modeling framework for this spread in an HJM setting is \cite{fujiiPhd}: our work greatly generalizes this setting. 

\subsection{IBOR-OIS spread} A similar discrepancy is observed in single-currency interest-rate markets when trying to replicate the market quotes of forward-rate agreements (FRA) with unsecured borrowing/lending on IBOR zero-coupon bonds. Let $L^{k_0}_t(T_1,T_2)$ be the (collateralized) FRA rate at time $t$ for the period $[T_1,T_2]$, where $k_0$ denotes a generic currency. If the replication was possible, we should empirically observe that
\begin{equation*}
    L^{k_0}_t(T_1,T_2)= \frac{1}{T_2-T_1}\left(\frac{B^{k_0}(t,T_1)}{B^{k_0}(t,T_2)}-1\right),
\end{equation*}
where $B^{k_0}(t,\cdot)$ denotes the unsecured (IBOR) zero-coupon bonds  for the currency $k_0$. However, this replication argument fails, as shown empirically e.g. in \cite{BiaCar2013}.  Moreover, it is also observed that the forward rate can not be reconstructed from collateralized zero-coupon bonds $B^{k_0, k_0}(t,\cdot)$ in the currency $k_0$, namely
\begin{equation*}
    L^{k_0}_t(T_1,T_2)\neq L^{k_0,k_0,D}_t(T_1,T_2):=\frac{1}{T_2-T_1}\left(\frac{B^{k_0, k_0}(t,T_1)}{B^{k_0, k_0}(t,T_2)}-1\right).
\end{equation*}
This discrepancy has led to the multiple-curve framework initiated by the seminal work by \cite{hen07}, and later studied by several authors, e.g. \cite{backwell2019} and \cite{cre12}, \cite{Cuchiero2016}, \cite{Cuchiero2019}, \cite{egg2020}, \cite{fitr12}, \cite{GM:14}, \cite{GMR:15}, \cite{GPSS14}, \cite{GR15}, \cite{hen10}, \cite{Henr14}, \cite{ken10}, \cite{kitawo09}, \cite{macrina2018}, \cite{mer10b}, \cite{mer10}, \cite{merxie12},  \cite{mopa10},  \cite{MR14}, \cite{pata10}.

\subsection{The LIBOR discontinuation is not the end of IBORs} A further element of complexity in this picture is the ongoing reform of certain interest rate benchmarks. First of all, we need a word of clarity: the term LIBOR refers to the London inter-bank offered rate, which is an unsecured inter-bank rate available for several tenors, maturities, and currencies. It has been administrated by the British Bankers Association (BBA) until 2014, and by the Inter Continent Exchange (ICE) afterwards. It is ICE that is managing its discontinuation by publishing selected tenors for the USD and GBP areas via an unrepresentative synthetic methodology. LIBOR, however, is only an example of unsecured inter-bank offered rate subject to a certain jurisdiction. There are indeed several other unsecured inter-bank rates, such as EURIBOR for the EUR area or TIBOR for the JPY area. Hence, we should not take LIBOR as a synonym for inter-bank offered rates in general, and the fact that LIBOR rates are being discontinued does not mean that unsecured inter-bank rates are being discontinued in general. In the following, we clarify the ongoing situation for the EUR, JPY and the USD area.

In the EUR area the reform of interest rate benchmarks led to the discontinuation of the unsecured EONIA (Euro Overnight Index Average) which was substituted by ESTR (Euro Short Term Rate). The calculation methodology of the EURIBOR rate has been updated using a three-step waterfall methodology\footnote{\url{https://www.emmi-benchmarks.eu/benchmarks/euribor/reforms/}}. There are no plans for a discontinuation of EURIBOR, meaning that for the EUR area, a multiple-curve model is still needed to properly describe the market of interest rate products\footnote{\url{https://www.esma.europa.eu/sites/default/files/2023-12/ESMA81-1071567537-121_EUR_RFR_WG_Final_Statement.pdf}}. Similarly, in the JPY area there are no plans to discontinue the Tokyo Inter Bank Offered Rate (TIBOR)\footnote{\url{https://www.jbatibor.or.jp/english/reform/}}.

The situation in the USD area is more involved. Here the overnight Fed Fund rate has not been discontinued. However, a second overnight rate has been introduced as the central building block of the interest rate market: this is the secured overnight financing rate (SOFR) which is a repo rate where the collateral is given by treasury bills. This means that for the USD area, there are two overnight rates, namely an unsecured one (Fed Fund) and a secured one (SOFR). In particular, it is SOFR which is now the market standard for the remuneration of collateral. This means that, for example, a proper valuation of swaps depending on the Fed Fund rate should be performed using a two-curve setting to account for the spread between the Fed Fund rate and SOFR. Notice that swaps on the Fed Fund rate used to be ``old'' OIS swaps in the terminology of \cite{Cuchiero2016}. 

With the demise of USD LIBOR, the market of interest rate swaps and interest rate options mostly moved, in terms of liquidity, to SOFR-based instruments, where the floating rate relevant for a certain coupon is constructed by compounding SOFR over the relevant time window. However, these overnight-based instruments are not suitable for the Asset Liability Management hedging needs of medium and smaller financial institutions. This led on the one hand to criticisms against SOFR, see for example \cite{cdlwy2023}, and on the other hand to the introduction of alternative inter-bank rates such as AMERIBOR T30 or AMERIBOR T90 administrated by the American Financial Exchange\footnote{The American Financial Exchange is  an electronic exchange for direct lending and borrowing for American banks.}. As previously mentioned, in the US area, interest rate option markets moved to SOFR-based instruments which have been analyzed in several papers such as \cite{andersen2020},  \cite{backwell2022}, \cite{brace2022}, \cite{fontana2023}, \cite{fontana2023b},    \cite{gellert2021}, \cite{heitfield2019},   \cite{huggins2022}, \cite{LyasMer2019}, \cite{macrina2020}, \cite{mercurio2018}, \cite{rutkowski2021}, \cite{schloegl2023},  \cite{skov2021}, \cite{turfus2020}, \cite{willems2020}.

\subsection{Summary of the requirements}
Our objective is to devise a general framework for cross-currency markets that makes it possible to jointly capture all the previously mentioned stylized facts. The paper is structured as follows. In Section \ref{sec:HJMs} we introduce the HJM framework for the multiple discount curves which we need to account for the presence of the cross-currency basis spread. In Section \ref{sec:hjmforward} we study abstract indices, allowing us to span the whole interest rate market, and, more generally, any market with quoted forwards on indices. Finally, Section \ref{sec:example} shows the relevance of the framework in the context of cross-currency swaps valuation. Moreover, the paper is complemented by an extended analysis that provides all details regarding pricing of contingent claims in a market with multiple curves and collateralization, by following and extending \cite{gnoSei2021}: in Section \ref{sec:MultiCurrTrading} we construct the cross-currency basis market  which includes general risky assets. In Section \ref{sec:pricing} we obtain valuation formulas for fully-collateralized contingent claims by extending \cite{gnoSei2021}. As a preparatory step in view of the HJM framework and as an application, we thoroughly study zero-coupon bonds (ZCBs) as basic building blocks for term-structure models. Section \ref{sec:measureChanges} presents all the measure changes that are relevant for defining our HJM framework.

\section{Cross-currency HJM framework}\label{sec:HJMs}
We start by introducing some notations. Let $T>0$ be a fixed time horizon and $\left(\Omega,\cG,\GG,\PP\right)$ be a filtered probability space  with the filtration $\GG=\left(\cG_t\right)_{t\in [0,T]}$ satisfying the usual conditions. Here $\cG_0$ is assumed to be trivial, and all the processes to be introduced in the sequel are assumed to be $\GG$-adapted right-continuous with left limits (RCLL) semimartingales. We postulate the existence of $L\in\mathbb{N}$ currency areas, and we denote with the index $k_0$, ranging from $1$ to $L$, the domestic currency. Following Section \ref{sec:pricing}, we then introduce collateralization and we use the index $k_3$, ranging from $1$ to $L$, for the currency of denomination of the collateral. In particular, we consider the cash-collateral convention, meaning that collateral is exchanged in units of the currency $k_3$ and not in units of a risky security. Moreover, the agent receives or pays interest contingent on being the poster or the receiver of the collateral. We refer the reader to Section \ref{sec:pricing} for more details. In particular, we work under the assumption of perfect or full collateralization, accordingly to equation \eqref{eq:fullcoll}. 

We further denote by $\mathcal{X}^{k_0,k}$ the price of one unit of a certain currency $k$ in terms of currency $k_0$, for every $k\ne k_0$. Following the usual FORDOM convention, we have, e.g. for EURUSD, that $\mathcal{X}^{USD,EUR}$ is the price in USD of 1 EUR. 

We introduce the following interest rates:
\begin{enumerate}
    \item $r^{k_0}=(r^{k_0}_t)_{t\geq 0}$ represents the unsecured rate of the currency area $k_0$;
    \item $r^{c,k_0}=(r^{c,k_0}_t)_{t\geq 0}$ is the rate of remuneration of cash collateral, when collateral is posted in units of the currency $k_0$.
\end{enumerate}
Next, we introduce the following spreads capturing the discrepancy between unsecured rates and collateral rates of different currency denominations.
\begin{definition}\label{def:spreads}
Let $1\leq k_0\leq L$. We define:
\begin{enumerate}
\item The \emph{liquidity spread} $q^{k_0}=(q^{k_0}_t)_{t\geq 0}$ as the difference between the unsecured funding rate $r^{k_0}$ and the collateral rate $r^{c,k_0}$, namely $q^{k_0}:=r^{k_0}-r^{c,k_0}$;
\item For any $k_3\ne k_0$, the \emph{cross-currency basis spread} $q^{k_0,k_3}=(q^{k_0, k_3}_t)_{t\geq 0}$ as the difference between the liquidity spread for the currency $k_0$ and the liquidity spread for the currency $k_3$, namely $q^{k_0,k_3}:=q^{k_0}-q^{k_3}$.
\end{enumerate}
\end{definition}
We then introduce the family of collateral cash accounts $B^{c, k_0, k_3}=(B^{c, k_0, k_3}_t)_{t\geq 0}$ with interest rate $r^{c, k_0, k_3}:= r^{c, k_0}+ q^{k_0, k_3}$, namely
\begin{equation}\label{eq:Bck0k3def}
    B^{c,k_0,k_3}_t:=\exp\left\{\int_0^t \left(r^{c,k_0}_s+q^{k_0,k_3}_s\right)ds\right\}, \qquad 1\leq k_0,k_3\leq L.
\end{equation}
Whenever $k_3=k_0$, we have by definition that $q^{k_0,k_3}\equiv 0$ $d \PP \otimes  dt$-a.s. and we write $B^{c,k_0}$ in place of $B^{c,k_0,k_0}$ for simplicity of notation.

Let $T\geq0$. Following the notation in Section \ref{sec:pricing}, we denote the price process of a domestic ZCB collateralized in domestic currency by $\left\{B^{k_0,k_0}(t,T), \ 0\leq t\leq T\right\}$, and the price process of a domestic ZCB collateralized in foreign currency by $\left\{B^{k_0,k_3}(t,T), \ 0\leq t\leq T\right\}$. Based on the preliminary work in Section \ref{sec:pricing}, the price of these ZCBs in the presence of a perfect collateral agreement (see equation \eqref{eq:fullcoll}), for any $0\leq t\leq T$, is given respectively by
\begin{align*}
B^{k_0,k_0}(t,T)=B^{c,k_0}_t\mathbb{E}^{\QQ^{k_0}}\left[\left.\frac{1}{B^{c,k_0}_T}\right|\mathcal{G}_t\right]\quad \mbox{ and } \quad B^{k_0,k_3}(t,T)=B^{c,k_0,k_3}_t\mathbb{E}^{\QQ^{k_0}}\left[\left.\frac{1}{B^{c,k_0,k_3}_T}\right|\mathcal{G}_t\right],
\end{align*}
with $\QQ^{k_0}$ being the domestic pricing measure introduced in Section \ref{sec:MultiCurrTrading}. In particular, these pricing formulas highlight dependencies of the various term structures of ZCBs on a common set of interest rates and spreads. 
We aim to construct a modeling framework that takes into account the intimate links existing among the different curves. In view of this, instead of directly modeling all the families of ZCBs, we choose to fix a \textit{reference curve} for each currency area, and to model the remaining curves by means of (either positive or real-valued) spreads with respect to the reference curve. This is intuitive from a financial point of view and in line with the modeling philosophy of \cite{Cuchiero2016}.

In the following, for ${d}_{X}\in\mathbb{N}$, let $X=(X_t)_{t\geq 0}$ be an $\RR^{{d}_{X}}$-valued It\^o semimartingale with differential characteristics $(b, c, K)$ with respect to a truncation function $\chi$. We recall the notion of local exponent (see also \cite[Definition A.6]{KK13}), which we state under the probability space $(\Omega, \mathcal{G}, \mathbb{G}, \PP)$. 
\begin{definition}
    Let $\beta = (\beta_t)_{t\ge 0}$ an $\RR^{{d}_{X}}$-valued predictable and $X$-integrable process. The \emph{local exponent} of $X$ at $\beta$ under the measure $\PP$ is a predictable real-valued process $(\Psi_t^{\PP, X}(\beta_t))_{t\ge0}$ such that $\left( \exp(\int_0^t\beta_sdX_s - \int_0^t \Psi_s^{\PP, X}(\beta_s)ds)\right)_{t\ge0}$ is a local $(\PP,\GG)$-martingale. We denote by $\mathcal{U}^{\PP, X}$ the set of processes $\beta$ such that $\Psi^{\PP, X}(\beta)$ exists. 
\end{definition}

Moreover, with the following lemma we can express the local exponent in Lévy-Kintchine form.

\begin{lemma}\label{prop:localexponent} 
    For any $\beta\in \mathcal{U}^{\PP, X}$, outside some $d\PP\otimes dt$-nullset, it holds that
    \begin{equation*}
        \Psi_t^{\PP, X}(\beta_t) = \beta_t^\top b_t + \frac{1}{2}\beta_t^\top c_t\beta_t + \int\left( e^{\beta_t^\top \xi} - 1 -\beta_t^\top \chi(\xi)\right) K_t(d\xi).
    \end{equation*}
    Moreover, the gradient $\nabla_{\beta_t} \Psi_t^{\PP, X}(\beta_t)$ of $\Psi_t^{\PP, X}(\beta_t)$ in the direction of $\beta_t$ is the $\RR^{{d}_{X}}$-valued vector given by 
    \begin{equation}\label{eq:levykintder}
    \nabla_{\beta_t} \Psi_t^{\PP, X}(\beta_t) = b_t +  c_t\beta_t + \int\left( e^{\beta_t^\top \xi} \xi- \chi(\xi)\right) K_t(d\xi).
    \end{equation}
    We shall use the shorthand $\nabla\Psi_t^{\PP, X}(\beta_t) := \nabla_{\beta_t} \Psi_t^{\PP, X}(\beta_t)$.
\end{lemma}
\begin{proof}
    For the first part, see \cite[Proposition 3.3]{Cuchiero2016}. Equation \eqref{eq:levykintder} is obtained by simple computations.
\end{proof}

\begin{remark}[NAFLVR]
    From Section \ref{sec:pricing}, we know that
    \begin{align*}
    \left(\frac{B^{k_0,k_0}(t,T)}{B^{c,k_0}_t}\right)_{t\in[0,T]} \quad \mbox{ and } \quad \left(\frac{B^{k_0,k_3}(t,T)}{B^{c,k_0,k_3}_t}\right)_{t\in[0,T]}, \quad \mbox{ for } 1\leq k_0,k_3\leq L, \ k_3\neq k_0,
    \end{align*}
    are $(\QQ^{k_0}, \mathbb{G})$-martingales. In particular, each term structure of ZCBs in the equation above
    constitutes a large financial market with uncountably many traded instruments. The relevant notion of absence of arbitrage is provided by no asymptotic free lunch with vanishing risk (NAFLVR) introduced in \cite[ Section 3]{cuchiero2014}. With respect to the setting of \cite{cuchiero2014}, we need to take into account that we are in the presence of a multitude of term structures of ZCBs and money market accounts. However, in line with \cite{gnoSei2021} it is sufficient to introduce a suitable version of the repo constraint \eqref{eq:repoConstraint}: let $T^\star\leq \infty$ and consider $\mathtt{T}$ as the family of all subsets $\mathbb{T}\subseteq [0,T^\star]$ with $|\mathbb{T}|=M$, $M\in\mathbb{N}$. Then, for every possible choice of $\mathbb{T}=(T_1,\ldots,T_M)^\top\in \mathtt{T}$, the repo constraint has the form
    \begin{align}
    \label{eq:repoConstraint2}
        \psi^{k_0,k_3}_tB^{c,k_0,k_3}_t+\sum_{m=1}^{M}\xi^{T_m,k_0,k_3}_tB^{k_0,k_3}(t,T_m)=0, \quad 1\leq k_0,k_3\leq L, 
    \end{align}
    where $\psi^{k_0,k_3}$ is the amount of units invested in $B^{c,k_0,k_3}$ and $\xi^{T_m,k_0,k_3}$ represents the investment in the ZCB $B^{k_0,k_3}(\cdot,T_m)$.
    In this paper, we identify no arbitrage with the existence of an equivalent measure $\QQ^{k_0}$ that satisfies the conditions of Proposition \ref{prop:aoabasemkt}, together with the repo constraint \eqref{eq:repoConstraint2} and such that every ZCB denominated in units of its own specific cash account is a martingale.
\end{remark}


\subsection{HJM framework for collateral discount curves}\label{sec:hjmCollCurve}
We first concentrate on the domestic bond with domestic collateral. This is the case that is considered by the whole literature on multiple-curve models. The terminology OIS bond is common for this process, since the market-observed initial term structure of such ZCBs is obtained from a bootstrap procedure applied to overnight indexed swaps (OIS). 

We consider a filtered probability space $(\Omega,\cG,\mathbb{G},\QQ^{1},\ldots, \QQ^L)$ with the filtration $\mathbb{G} = (\cG_t)_{t\geq 0}$ satisfying the usual assumptions. On the probability space we postulate the existence of multiple probability measures $\QQ^{k_0}$, for $1\leq k_0\leq L$. Following \cite{Cuchiero2016}, we introduce a term-structure model for $\{(B^{k_0,k_0}(t,T))_{t\in[0,T]},\ T\geq 0\}$ for each currency $k_0$. Let $\cP$ denote the predictable $\sigma$-algebra.

\begin{definition}\label{def:basiccond}
    We say that the triple $(f, \alpha, \sigma)$ satisfies the \emph{HJM-basic condition}, if:
    \begin{enumerate}
    \item The map $f:\RR_+ \to \RR$ is measurable with $\int_0^T|f(u)|du<\infty$, $\QQ^{k_0}$-a.s. for all $T\in\RR_+$;
    \item The map $(\omega,t,T)\mapsto \alpha_t(T)(\omega)$ is a $\cP\otimes \mathcal{B}(\RR_+)$-measurable $\RR$-valued process such that \newline $\int_0^t\int_0^T|\alpha_s(u)|duds<\infty$, $\QQ^{k_0}$-a.s. for all $t,T\in\RR_+$;
    \item The map $(\omega,t,T)\mapsto \sigma_t(T)(\omega)$ is $\cP\otimes \mathcal{B}(\RR_+)$-measurable $\RR^{{d}_{X}}$-valued process such that \newline $\int_0^T\sigma_t(u)^\top\sigma_t(u) du<\infty$, $\QQ^{k_0}$-a.s. for all $t,T\in\RR_+$, and the process $\left(\left(\int_0^T|\sigma_{t,j}(u)|^2du\right)^{\frac{1}{2}}\right)_{t\geq 0}$ is integrable with respect to the $j$-th component of the semimartingale $X$.
    \end{enumerate}
\end{definition}

\begin{definition}
\label{def:bondpricemodel}
For any $1\leq k_0\leq L$, a \emph{bond-price model} for the currency $k_0$ is a quintuple $(B^{c,k_0}, X, f^{c,k_0}_0,\alpha^{c,k_0},\sigma^{c,k_0})$ where:
\begin{enumerate}
    \item The collateral cash account $B^{c,k_0}$ is absolutely continuous with respect to the Lebesgue measure, i.e. $B^{c,k_0}_t=e^{\int_0^t r^{c,k_0}_s ds}$ with collateral short rate $r^{c,k_0}=(r^{c,k_0}_t)_{t\geq 0}$;
    \item $X$ is an $\RR^{{d}_{X}}$-valued It\^o semimartingale;
    \item The triple $(f^{c,k_0}_0, \alpha^{c,k_0}, \sigma^{c,k_0})$  satisfies the HJM-basic condition in Definition \ref{def:basiccond};
    \item For every $T\ge 0$, the instantaneous collateral forward rate $(f^{c,k_0}_t(T))_{t\in[0,T]}$ is given by
    \begin{align}\label{eq:HJMfck0}
f^{c,k_0}_t(T)=f^{c,k_0}_0(T)+\int_0^t\alpha^{c,k_0}_s(T)ds+\int_0^t\sigma^{c,k_0}_s(T)dX_s;
    \end{align}
    \item The $k_0$-collateralized $k_0$-ZCB prices $\{(B^{k_0,k_0}(t,T))_{t\in[0,T]}, \ T\geq 0\}$ satisfy 
    \begin{equation}\label{eq:expfck0}
    B^{k_0,k_0}(t,T)=e^{-\int_t^T f^{c,k_0}_t(u)du},  
    \end{equation}
    for all $t\leq T$ and $T\geq 0$. Moreover $B^{k_0,k_0}(t,t)=1$ for all $t\geq 0$.
\end{enumerate}
\end{definition}

Next, we characterize the measures $\QQ^{k_0}$ as the ``domestic" risk-neutral measure for each economy $k_0$. This is in line with the previously introduced measure $\QQ^{k_0}$. The next definition is in line with \cite[Definition 3.9]{Cuchiero2016}.

\begin{definition}\label{def:bondModelRiskNeutral} Let $1\leq k_0 \leq L$. We say that the bond price model for the $k_0$-collateralized $k_0$-ZCBs is \emph{risk-neutral} if the processes
\begin{align}\label{eq:bondModelRiskNeutral}
    \left\{\left(\frac{B^{k_0,k_0}(t,T)}{B^{c,k_0}_t}\right)_{t\in [0,T]}, \ T\geq 0\right\}
\end{align}
are $(\QQ^{k_0},\GG)$-martingales.    
\end{definition}

We stress the fact that \eqref{eq:bondModelRiskNeutral} being martingales does not mean that the cash account $B^{c,k_0}$ is the num\'eraire of the measure $\QQ^{k_0}$. In fact, it is not. As discussed in Section \ref{sec:pricing}, $B^{c,k_0}$ is the specific funding account of the $k_0$-collateralized $k_0$-ZCB. When considering multiple risky assets in the $k_0$ economy as in Section \ref{sec:MultiCurrTrading}, it becomes clear that $B^{c,k_0}$ is only the funding account for one of the several risky assets in the economy. 

We shall now characterize the martingale property for the family of processes in \eqref{eq:bondModelRiskNeutral}. Let $$\Sigma^{c,k_0}_t(T):=\int_t^T \sigma^{c,k_0}_t(u) d u.$$ We state the following result. 

\begin{proposition} \label{prop:driftcondbondmodel}
Let $1\leq k_0\leq L$ and $0\le t\le T$. The followings are equivalent:
\begin{enumerate}
    \item The bond-price model for the currency $k_0$ is risk neutral;
    \item The conditional expectation hypothesis holds, i.e.
    \begin{align*}
        \Excond{\QQ^{k_0}}{\frac{B^{c,k_0}_t}{B^{c,k_0}_T}}{\cG_t}=e^{-\int_t^Tf^{c,k_0}_t(u) du};
    \end{align*}
    \item The process $-\Sigma^{c,k_0}_t(T)\in\mathcal{U}^{\QQ^{k_0}, X}$ and the following conditions are satisfied:
    \begin{enumerate}
        \item The process
        \begin{align*}
            \left(\exp\left\{-\int_0^t\Sigma^{c,k_0}_s(T)dX_s-\int_0^t\Psi^{\QQ^{k_0}, X}_s(-\Sigma^{c,k_0}_s(T))ds\right\}\right)_{t\in[0,T]}
        \end{align*}
        is a $(\QQ^{k_0}, \GG)$-martingale;
        \item The consistency condition holds, meaning that
        \begin{align}\label{eq:conscondftt}
            \Psi_t^{\QQ^{k_0},-\int_0^\cdot r^{c,k_0}_sds}(1)=-r^{c,k_0}_{t-}=-f^{c,k_0}_{t-}(t);
        \end{align}
        \item The HJM drift condition 
        \begin{align}\label{eq:HJMcondck0}
            \int_t^T\alpha^{c,k_0}_t(u)du=\Psi^{\QQ^{k_0},X}_t(-\Sigma^{c,k_0}_t(T))
        \end{align}
        holds.
    \end{enumerate}
\end{enumerate}
\end{proposition}

\begin{proof}
    See the proof of \cite[Proposition 3.9]{Cuchiero2016}.
\end{proof}

\begin{corollary}\label{cor:driftcondk0k0}
    If the bond-price model for the currency $k_0$ is risk neutral, then:
    \begin{enumerate}
        \item For every $T>0$, the instantaneous collateral forward rate $(f^{c,k_0}_t(T))_{t\in[0,T]}$ is given by
    \begin{align}\label{eq:instFwdCollRate}
       f^{c,k_0}_t(T)=f^{c,k_0}_0(T)-\int_0^t\sigma^{c,k_0}_s(T)\nabla\Psi_s^{\QQ^{k_0}, X}(-\Sigma^{c,k_0}_s(T))ds+\int_0^t\sigma^{c,k_0}_s(T)dX_s;
    \end{align}
    \item For every $t\ge0$, the collateral short rate $r^{c,k_0}_t$ at time $t$ is given by
       \begin{align*}
       r^{c,k_0}_t=f^{c,k_0}_0(t)-\int_0^t\sigma^{c,k_0}_s(t)\nabla\Psi_s^{\QQ^{k_0}, X}(-\Sigma^{c,k_0}_s(t))ds+\int_0^t\sigma^{c,k_0}_s(t)dX_s.
    \end{align*}
    \end{enumerate}
\end{corollary}
\begin{proof}
    By taking the derivative with respect to $T>0$ on both sides of equation \eqref{eq:HJMcondck0} we get
    \begin{equation}\label{eq:alphack0}
        \alpha^{c,k_0}_t(T)=\frac{\partial \Psi^{\QQ^{k_0},X}_t(-\Sigma^{c,k_0}_t(T))}{\partial T}= -\sigma^{c,k_0}_t(T) \,\nabla\Psi_t^{\QQ^{k_0}, X}(-\Sigma^{c,k_0}_t(T)).
    \end{equation}
    By substituting equation \eqref{eq:alphack0} into \eqref{eq:HJMfck0} we get (i). By letting $t \to T$ we then obtain (ii).
\end{proof}

\begin{remark}
    We remark that if $X$ is a standard Brownian motion, then the differential characteristics of $X$ are $(0, I_{d_X}, 0)$, where $I_{d_X}$ denotes the identity matrix of dimension $d_X$. Hence $$\nabla\Psi_t^{\QQ^{k_0}, X}(-\Sigma^{c,k_0}_t(T)) = -  \Sigma^{c,k_0}_t(T),$$ and we recover the classical formulation for the instantaneous collateral forward rate, namely 
    \begin{align*}
f^{c,k_0}_t(T)=f^{c,k_0}_0(T)+\int_0^t\sigma^{c,k_0}_s(T)\left( \int_s^T \sigma^{c,k_0}_s(u) d u\right)ds+\int_0^t\sigma^{c,k_0}_s(T)dX_s,
    \end{align*}
as found, e.g, in \cite[Theorem 6.1]{fil09}.
\end{remark}

\subsection{HJM framework for cross-currency basis curves}\label{sec:HjmCcyBsis}
The previous results conveniently summarize the HJM methodology applied to the domestic-collateralized domestic-ZCBs, $\left(B^{k_0, k_0}(t, T)\right)_{t\in [0, T]}$, for every currency denomination $1\leq k_0\leq L$. The next step is to model foreign-collateralized domestic-ZCBs, $\left(B^{k_0, k_3}(t, T)\right)_{t\in [0, T]}$, for every $T\ge 0$. The starting point is \eqref{eq:priceBk0k3} according to which the processes
\begin{align*}
    \left(\frac{B^{k_0,k_3}(t,T)}{B^{c,k_0,k_3}_t}\right)_{t\in[0,T]}
\end{align*}
should be $(\QQ^{k_0},\GG)$-martingales. We formulate this requirement in the following assumption.

\begin{assumption}\label{assu:assumptionk0k3ZCBs}
    For all $1\leq k_0,k_3 \leq L$ with $k_0\neq k_3$, we assume that the processes
    \begin{align*}
        \left\{\left(\frac{B^{k_0,k_3}(t,T)}{B^{c,k_0,k_3}_t}\right)_{t\in[0,T]},\, T\geq 0\right\}
    \end{align*}
    are $(\QQ^{k_0},\GG)$-martingales.
\end{assumption}

One possible approach is to introduce a specific bond-price model for $B^{k_0, k_3}(t, T)$ in the spirit of Definition \ref{def:bondpricemodel}. This, however, would not exploit the link between $B^{k_0, k_0}(t, T)$ and $B^{k_0, k_3}(t, T)$, and would lead to some redundancy. With a similar approach to \cite{Cuchiero2016}, one can instead model the multiplicative spread between $B^{k_0, k_0}(t, T)$ and $B^{k_0, k_3}(t, T)$. 
By doing this, the previously studied domestic-collateralized domestic-ZCB price models serve as reference curves, while the remaining curves are obtained by means of spreads with respect to these reference curves. In particular, this will lead us to construct a HJM framework for the instantaneous cross-currency basis spreads $\{(q^{k_0,k_3}_t(T))_{t\in[0,T]}, \ T\geq 0\}$.

We first define the modeling quantities.
\begin{definition}
    Let $t\in[0,T]$ with $T\geq 0$, and $1\leq k_0,k_3 \leq L$ such that $k_0\neq k_3$. We define the \emph{$k_0$-$k_3$ cross-currency spread bond} via
    \begin{align}\label{eq:bondspread}
        Q^{k_0,k_3}(t,T):=\frac{B^{k_0,k_3}(t,T)}{B^{k_0,k_0}(t,T)},
    \end{align}
    and the \emph{$k_0$-$k_3$ cross-currency spread cash account} $Q^{k_0,k_3}=(Q^{k_0,k_3}_t)_{t\geq 0}$ by setting
    \begin{align*}
        Q^{k_0,k_3}_t:=e^{\int_0^t q^{k_0,k_3}_s ds},
    \end{align*}
    with $q^{k_0,k_3}$ being the $k_0$-$k_3$ cross-currency basis spreads introduced in Definition \ref{def:spreads}.
\end{definition}

The following lemma states the relevant martingale property for our purposes.

\begin{lemma}\label{lem:bayesforQ}
    Let $1\leq k_0,k_3\leq L$ with $k_0\neq k_3$ and $T\geq 0$. Assume that the bond-price model for the $k_0$-collateralized $k_0$-ZCB is risk neutral. Then Assumption \ref{assu:assumptionk0k3ZCBs} holds if and only if the processes
    \begin{align*}
        \left\{\left(\frac{Q^{k_0,k_3}(t,T)}{Q^{k_0,k_3}_t}\right)_{t\in[0,T]},\ T\geq 0\right\}
    \end{align*}
    are $(\QQ^{T,k_0,k_0},\GG)$-martingales, where $\QQ^{T,k_0,k_0}$ is the domestic-collateralized domestic $T$-forward measure introduced in Definition \ref{def:forwardmeasure}.
\end{lemma}
\begin{proof} We use Bayes's formula for conditional expectations. For some fixed $1\leq k_0,k_3\leq L$ with $k_0\neq k_3$ and $T\geq 0$, the process $\left(\frac{Q^{k_0,k_3}(t,T)}{Q^{k_0,k_3}_t}\right)_{t\in[0,T]}$ is a $(\QQ^{T,k_0,k_0},\GG)$-martingale if and only if the process
\begin{align*}
        \left(\frac{Q^{k_0,k_3}(t,T)}{Q^{k_0,k_3}_t}\frac{B^{k_0,k_0}(t,T)}{B^{c,k_0}_tB^{k_0,k_0}(0,T)}\right)_{t\in[0,T]}
\end{align*}
is a $(\QQ^{k_0},\GG)$-martingale. Since $\frac{Q^{k_0,k_3}(t,T)}{Q^{k_0,k_3}_t}\frac{B^{k_0,k_0}(t,T)}{B^{c,k_0}_t}=\frac{B^{k_0,k_3}(t,T)}{B^{c,k_0,k_3}_t}$, we then conclude thanks to Assumption \ref{assu:assumptionk0k3ZCBs}.
\end{proof}

\begin{remark}
We understand from Lemma \ref{lem:bayesforQ} that cross-currency spread models should satisfy the martingale property under the forward measure $\QQ^{T,k_0,k_0}$. Looking at the time series of cross-currency basis swaps in Figure \ref{fig:ccySwapSpreads}, we also observe that cross-currency spreads could be positive or negative and could exhibit a term structure of different shapes (either increasing or decreasing). This is different from the IBOR-OIS basis considered in \cite{Cuchiero2016}: the multiplicative spot spreads between OIS and IBORs is indeed greater than one and is increasing with respect to the tenor length. This poses us in a more flexible modelling setting.
\end{remark}

For any given currency $k_0$, we shall now put together a bond price model for the $k_0$-collateralized $k_0$-ZCBs with a family of $L-1$ models for the $k_0$-$k_3$ cross-currency bond spreads. We call such a combination an \textit{extended bond price model}.

\begin{definition}\label{def:extendedbondpricemodel}
Let $1\leq k_0\leq L$ be fixed. We call a model consisting of
\begin{enumerate}\renewcommand{\labelenumi}{\Roman{enumi}.}
    \item The $\RR^{{d}_{X}+L}$-valued It\^o semimartingale $(X,Q^{k_0,1},\ldots,Q^{k_0,k_0-1},Q^{k_0,k_0+1} \ldots,Q^{k_0,L},B^{c,k_0})$;
    \item The functions $f^{c,k_0}_0, q^{k_0,1}_0, \dots, q^{k_0,k_0-1}_0, q^{k_0,k_0+1}_0, \dots, q^{k_0,L}_0$;
    \item The processes $$\alpha^{c,k_0}, \alpha^{k_0,1}, \dots, \alpha^{k_0,k_0-1}, \alpha^{k_0,k_0+1}, \dots, \alpha^{k_0,L},$$ and $$\sigma^{c,k_0}, \sigma^{k_0,1}, \dots, \sigma^{k_0,k_0-1}, \sigma^{k_0,k_0+1}, \dots, \sigma^{k_0,L};$$
\end{enumerate}
an \emph{extended bond-price model} for the currency $k_0$, 
if for every $k_3\neq k_0$ the following conditions are satisfied:
\begin{enumerate}
    \item The quintuple $(B^{c,k_0}, X, f^{c,k_0}_0,\alpha^{c,k_0},\sigma^{c,k_0})$ is a bond-price model in the sense of Definition \ref{def:bondpricemodel}.
    \item The $k_0$-$k_3$ cross-currency spread cash account $Q^{k_0,k_3}$ is absolutely continuous with respect to the Lebesgue measure, i.e. $Q^{k_0,k_3}_t=e^{\int_0^t q^{k_0,k_3}_s ds}$ with cross-currency basis rate $q^{k_0,k_3}=(q^{k_0,k_3}_t)_{t\geq 0}$;
    \item The triple $(q^{k_0, k_3}_0, \alpha^{k_0, k_3}, \sigma^{k_0, k_3})$  satisfies the HJM-basic condition in Definition \ref{def:basiccond};
    \item For $T\ge 0$, the instantaneous cross-currency basis spread $(q^{k_0,k_3}_t(T))_{t\in[0,T]}$ is given by
    \begin{align*}
q^{k_0,k_3}_t(T)=q^{k_0,k_3}_0(T)+\int_0^t\alpha^{k_0,k_3}_s(T)ds+\int_0^t\sigma^{k_0,k_3}_s(T)dX_s;
    \end{align*}
    \item[vi)] The $k_0$-$k_3$ cross-currency spread bonds $\{(Q^{k_0,k_3}(t,T))_{t\in[0,T]}, \ T\geq 0\}$ satisfy
    \begin{equation}\label{eq:expqk0k3}
    Q^{k_0,k_3}(t,T)=e^{-\int_t^T q^{k_0,k_3}_t(u)du}
    \end{equation}
    for all $t\leq T$ and $T\geq 0$. Moreover $Q^{k_0,k_3}(t,t)=1$ for all $t\geq 0$.
\end{enumerate}
\end{definition}

\begin{remark}
    We point out that for each currency $k_0$, the dynamics of the instantaneous cross-currency basis spreads $q^{k_0,k_3}_t(T)$ are defined under the measure $\QQ^{k_0}$.
\end{remark}

The next definition naturally collects the martingale conditions that are relevant in the current setting.

\begin{definition}
\label{def:combinedmartingale}
    Let $1\leq k_0\leq L$. We say that the extended bond-price model for the  currency $k_0$ is \emph{risk neutral} if the following conditions hold:
    \begin{enumerate}
        \item The bond-price model for the $k_0$-collateralized $k_0$-ZCB $(B^{c,k_0}, X, f^{c,k_0}_0,\alpha^{c,k_0},\sigma^{c,k_0})$ is risk neutral in the sense of Definition \ref{def:bondModelRiskNeutral};
        \item For each $1\leq k_3 \leq L$ with $k_3\neq k_0$ and $T\geq 0$, the processes
        \begin{align*}
        \left\{\left(\frac{Q^{k_0,k_3}(t,T)}{Q^{k_0,k_3}_t}\right)_{t\in[0,T]},\ T\geq 0\right\}
    \end{align*}
    are $(\QQ^{T,k_0,k_0},\GG)$-martingales.
    \end{enumerate}
\end{definition}

The next result characterizes condition (ii) of Definition \ref{def:combinedmartingale}. We define $$\Sigma^{k_0,k_3}_t(T):=\int_t^T \sigma^{k_0,k_3}_t(u) d u.$$

\begin{theorem}\label{prop:driftcondspreadmodel}
    Let $1\leq k_0\leq L$ and $0\le t\le T$. For an extended bond-price model for the currency $k_0$ satisfying condition \emph{(i)} of Definition \ref{def:combinedmartingale}, the followings are equivalent:
    \begin{enumerate}
        \item The extended bond-price model satisfies condition \emph{(ii)} of Definition \ref{def:combinedmartingale};
        \item For every $k_3\neq k_0$, the conditional expectation hypothesis holds, i.e.
        \begin{equation*}
            \Excond{\QQ^{T, k_0, k_0}}{\frac{Q^{k_0,k_3}_t}{Q^{k_0,k_3}_T}}{\cG_t}=e^{-\int_t^Tq^{k_0,k_3}_t(u) du};
        \end{equation*}
        \item For every $k_3\neq k_0$, the process $-\left(\Sigma^{c, k_0}(T)+\Sigma^{k_0, k_3}(T)\right)\in\mathcal{U}^{\QQ^{k_0},X}$ and the following conditions are satisfied:
    \begin{enumerate}
        \item The process
        \begin{align}\label{eq:martexp}
            \left(\exp\left\{-\int_0^t\left(\Sigma^{c, k_0}_s(T)+\Sigma^{k_0,k_3}_s(T)\right)dX_s-\int_0^t\Psi^{\QQ^{k_0},X}_s(-\Sigma^{c, k_0}_s(T)-\Sigma^{k_0,k_3}_s(T))ds\right\}\right)_{t\in[0,T]}
        \end{align}
        is a $(\QQ^{k_0}, \mathbb{G})$-martingale;
        \item The consistency condition holds, meaning that
        \begin{align}\label{eq:consistency}
            \Psi_t^{\QQ^{k_0},-\int_0^\cdot q^{k_0,k_3}_sds}(1)=-q^{k_0,k_3}_{t-}=-q^{k_0,k_3}_{t-}(t);
        \end{align}
        \item The HJM drift condition 
        \begin{align}\label{eq:driftcondak3k0}
            \int_t^T\alpha^{k_0,k_3}_t(u)du=\Psi^{\QQ^{k_0},X}_t(-\Sigma^{c, k_0}_t(T)-\Sigma^{k_0,k_3}_t(T))- \Psi^{\QQ^{k_0},X}_t(-\Sigma^{c,k_0}_t(T))
        \end{align}
        holds for  every $k_3\neq k_0$.
    \end{enumerate}
    \end{enumerate}
\end{theorem}

\begin{proof} Let $T>0$ and $1\leq k_3\leq L$ with $k_3\neq k_0$ be fixed.
    \begin{itemize}
        \item[] (i) $\Rightarrow$ (ii) Since the process $\left(\frac{Q^{k_0,k_3}(t,T)}{Q^{k_0,k_3}_t}\right)_{t\in[0,T]}$ is a $(\QQ^{T,k_0,k_0},\GG)$-martingale, it follows that
        \begin{equation*}
            \Excond{\QQ^{T, k_0, k_0}}{\frac{1}{Q^{k_0,k_3}_T}}{\cG_t} = \frac{Q^{k_0,k_3}(t,T)}{Q^{k_0,k_3}_t},
        \end{equation*}
        which, rearranged, gives (ii).

        \item[] (i) $\Rightarrow$ (iii) Since the process $\left(\frac{Q^{k_0,k_3}(t,T)}{Q^{k_0,k_3}_t}\right)_{t\in[0,T]}$ is a $(\QQ^{T,k_0,k_0},\GG)$-martingale, we know from Lemma \ref{lem:bayesforQ} that the process
    \begin{align}\label{eq:Qk0mart}
        \left(\frac{B^{k_0,k_3}(t,T)}{B^{c,k_0,k_3}_t} = \frac{B^{k_0,k_0}(t,T) Q^{k_0,k_3}(t,T)}{B^{c,k_0,k_3}_t} = \frac{e^{-\int_t^T(f_t^{c, k_0}(u)+q_t^{k_0, k_3}(u))du}}{e^{\int_0^t(r_s^{c, k_0}+q_s^{k_0, k_3})ds}}\right)_{t\in[0,T]}
    \end{align}
    is a $(\QQ^{k_0},\GG)$-martingale. Let $R_t := -\int_t^T(f_t^{c, k_0}(u)+q_t^{k_0, k_3}(u))du - \int_0^t(r_s^{c, k_0}+q_s^{k_0, k_3})ds$. Then the martingale property of \eqref{eq:Qk0mart} is equivalent to the martingale property of $\exp\left(R\right)$, which implies that $1\in \mathcal{U}^{\QQ^{k_0},R}$ and $\Psi_t^{\QQ^{k_0},R}(1) = 0$. Due to the integrability conditions on $\alpha^{c, k_0}$ and $\sigma^{c, k_0}$ in Definition \ref{def:bondpricemodel}, and on $\alpha^{k_0, k_3}$ and $\sigma^{k_0, k_3}$ in Definition \ref{def:extendedbondpricemodel}, we can apply the classical and the stochastic Fubini theorem, which yield
        \begin{equation}\label{eq:conscondproof}
        \begin{aligned}
            &\int_t^T (f_t^{c, k_0}(u) + q_t^{k_0, k_3}(u))du \\&= \int_0^T (f_0^{c, k_0}(u)+ q_0^{k_0, k_3}(u))du +  \int_0^t\int_s^T(\alpha_s^{c, k_0}(u) + \alpha_s^{k_0, k_3}(u))duds \\&+ \int_0^t(\Sigma^{c, k_0}_s(T) + \Sigma^{k_0, k_3}_s(T)) dX_s - \int_0^t(f_u^{c, k_0}(u)+q_u^{k_0, k_3}(u))du.
        \end{aligned}
        \end{equation}
        By applying \cite[Lemma A.13]{KK13} we then obtain that
        \begin{align*}
            0 = \Psi^{\QQ^{k_0},R}_t(1) &= \Psi_t^{\QQ^{k_0},(-\int_0^\cdot (r_s^{c, k_0} + q^{k_0,k_3}_s)ds, X)}\left( \left(1, -(\Sigma_t^{c, k_0}(T) + \Sigma_t^{k_0, k_3}(T))^\top\right)^\top\right)\\&
            -\int_t^T(\alpha_t^{c, k_0}(u) + \alpha_t^{k_0, k_3}(u))du + f^{c, k_0}_{t-}(t)+ q^{k_0, k_3}_{t-}(t).
        \end{align*}
        Set now $T=t$. Since $\Sigma_t^{c, k_0}(t)=\Sigma_t^{k_0, k_3}(t)=0$, we get 
        \begin{equation}\label{eq:conscondproof2}
        \Psi_t^{\QQ^{k_0},-\int_0^\cdot r_s^{c, k_0}ds}(1)+\Psi_t^{\QQ^{k_0},-\int_0^\cdot  q^{k_0,k_3}_sds}(1) = \Psi_t^{\QQ^{k_0},-\int_0^\cdot (r_s^{c, k_0} + q^{k_0,k_3}_s)ds}(1)= -f^{c, k_0}_{t-}(t)-q^{k_0, k_3}_{t-}(t),
        \end{equation}
        hence \eqref{eq:consistency} due to the consistency condition \eqref{eq:conscondftt}. Moreover, substituting \eqref{eq:conscondproof2} into \eqref{eq:conscondproof} yields the following drift condition:
    \begin{equation*}
        -\int_t^T \left( \alpha_t^{c, k_0}(u) +  \alpha_t^{k_0, k_3}(u) \right) du = -  \Psi^{\QQ^{k_0},X}_t( -\Sigma^{c, k_0}_t(T)-\Sigma^{k_0,k_3}_t(T)),
    \end{equation*}
    hence 
    \begin{align*}
    \int_t^T\alpha^{k_0,k_3}_t(u)du&=\Psi^{\QQ^{k_0},X}_t(-\Sigma^{c, k_0}_t(T)-\Sigma^{k_0,k_3}_t(T))- \int_t^T\alpha^{c, k_0}_t(u)du\\&= \Psi^{\QQ^{k_0},X}_t(-\Sigma^{c, k_0}_t(T)-\Sigma^{k_0,k_3}_t(T))- \Psi^{\QQ^{k_0},X}_t(-\Sigma^{c,k_0}_t(T)),
    \end{align*}
    where the last equality is due to the drift condition for the bond-price model in Proposition~\ref{prop:driftcondbondmodel}. We now have both the consistency condition and the drift condition. By substituting them into \eqref{eq:conscondproof}, together with \eqref{eq:Qk0mart} we can write that 
    \begin{equation}\label{eq:conscondproof3}
    \begin{aligned}
        \frac{B^{k_0,k_3}(t,T)}{B^{c,k_0,k_3}_t} &=\exp\left\{-\int_t^T(f_t^{c, k_0}(u)+q_t^{k_0, k_3}(u))du-\int_0^t(r_s^{c, k_0}+q_s^{k_0, k_3})ds\right\}\\
        & = \exp \left\{-\int_0^T (f_0^{c, k_0}(u)+ q_0^{k_0, k_3}(u))du -\int_0^t(\Sigma^{c, k_0}_s(T) + \Sigma^{k_0, k_3}_s(T)) dX_s \right.\\&\left.\quad \quad \quad - \int_0^t\Psi^{\QQ^{k_0},X}_s(-\Sigma^{c, k_0}_s(T) - \Sigma^{k_0, k_3}_s(T)) dX_s\right\},
    \end{aligned}
    \end{equation}
    from which we deduce that the process \eqref{eq:martexp} is a $(\QQ^{k_0}, \mathbb{G})$-martingale for every $T\geq 0$.
    \item[] (iii) $\Rightarrow$ (i) The consistency condition and the drift condition yield again equation \eqref{eq:conscondproof3}. The martingale property of \eqref{eq:martexp} together with Lemma  \ref{lem:bayesforQ} implies then the last statement. 

    \end{itemize}
\end{proof}

\begin{remark}\label{rmk:changepsi}
    We emphasise that the HJM drift condition \eqref{eq:driftcondak3k0} is given in terms of the local exponent $\Psi^{\QQ^{k_0},X}$ under the measure $\QQ^{k_0}$. However, one can show that the local exponent $\Psi^{\QQ^{T, k_0, k_0},X}$ under the measure $\QQ^{T, k_0, k_0}$ is obtained from  $\Psi^{\QQ^{k_0},X}$ through the following relation:
    \begin{equation}\label{eq:relationlocalexp2}
        \Psi^{\QQ^{T, k_0, k_0},X}(\beta) = \Psi^{\QQ^{k_0},X}(\beta - \Sigma^{c, k_0}(T))- \Psi^{\QQ^{k_0},X}(- \Sigma^{c, k_0}(T)),
    \end{equation}
    for any $\beta \in \mathcal{U}^{\QQ^{k_0}, X}\cap\, \mathcal{U}^{\QQ^{T, k_0,k_0}, X}$. 
    Hence the HJM drift condition \eqref{eq:driftcondak3k0} can be rewritten as
    \begin{align*}
    \int_t^T\alpha^{k_0,k_3}_t(u)du=\Psi_t^{\QQ^{T, k_0, k_0},X}(-\Sigma^{k_0,k_3}_t(T)).
    \end{align*}
    Notice also that relation \eqref{eq:relationlocalexp2} could be used to formulate an alternative proof for Theorem \ref{prop:driftcondspreadmodel} by working under the measure $\QQ^{T, k_0, k_0}$ instead of under the measure $\QQ^{k_0}$.
\end{remark}

\begin{corollary}\label{cor:driftcondk0k3}
    If the extended bond-price model for the currency $k_0$ is risk neutral, then for every $1\leq k_3\leq L$ with $k_3\neq k_0$, we have that:
\begin{enumerate}
    \item For every $T>0$, the instantaneous cross-currency basis spread $(q^{k_0,k_3}_t(T))_{t\in[0,T]}$ is given by
    \begin{equation}\label{eq:ccbs}
    \begin{aligned}
   q^{k_0,k_3}_t(T)&=q^{k_0,k_3}_0(T)\\&-\int_0^t\left((\sigma^{c, k_0}_s(T) + \sigma^{k_0, k_3}_s(T))\nabla\Psi_s^{\QQ^{k_0}, X}(-\Sigma^{c, k_0}_s(T)-\Sigma^{k_0, k_3}_s(T))\right.\\
   &\left.-\sigma^{c,k_0}_s(T)\nabla\Psi_s^{\QQ^{k_0}, X}(-\Sigma^{c, k_0}_s(T))\right)ds+\int_0^t\sigma^{k_0,k_3}_s(T)dX_s;
    \end{aligned}
    \end{equation}
    \item For every $t\ge0$, the $k_0$-$k_3$ cross-currency basis rate $q^{k_0,k_3}_t$ at time $t$ is given by
    \begin{equation*}
    \begin{aligned}
   &q^{k_0,k_3}_t=q^{k_0,k_3}_0\\&-\int_0^t\left((\sigma^{c, k_0}_s(t) + \sigma^{k_0, k_3}_s(t))\nabla\Psi_s^{\QQ^{k_0}, X}(-\Sigma^{c, k_0}_s(t)-\Sigma^{k_0, k_3}_s(t))-\sigma^{c,k_0}_s(t)\nabla\Psi_s^{\QQ^{k_0}, X}(-\Sigma^{c, k_0}_s(t))\right)ds\\&+\int_0^t\sigma^{k_0,k_3}_s(t)dX_s.
    \end{aligned}
    \end{equation*}
\end{enumerate}
\end{corollary}
\begin{proof}
    The proof proceeds similarly to the proof of Corollary \ref{cor:driftcondk0k0}.
\end{proof}

\begin{remark}
    We remark that if $X$ is a standard Brownian motion, then the differential characteristics of $X$ are $(0, I_{d_X}, 0)$, hence $$\nabla\Psi_t^{\QQ^{k_0}, X}(-\Sigma^{c,k_0}_t(T)) = -  \Sigma^{c,k_0}_t(T),$$ and $$\nabla\Psi_t^{\QQ^{k_0}, X}(-\Sigma^{c, k_0}_t(T)-\Sigma^{k_0, k_3}_t(T)) = -\Sigma^{c, k_0}_t(T)-\Sigma^{k_0, k_3}_t(T).$$
    Then, after some simplifications, the dynamics for the instantaneous cross-currency basis spread in equation \eqref{eq:ccbs} becomes
    \begin{equation*}
    \begin{aligned}
q^{k_0,k_3}_t(T)&=q^{k_0,k_3}_0(T)\\&+\int_0^t\left(\sigma^{k_0,k_3}_s(T)\left( \int_s^T (\sigma^{c, k_0}_s(u)+\sigma^{k_0,k_3}_s(u) )d u\right)+\sigma^{c,k_0}_s(T)\left( \int_s^T \sigma^{k_0,k_3}_s(u) d u\right)\right)ds\\&+\int_0^t\sigma^{k_0,k_3}_s(T)dX_s.
    \end{aligned}
    \end{equation*}
This is in line with the results found in \cite{pit12}.
\end{remark}

\subsection{HJM framework for foreign-collateral discount curves}\label{sec:forCollDiscCurv}
We consider now the domestic bond with foreign collateral. More precisely, for a fixed currency $1 \le k_0 \le L$, we consider $L-1$ bonds $\{(B^{k_0, k_3}(t, T))_{t\in [0, T]}, T\ge 0 \}$, one for each currency $k_3 \ne k_0$. From equation \eqref{eq:bondspread}, $B^{k_0, k_3}(t, T)$ is obtained as the product between the domestic bond with domestic collateral, $B^{k_0, k_0}(t, T)$, and the $k_0$-$k_3$ cross-currency bond spread, $Q^{k_0,k_3}(t,T)$, namely 
\begin{equation}\label{eq:bondspread2}
    B^{k_0,k_3}(t,T) = B^{k_0,k_0}(t,T) Q^{k_0,k_3}(t,T).
\end{equation}
Moreover, by combining equations \eqref{eq:expfck0}, \eqref{eq:expqk0k3} and \eqref{eq:bondspread2}, we can rewrite $B^{k_0,k_3}(t,T)$ in terms of the instantaneous collateral  forward rate and of the instantaneous cross-currency basis spread, namely
\begin{equation}\label{eq:bondspread3}
    B^{k_0,k_3}(t,T) = e^{-\int_t^T(f_t^{c, k_0}(u)+q_t^{k_0, k_3}(u))du}.
\end{equation}
Remember also that, by definition, the $k_0$-$k_3$ collateral cash account at time $t$ is given by
\begin{equation}\label{eq:Bck0k3}
    B^{c, k_0,k_3}_t = e^{\int_0^t(r_s^{c, k_0}+q_s^{k_0, k_3})ds}.
\end{equation}
In other words, given an extended bond-price model as introduced in the previous section, the foreign collateral discount curve is also implicitly modelled. In particular, for every $T\ge 0$, the instantaneous foreign-collateral  forward rate  $(f_t^{c, k_0, k_3}(T))_{t\in [0, T]}$ follows the HJM dynamics    
\begin{align}
    \notag f_t^{c, k_0, k_3}(T) &= f_t^{c, k_0}(T)+q_t^{k_0, k_3}(T)\\\label{eq:hjmck0k3}&=f_0^{c, k_0}(T)+q_0^{k_0, k_3}(T)+\int_0^t(\alpha^{c, k_0}_s(T)+\alpha^{k_0, k_3}_s(T))ds+\int_0^t(\sigma^{c, k_0}_s(T)+\sigma^{k_0, k_3}_s(T))dX_s,
\end{align}
and the $k_0$-$k_3$ collateral short rate $r^{c, k_0,k_3}_t$ at time $t$ is obtained by
\begin{align}
    \notag r^{c, k_0,k_3}_t  &= r_t^{c, k_0}+q_t^{k_0, k_3}\\\label{eq:hjmck0k3r}&=r_0^{c, k_0}+q_0^{k_0, k_3}+\int_0^t(\alpha^{c, k_0}_s(t)+\alpha^{k_0, k_3}_s(t))ds+\int_0^t(\sigma^{c, k_0}_s(t)+\sigma^{k_0, k_3}_s(t))dX_s.
\end{align}

We conclude with the following results.
\begin{lemma}
    Let $1\leq k_0\leq L$ and $0\le t\le T$. For an extended bond-price model for the currency $k_0$, the followings are equivalent:
    \begin{enumerate}
        \item The extended bond-price model is risk neutral;
        \item For every $k_3\neq k_0$, the conditional expectation hypothesis holds, i.e.
        \begin{equation*}
            \Excond{\QQ^{k_0}}{\frac{B^{c, k_0,k_3}_t}{B^{c, k_0,k_3}_T}}{\cG_t}=e^{-\int_t^Tf^{c, k_0,k_3}_t(u) du};
        \end{equation*}
        \item For every $k_3\neq k_0$, the following conditions are satisfied:
    \begin{enumerate}
        \item The consistency condition holds, meaning that
        \begin{align}\label{eq:consistencyk0k3}
            \Psi_t^{\QQ^{k_0},-\int_0^\cdot r^{c, k_0, k_3}_sds}(1)=-r^{c, k_0, k_3}_{t-}=-f^{c, k_0,k_3}_{t-}(t);
        \end{align}
        \item The HJM drift condition 
        \begin{align}\label{eq:driftcondak3k0f}
     \int_t^T (\alpha_t^{c, k_0}(u) + \alpha_t^{k_0, k_3}(u))du = \Psi^{\QQ^{k_0},X}_t(-\Sigma^{c, k_0}_t(T)-\Sigma^{k_0,k_3}_t(T))
        \end{align}
        holds.
    \end{enumerate}
    \end{enumerate}
\end{lemma}
\begin{proof}
    These results are implicitly obtained in the proof of Theorem \ref{prop:driftcondspreadmodel}.
\end{proof}

\begin{corollary}    
   If the extended bond-price model for the currency $k_0$ is risk neutral, then for every $k_3\neq k_0$ we have that:
\begin{enumerate}
    \item For every $T>0$, the instantaneous foreign-collateral  forward rate $(f_t^{c, k_0, k_3}(T))_{t\in[0,T]}$ is given by
    \begin{equation*}
    \begin{aligned}
   f_t^{c, k_0, k_3}(T)&=f_0^{c, k_0}(T)+q_0^{k_0, k_3}(T)\\&-\int_0^t(\sigma^{c, k_0}_s(T) + \sigma^{k_0, k_3}_s(T))\nabla\Psi_s^{\QQ^{k_0}, X}(-\Sigma^{c, k_0}_s(T)-\Sigma^{k_0, k_3}_s(T))ds\\&+\int_0^t(\sigma^{c, k_0}_s(T)+\sigma^{k_0, k_3}_s(T))dX_s;
    \end{aligned}
    \end{equation*}
    \item For every $t\ge0$, the $k_0$-$k_3$ collateral short rate  $r^{c, k_0,k_3}_t$ at time $t$ is given by
    \begin{equation*}
    \begin{aligned}
   r^{c, k_0,k_3}_t&=r_0^{c, k_0}+q_0^{k_0, k_3}\\&-\int_0^t(\sigma^{c, k_0}_s(t) + \sigma^{k_0, k_3}_s(t))\nabla\Psi_s^{\QQ^{k_0}, X}(-\Sigma^{c, k_0}_s(t)-\Sigma^{k_0, k_3}_s(t))ds\\&+\int_0^t(\sigma^{c, k_0}_s(t)+\sigma^{k_0, k_3}_s(t))dX_s.
    \end{aligned}
    \end{equation*}
\end{enumerate}
\end{corollary}
\begin{proof}
     The two results are obtained by inserting the drift condition \eqref{eq:driftcondak3k0f} into equation \eqref{eq:hjmck0k3} and equation \eqref{eq:hjmck0k3r}, respectively.
\end{proof}

\subsection{Foreign exchange rate models and changes of measure}\label{sec:FXrate}
We considered so far HJM models for the collateral discount curves, for the cross-currency basis curves, and for the foreign collateral discount curves. All these models have been presented under a domestic currency measure $\QQ^{k_0}$, with the index $k_0$ ranging from $1$ to $L$. In this section, we derive the corresponding dynamics under a different measure $\QQ^{k}$ for $k \ne k_0$. This allows to specify the term-structure models of every economy under a single probability measure, which is essential when performing Monte Carlo simulations.

The first step is to link the different economies by means of general foreign exchange (FX) rate processes. From Proposition \ref{prop:aoabasemkt} we learned that, to guarantee absence of arbitrage in the unextended market, it is enough to assume the $(\QQ^{k_0},\GG)$-local martingale property of \eqref{eq:secondMartingale}, namely, we require that, given a fixed currency $k_0$, for any choice of 
$k\neq k_0$, the processes
\begin{align}\label{eq:martXBB}
\left(\frac{\mathcal{X}^{k_0,k}_tB^{k}_t}{B^{k_0}_t}\right)_{t\in [0,T]}
\end{align}
are $(\QQ^{k_0},\GG)$-local martingales. Moreover, in Section \ref{sec:pricing} we worked under Assumption \ref{assu:martingales}, stating that the processes \eqref{eq:martXBB} are true $(\QQ^{k_0},\GG)$-martingales. We now proceed to construct models which are consistent with this setting.

\begin{definition}
    Let $1\leq k_0\leq L$. We call a model consisting of
    \begin{enumerate}\renewcommand{\labelenumi}{\Roman{enumi}.}
        \item The $\RR^{d_X+2L-1}$-valued It\^o semimartingale $$(X,B^{c,1},\ldots,B^{c,L},Q^{k_0,1},\ldots,Q^{k_0,k_0-1},Q^{k_0,k_0+1},\ldots,Q^{k_0,L});$$
        \item The initial conditions 
        $(\mathcal{X}^{k_0,1}_0,\ldots, \mathcal{X}^{k_0,k_0-1}_0,\mathcal{X}^{k_0,k_0+1}_0,\ldots, \mathcal{X}^{k_0,L}_0)\in\RR^{L-1}_{+}$ ;  \item The processes $(\sigma^{\mathcal{X},k_0,1},\ldots,\sigma^{\mathcal{X},k_0,k_0-1},\sigma^{\mathcal{X},k_0,k_0+1},\ldots, \sigma^{\mathcal{X},k_0,L})$;
    \end{enumerate}
    an \emph{FX market model} for the currency $k_0$, 
    if for every $k\neq k_0$ the following conditions are satisfied:
    \begin{enumerate}
        \item The collateral cash accounts $B^{c,k}$ satisfies condition \emph{(i)} of Definition \ref{def:bondpricemodel};
        \item The $k_0$-$k$ cross-currency spread cash accounts $Q^{k_0,k}$ satisfies condition \emph{(ii)} of Definition \ref{def:extendedbondpricemodel};
        \item The maps $(\omega,t)\mapsto \sigma^{\mathcal{X},k_0,k}_t$ are $\cP\otimes\cB (\RR_+)$-measurable $\RR^{d_X}$-valued processes such that:
        \begin{enumerate}
            \item $\sigma^{\mathcal{X},k_0,k}_{t,j}$ is integrable with respect to the $j$-th component of the semimartingale $X$;
            \item $\sigma^{\mathcal{X},k_0,k}_{t}\in\cU^{\QQ^{k_0},X}$.
        \end{enumerate}
        \end{enumerate} 
        We then postulate the following dynamics for the FX rates:
        \begin{align}
        \label{eq:fxRateDyn}
\mathcal{X}^{k_0,k}_t=\mathcal{X}^{k_0,k}_0\frac{B^{c,k_0}_tQ^{k_0,k}_t}{B^{c,k}_t}\exp\left\{- \int_0^t \Psi_s^{\QQ^{k_0}, X}(\sigma^{\mathcal{X},k_0,k}_s)ds +\int_0^t\sigma^{\mathcal{X},k_0,k}_sdX_s\right\},
        \end{align} 
or, analogously, 
\begin{align*}
\begin{aligned}
\mathcal{X}^{k_0,k}_t
&=\mathcal{X}^{k_0,k}_0+\int_0^t\mathcal{X}^{k_0,k}_{s}(r^{c,k_0}_{s}-r^{c,k}_{s}+q^{k_0,k}_{s})ds+ \int_0^t\mathcal{X}^{k_0,k}_{s-}\sigma^{\mathcal{X},k_0,k}_{s}dX_s.
\end{aligned}
\end{align*}
\end{definition}
The dynamics of the FX rates \eqref{eq:fxRateDyn} immediately implies that the processes \eqref{eq:martXBB} are $(\QQ^{k_0},\GG)$-local martingales. Next, we must further impose the assumption that the processes \eqref{eq:fxRateDyn} are such that \eqref{eq:martXBB} are $(\QQ^{k_0},\GG)$-true martingales in line with Assumption \ref{assu:martingales}. Concretely, for general It\^o semimartingales, this can be achieved by imposing the conditions of \cite{kalshi02}, see also \cite{cgg2017}. This will allow us to introduce changes between different spot martingale measures via the processes \eqref{eq:martXBB}.

But before defining the change of measure starting from \eqref{eq:martXBB}, we introduce some assumptions on the differential characteristics of the It\^o semimartingale $X$. This is in line with \cite[Condition (B2)]{cgg2017} and leads us to a concrete version of Assumption \ref{assu:martingales} for the process \eqref{eq:martXBB} in terms of the semimartingale characteristics.
\begin{assumption}\label{ass:X}
    Let $1\le k_0\le L$ be fixed. For the $\RR^{d_X}$-valued It\^o semimartingale $X$ with $\QQ^{k_0}$-differential characteristics $(b^{\QQ^{k_0}}, c, K^{\QQ^{k_0}})$ with respect to the truncation function $\chi$, we assume that:
    \begin{enumerate}
        \item For all $t\ge 0$ and all $k \ne k_0$,
        \begin{equation*}
            \int_0^t \int_{\bigl|\bigl(\sigma^{\mathcal{X},k_0,k}_{s}\bigr)^\top\xi\bigr| > 1} \exp\left\{\bigl(\sigma^{\mathcal{X},k_0,k}_{s}\bigr)^\top\xi\right\}\bigl(\sigma^{\mathcal{X},k_0,k}_{s}\bigr)^\top\xi \,K^{\QQ^{k_0}}_s(d\xi) ds < \infty, \quad \QQ^{k_0}\mbox{-a.s.};
        \end{equation*}
        \item For all $T\ge 0$ and all $k \ne k_0$,
        \begin{equation*}
        \begin{aligned}
            \sup_{t\le T} \;&\EE^{\QQ^{k_0}}\left[\exp\left\{\frac{1}{2} \int_0^t \bigl(\sigma^{\mathcal{X},k_0,k}_{s}\bigr)^\top c_s\sigma^{\mathcal{X},k_0,k}_{s} ds \right.\right.\\& \qquad \left.\left.+ \int_0^t \int_{\RR^{d_X}} \left(\exp\left\{\bigl(\sigma^{\mathcal{X},k_0,k}_{s}\bigr)^\top\xi\right\}\left(\bigl(\sigma^{\mathcal{X},k_0,k}_{s}\bigr)^\top\xi -1\right)+1 \right)K^{\QQ^{k_0}}_s(d\xi) ds \right\} \right]<\infty, \quad \QQ^{k_0}\mbox{-a.s.}.
            \end{aligned}
        \end{equation*}
    \end{enumerate} 
\end{assumption}

Assumption \ref{ass:X} allows us to derive the following representation for the semimartingale $X$.
\begin{lemma}
Under Assumption \ref{ass:X}, the semimartingale $X$ admits the following representation  under the measure $\QQ^{k_0}$:
\begin{equation}\label{eq:reprX}
   X_t = X_0 + \int_0^t b_s^{\QQ^{k_0}} ds + \int_0^t \sqrt{c_s}dW^{\QQ^{k_0}}_s  +\int_0^t \int_{\RR^{d_X}} \xi \left(\mu^X -K_s^{\QQ^{k_0}}\right) (d\xi, ds),
\end{equation}
where $W^{\QQ^{k_0}}$ is an $\RR^{d_X}$-valued $(\QQ^{k_0}, \GG)$-Brownian motion, $\sqrt{\cdot}$ denotes the matrix-square root for symmetric positive semidefinite matrices, and $\mu^X$ is the random measure for the jump components of $X$ with compensator $K^{\QQ^{k_0}}$ under the measure $\QQ^{k_0}$.
\end{lemma}
\begin{proof}
    The canonical decomposition of $X$ under $\QQ^{k_0}$ is
\begin{align*}
   X_t =& X_0 + \int_0^t b_s^{\QQ^{k_0}} ds + \int_0^t \sqrt{c_s}dW_s^{\QQ^{k_0}}  \\&+\int_0^t \int_{\RR^{d_X}} \chi(\xi) \left(\mu^X -K_s^{\QQ^{k_0}}\right) (d\xi, ds)+\int_0^t \int_{\RR^{d_X}} \left(\xi-\chi(\xi)\right) \mu^X (d\xi, ds).
\end{align*}
Because of Assumption \ref{ass:X}, for the last term we notice that we can add and subtract
\begin{equation*}
    \int_0^t \int_{\RR^{d_X}} \left(\xi-\chi(\xi)\right) \mu^X (d\xi, ds)<\infty,
\end{equation*}
hence we obtain the decomposition \eqref{eq:reprX} with 
\begin{equation*}
    \tilde{b}_s^{\QQ^{k_0}} := b_s^{\QQ^{k_0}} + \int_{\RR^{d_X}} \left(\xi-\chi(\xi)\right)K_s^{\QQ^{k_0}} (d\xi),
\end{equation*}
which, with an abuse of notation, we denoted again by $b^{\QQ^{k_0}}$.
\end{proof}

Measure changes between different spot martingale measures are now defined via the processes \eqref{eq:martXBB}. We detail the measure transformation in the following result.
\begin{lemma}\label{lem:measureChange}
    Under Assumption \ref{ass:X}, for any $k \ne k_0$, we introduce the risk-neutral measure $\QQ^{k}\sim \QQ^{k_0}$ on $\mathcal{G}_T$ by
    \begin{equation*}
        \frac{\partial \QQ^{k}}{\partial \QQ^{k_0}}:= \frac{\mathcal{X}^{k_0,k}_TB^{k}_T}{B^{k_0}_T}\frac{B^{k_0}_0}{\mathcal{X}^{k_0,k}_0B^{k}_0},
    \end{equation*}
    such that
        \begin{equation}\label{eq:changek0k0p}
        \left.\frac{\partial \QQ^{k}}{\partial \QQ^{k_0}}\right|_{\mathcal{G}_t}= \Excond{\QQ^{k_0}}{\frac{\partial \QQ^{k}}{\partial \QQ^{k_0}}}{\mathcal{G}_t}=\frac{\mathcal{X}^{k_0,k}_tB^{k}_t}{B^{k_0}_t}\frac{B^{k_0}_0}{\mathcal{X}^{k_0,k}_0B^{k}_0}.
    \end{equation}
    Then
    \begin{align}
        &\label{eq:Wprime} W^{\QQ^{k}} := W^{\QQ^{k_0}} -\int_0^{\cdot} \bigl( \sigma^{\mathcal{X},k_0,k}_s\bigr)^{\top}\sqrt{c_s}ds,\,\mbox{ and}\\
        &\label{eq:Kprime} K^{\QQ^{k}} (d\xi) := \exp\left\{\bigl( \sigma^{\mathcal{X},k_0,k}\bigr)^{\top}\xi\right\} K^{\QQ^{k_0}} (d\xi),
    \end{align}
    represent, respectively, a $(\QQ^{k}, \GG)$-Brownian motion and the compensator of $\mu^{X}$ under $\QQ^{k}$. 
    Finally, the semimartingale $X$ admits the following representation under $\QQ^{k}$:
\begin{equation}\label{eq:reprXprime}
   X_t = X_0 + \int_0^t b_s^{\QQ^{k}} ds + \int_0^t \sqrt{c_s}dW^{\QQ^{k}}_s  +\int_0^t \int_{\RR^{d_X}} \xi \left(\mu^X -K_s^{\QQ^{k}}\right) (d\xi, ds),
\end{equation}
where
\begin{equation*}
    b_s^{\QQ^{k}}:= b_s^{\QQ^{k_0}}+\bigl( \sigma_s^{\mathcal{X},k_0,k}\bigr)^{\top}c_s- \int_{\RR^{d_X}}\xi \biggl(1- \exp\left\{\bigl( \sigma_s^{\mathcal{X},k_0,k}\bigr)^{\top}\xi\right\} \biggr)K_s^{\QQ^{k_0}} (d\xi).
\end{equation*}
\end{lemma}
\begin{proof}
    Equality \eqref{eq:changek0k0p} is easily satisfied since the processes \eqref{eq:martXBB} are $(\QQ^{k_0},\GG)$-local martingales, while the transformations \eqref{eq:Wprime} and \eqref{eq:Kprime} are due to the Girsanov transform, see, e.g., \cite{jashi03}. It is left to show \eqref{eq:reprXprime}.
    
    Combining \eqref{eq:reprX} with \eqref{eq:Wprime} and \eqref{eq:Kprime}, and by adding and subtracting $\int_0^t \int_{\RR^{d_X}} \xi K_s^{\QQ^{k}} (d\xi, ds)$,  we get 
    \begin{align*}
           X_t &= X_0 + \int_0^t b_s^{\QQ^{k_0}} ds + \int_0^t \sqrt{c_s}\left(dW^{\QQ^{k}}_s +\bigl( \sigma^{\mathcal{X},k_0,k}_s\bigr)^{\top}\sqrt{c_s}ds  \right) \\
           &+\int_0^t \int_{\RR^{d_X}} \xi \left(\mu^X -K_s^{\QQ^{k}}\right) (d\xi, ds)+\int_0^t \int_{\RR^{d_X}} \xi\left(e^{\bigl( \sigma_s^{\mathcal{X},k_0,k}\bigr)^{\top}\xi}-1\right)K_s^{\QQ^{k_0}} (d\xi, ds).
    \end{align*}
   Regrouping the terms concludes the proof.
\end{proof}

We provide an application of the previous computations. Consider two currency areas $k_0,k$. To underline the fact that the characteristics of the semimartingale $X$ are specified under a certain probability measure, we write $X^{\QQ^{k_0}}$ and $X^{\QQ^{k}}$, respectively. We consider the following hybrid model resulting from the combination of \eqref{eq:instFwdCollRate} for the two currency areas, \eqref{eq:ccbs} and \eqref{eq:fxRateDyn}:
\begin{align}
\label{eq:hybridFXIRModel}
    \begin{aligned}
     f^{c,k_0}_t(T)&=f^{c,k_0}_0(T)-\int_0^t\sigma^{c,k_0}_s(T)\nabla\Psi_s^{\QQ^{k_0}, X}(-\Sigma^{c,k_0}_s(T))ds+\int_0^t\sigma^{c,k_0}_s(T)dX_s^{\QQ^{k_0}};\\
      f^{c,k}_t(T)&=f^{c,k}_0(T)-\int_0^t\sigma^{c,k}_s(T)\nabla\Psi_s^{\QQ^{k}, X}(-\Sigma^{c,k}_s(T))ds+\int_0^t\sigma^{c,k}_s(T)dX_s^{\QQ^{k}};\\
      q^{k_0,k}_t(T)&=q^{k_0,k}_0(T)-\int_0^t\left((\sigma^{c, k_0}_s(T) + \sigma^{k_0, k}_s(T))\nabla\Psi_s^{\QQ^{k_0}, X}(-\Sigma^{c, k_0}_s(T)-\Sigma^{k_0, k}_s(T))\right.\\
      &\left.-\sigma^{c,k_0}_s(T)\nabla\Psi_s^{\QQ^{k_0}, X}(-\Sigma^{c, k_0}_s(T))\right)ds+\int_0^t\sigma^{k_0,k}_s(T)dX_s^{\QQ^{k_0}};\\
      \mathcal{X}^{k_0,k}_t&=\mathcal{X}^{k_0,k}_0\frac{B^{c,k_0}_tQ^{k_0,k}_t}{B^{c,k}_t}\exp\left\{- \int_0^t \Psi_s^{\QQ^{k_0}, X}(\sigma^{\mathcal{X},k_0,k}_s)ds +\int_0^t\sigma^{\mathcal{X},k_0,k}_sdX_s^{\QQ^{k_0}}\right\}.
    \end{aligned}
\end{align}
All the quantities are modelled under the measure $\QQ^{k_0}$, except for the $k$-instantaneous collateral  forward rate $\{(f^{c,k}_t(T))_{t\in [0,T]}, \ T\geq 0\}$. By Lemma \ref{lem:measureChange}, we can express the whole system under a unique measure $\QQ^{k_0}$ as follows:
\begin{align}
\label{eq:foreignRateDomesticMeasure}
    \begin{aligned}
        f^{c,k}_t(T)&=f^{c,k}_0(T)-\int_0^t\sigma^{c,k}_s(T)\nabla\Psi_s^{\QQ^{k}, X}(-\Sigma^{c,k}_s(T))ds\\
        &+\int_0^t\sigma^{c,k}_s(T)\left(b_s^{\QQ^{k}}ds-\bigl( \sigma_s^{\mathcal{X},k_0,k}\bigr)^{\top}c_sds+ \int_{\RR^{d_X}}\xi \bigl(1- e^{\bigl( \sigma_s^{\mathcal{X},k_0,k}\bigr)^{\top}\xi} \bigr)K_s^{\QQ^{k_0}} (d\xi)ds\right.\\
        &\left.+\sqrt{c_s}dW_s^{\QQ^{k_0}}+\int_{\RR^{d_X}}\xi\left(\mu^X -K_s^{\QQ^{k_0}}\right) (d\xi, ds)\right),
    \end{aligned}
\end{align}
where we recognize additional drift terms capturing the so-called \textit{quanto adjustments}.

\begin{remark}\label{rem:Ding}
\textcolor{black}{As an example for the generality of our framework, we mention that in \cite{ding2024} the study of cross-currency swaps is performed on the basis of the following model, presented in their Assumption 4.1. Using their notations, they assume a Hull-White type model as follows
\begin{align*}
\begin{aligned}
d r_t^d & =\left(a-b r_t^d\right) d t+\sigma_d d W_t, \\
d r_t^f & =\left(\widehat{a}-\widehat{\sigma} \bar{\sigma} \rho_{23}-\widehat{b} r_t^f\right) d t+\sigma_f d W_t, \\
d Q_t & =Q_t\left(r_t^d-r_t^f+\alpha_t^d-\alpha_t^f\right) d t+Q_t \sigma_q d W_t.
\end{aligned}
\end{align*}
Such model is a special case of our general framework. In particular it can be mapped to our semimartingale-diven dynamics \eqref{eq:hybridFXIRModel} by observing that
$f_t^{c, k_0}(t)=r_t^d$, $f_t^{c, k}(t) =r_t^f$, 
$q_t^{k_0, k}(t)=\alpha_t^d-\alpha_t^f$,
$\mathcal{X}_t^{k_0, k} =Q_t$. The presence of the quanto adjustments in the drift of the foreign short-rate dynamics indicates that they directly present the model under the domestic measure. The same object is provided in our general semimartingale-driven dynamics \eqref{eq:foreignRateDomesticMeasure}.}
\end{remark}

\section{Forwards of indices}\label{sec:hjmforward}
The aim of the present section is to analyze the financial concept of a forward contract in a general way. For this reason, we first introduce on the usual filtered probability space $(\Omega,\cG,(\cG_t)_{t\geq 0},\QQ^{1},\ldots, \QQ^L)$ an abstract index associated to a generic currency area $k_2$. Then, our definition of a forward price is that of a fixed rate that makes a forward contract fair, in the sense that the price of the contract at initiation time is zero. This principle is not new, but was first used by \cite{mer09} in order to introduce the concept of a forward rate agreement (FRA). The same principle has then been virtually employed in all the subsequent literature on multiple curves, and was also specialized in \cite{fries2016}. However, the definitions currently existing in the literature usually require to perform changes of measure, since the pricing measure is linked to the cash-account numéraire. We have seen instead that under the measure $\QQ^{k_0}$ there is a multitude of martingales, namely each risky asset discounted by means of its own asset-specific cash account. We can then define all the forwards under a unique spot pricing measure $\QQ^{k_0}$: this is a consequence of our fully coherent model of funding from the previous sections.

We denote by $I^{k_2}(T^s, T^f, T^e, T^p)$ the index referring to the period $[T^s, T^e]$ which is fixed in $T^f$ for a payment in $T^p$, with  $T^s\le T^f \le T^e$ and $T^p\ge T^f$, see Figure \ref{fig:index_structure} for an illustration of the structure of the index. This means that the index is treated as a $\cG_{T^f}$-measurable random variable, while $T^p$ fixes the time horizon for discounting. In particular, this definition includes both the case $T^p\ge T^e$ and the case $T^p \le T^e$, meaning that the payment may happen before, at, or after the end of the period, depending on the payment adjustment. In some cases, the index may be observed and paid simultaneously, namely $T^s= T^e= T^f = T^p$. Notice also that this extends the classical notation where $T^f=T^s=T$ and $T^p=T^e= T+\delta$. In this case, indeed, the index $I^{k_2}(T, T, T+\delta, T+\delta)$ corresponds to $I^{k_2}(T, T+\delta)$ in the usual notation, with the additional superscript $k_2$ for the currency denomination of the index.

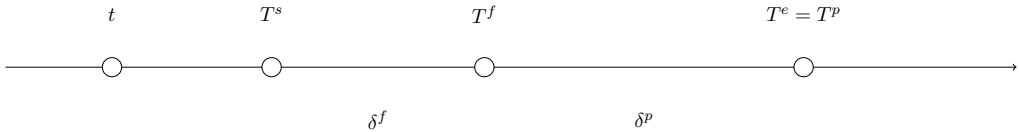
\begin{figure}[t]
    \centering
    \scalebox{0.7}{
    \begin{tikzpicture}
	\begin{pgfonlayer}{nodelayer}
		\node [style=none] (0) at (-9, 0) {};
		\node [style=none] (1) at (10, 0) {};
		\node [style=none] (8) at (0, 1) {$T^f$};
		\node [style=none] (9) at (-7, 1) {$t$};
		\node [style=emptyCircle] (10) at (-4, 0) {};
		\node [style=emptyCircle] (11) at (6, 0) {};
		\node [style=none] (12) at (-4, 1) {$T^s$};
		\node [style=none] (13) at (6, 1) {$T^e=T^p$};
		\node [style=emptyCircle] (14) at (0, 0) {};
		\node [style=none] (16) at (-2, -1) {$\delta^f$};
		\node [style=none] (17) at (3, -1) {$\delta^p$};
		\node [style=emptyCircle] (18) at (-7, 0) {};
	\end{pgfonlayer}
	\begin{pgfonlayer}{edgelayer}
		\draw [style=new edge style 0] (0.center) to (1.center);
	\end{pgfonlayer}
\end{tikzpicture}
}
    \caption{Illustration of an abstract index. The period starts at $T^s$ and ends at $T^e$. The fixing and payment times are adjusted versions of $T^s$ and $T^e$. In this case, we set $T^s < T^f<T^e$ and $T^e=T^p$.}
    \label{fig:index_structure}
\end{figure}

We then denote by $I^{k_0,k_2,k_3}(T^s,T^f,T^{e},T^{p})$ the value for a $k_0$-based agent of a forward,  written on the index $I^{k_2}(T^s, T^f, T^e,T^p)$, collateralized in currency $k_3$, and with payment date $T^p$. We define $I^{k_0,k_2,k_3}_t(T^s, T^f,T^{e},T^{p})$ as the $\cG_t$-measurable random variable that solves the following equation
\begin{align}
\label{eq:AbstractIndexFairFRA_withTe}
\Excond{\QQ^{k_0}}{\frac{B^{c,k_0,k_3}_t}{B^{c,k_0,k_3}_{T^p}}\left(I^{k_2}_{T^f}(T^s, T^f, T^e,T^p)\mathcal{X}^{k_0,k_2}_{T^p}-I^{k_0,k_2,k_3}_t(T^s, T^f,T^{e},T^{p})\right)}{\cG_t}=0,
\end{align}
where we assume that the conditional expectation is finite. We notice however that, from the mathematical point of view, the date $T^e$ does not play any role in the forward definition \eqref{eq:AbstractIndexFairFRA_withTe}. Indeed, the dates that really matter for the pricing of the forward are the start of the period $T^s$, which sets the start of the measuring of the index, the fixing date $T^f$, which sets the $\cG_{T^f}$-measurability of the index, and the payment date $T^p$, which sets the discounting. For this reason, without loss of generality, we shall drop the dependency on $T^e$ and work with $I^{k_2}(T^s, T^f, T^p)$ and $I^{k_0,k_2,k_3}(T^s, T^{f},T^{p})$ for the index and the associated forward, respectively.

In view of specifying a dynamic model for the forward $I^{k_0,k_2,k_3}(T^s, T^{f},T^{p})$, we need however to consider a running version of the index $I^{k_2}(T^s, T^f, T^p)$. For this, let us define $\delta^f:=T^f-T^s$ and $\delta^p:=T^p-T^f$ being, respectively, the fixing and the payment adjustments with respect to the fixing date. 
Then, at a given time $t\ge \delta^f$, we observe the spot index fixed in $t$, i.e. $T^f=t$, and with fixing window which has started in the past in $T^s=t-\delta^f$ for the payment horizon $T^p=t+\delta^p$. In other words, for the spot index with fixing in $t$, 
we shall adopt the notation $I^{k_2}_t(t-\delta^f,t,t+\delta^p)$ instead of $ I^{k_2}_{t}(T^s, t,T^p)$. Notice that for every time instant $t$, the spot index $I^{k_2}_t(t-\delta^f,t,t+\delta^p)$ refers to a measurement period which has started in $t-\delta^f$, hence, with this new notation, the running time must start in $\delta^f$, namely $t\ge \delta^f$, since we need enough past information to evaluate the index. 

By working under this new notation, we now formalize the definition of a forward of an index.
\begin{definition}\label{def:abstractForwardIndex}
    For a fixed $\delta^f\ge 0$ and a fixed $\delta^p\ge 0$, let $\left(I^{k_2}_t(t-\delta^f, t, t+\delta^p)\right)_{t\ge \delta^f}$ be a generic index which at every time instant $t\ge \delta^f$ has reference period starting in $t-\delta^f$ and payment date $t+\delta^p$.
    For any $T\ge \delta^f$, we denote by
    \begin{align*}
    I^{k_0,k_2,k_3}_t(T-\delta^f, T, T+\delta^p), \quad \mbox{for }\, \delta^f\le t \le T+\delta^p,
    \end{align*}
    the value at time $t$ for a $k_0$-based agent of a forward,  written on the index $\left(I^{k_2}_t(t-\delta^f, t, t+\delta^p)\right)_{t\ge \delta^f}$, collateralized in currency $k_3$, with period starting in $T-\delta^f$, fixing in $T$ and payment date $T+\delta^p$. We define $I^{k_0,k_2,k_3}_t(T-\delta^f, T, T+\delta^p)$ as the $\cG_t$-measurable random variable that solves the following equation:
    \begin{align}
        \label{eq:AbstractIndexFairFRA}
        \Excond{\QQ^{k_0}}{\frac{B^{c,k_0,k_3}_t}{B^{c,k_0,k_3}_{T+\delta^p}}\left(I^{k_2}_{T}(T-\delta^f, T, T+\delta^p)\mathcal{X}^{k_0,k_2}_{T+\delta^p}-I^{k_0,k_2,k_3}_{t}(T-\delta^f, T, T+\delta^p)\right)}{\cG_t}=0,
    \end{align}
    where we assume that the conditional expectation is finite. 
    We set $B^{c,k_0,k_3}=B^{c,k_0}$ whenever $k_3=k_0$.
\end{definition}

From the above definition of a forward, we immediately obtain the relation
\begin{align}
    \label{eq:genForwardRate}
I^{k_0,k_2,k_3}_t(T-\delta^f, T, T+\delta^p)=\frac{\Excond{\QQ^{k_0}}{\frac{B^{c,k_0,k_3}_t}{B^{c,k_0,k_3}_{T+\delta^p}}I^{k_2}_{T}(T-\delta^f, T, T+\delta^p)\mathcal{X}^{k_0,k_2}_{T+\delta^p}}{\cG_t}}{\Excond{\QQ^{k_0}}{\frac{B^{c,k_0,k_3}_t}{B^{c,k_0,k_3}_{T+\delta^p}}}{\cG_t}}.
\end{align}
It is important to notice that, for a fixed index $\left(I^{k_2}_t(t-\delta^f, t, t+\delta^p)\right)_{t\ge \delta^f}$, there are as many forward prices as there are possible funding (collateralization) policies. This is captured by the cardinality of the index $1\le k_3 \le L$. In particular, by means of the change of measure \eqref{eq:QTk0k3}, i.e.
\begin{align*}
    \frac{\partial \QQ^{T,k_0,k_3}}{\partial \QQ^{k_0}}:=\frac{B^{k_0,k_3}(T,T)}{B^{c,k_0,k_3}_T}\frac{B^{c,k_0,k_3}_0}{B^{k_0,k_3}(0,T)},
\end{align*}
from \eqref{eq:genForwardRate} we obtain that
    \begin{align}
        \label{eq:fwdForColl}I^{k_0,k_2,k_3}_t(T-\delta^f, T, T+\delta^p)=\Excond{\QQ^{T+\delta^p,k_0,k_3}}{I^{k_2}_{T}(T-\delta^f, T, T+\delta^p)\mathcal{X}^{k_0,k_2}_{T+\delta^p}}{\cG_t},
    \end{align}
    hence the forward index is a $(\QQ^{T+\delta^p,k_0,k_3},\GG)$-martingale.

\begin{remark}
    We point out that in Definition \ref{def:abstractForwardIndex}, after that the spot index $I^{k_2}_{T}(T-\delta^f, T, T+\delta^p)$ has been fixed, namely for any $T< t \le T+\delta^p$, the only varying components are the rate $\mathcal{X}^{k_0,k_2}$ and the collateral cash account $B^{c,k_0,k_3}$, which are fixed at time $T+\delta^p$, see equation \eqref{eq:AbstractIndexFairFRA}. Hence, strictly speaking, the forward $I^{k_0,k_2,k_3}_t(T-\delta^f, T, T+\delta^p)$ is a genuine forward only for $\delta^f\le t\le  T$, but it keeps varying for $T< t \le T+\delta^p$ because of the fluctuations of $\mathcal{X}^{k_0,k_2}$ and  $B^{c,k_0,k_3}$.
\end{remark}

\subsection{Some examples}
We present in this section some examples which serve as motivation for our general framework which allows for the fixing date of the index to vary within the observation interval $[T^s, T^e]$. In particular, the formalism that we introduce can be used to recover various types of interest rates linked to the SOFR rate as introduced by \cite{LyasMer2019}. In what follows, for $M\in\NN$, let $0=T_0, T_1, \ldots, T_M$ be a schedule of times, with $\delta_m := T_m - T_{m-1}$ being the year fraction for the interval $[T_{m-1}, T_m)$,  for $1\le m\le M$.

\subsubsection{Backward-looking rate}
An example of index is the backward-looking rate. In the notation of \cite{LyasMer2019},  this is denoted by $R\left(T_{m-1}, T_m\right)$. In our notation, $T^s=T_{m-1}$, $T^f = T^p = T_m$ and $t = T_m$. We then write $I^{k_0}_{T_m}(T_{m-1}, T_m, T_m)=R(T_{m-1}, T_m)$ for the backward-looking rate in currency $k_2=k_0$, which is given by
\begin{align*}
I^{k_0}_{T_m}(T_{m-1}, T_m, T_m)=\frac{1}{\delta_m}\left(e^{\int_{T_{m-1}}^{T_m} r^{c,k_0}_u du}-1\right)=\frac{1}{\delta_m}\left(\frac{B^{c,k_0}_{T_m}}{B^{c,k_0}_{T_{m-1}}}-1\right), \quad m=1, \dots, M.
\end{align*}
Clearly $I^{k_0}_{T_m}(T_{m-1}, T_m, T_m)$ is a $\cG_{T_m}$-measurable spot index.

\subsubsection{Forward-looking rate}
The forward-looking rate, in the notation of \cite{LyasMer2019} $F(T_{m-1},T_m)$, is the $T_{m-1}$-time value of the fair FRA rate $K_F$ in the swaplet with payoff  $\delta_m\left(R(T_{m-1}, T_m) - K_F\right)$. This is a spot interest rate at time $T_{m-1}$, which we can however map in our general formalism of forwards. In this case, we have $T^s=T_{m-1}$, $T^f = T^p =  T_m$ and $t = T_{m-1}$. Let then $I^{k_0}_{T_{m-1}}(T_{m-1}, T_m, T_m)=F(T_{m-1}, T_m)$ be the forward-looking rate in currency $k_2=k_0$, which, from \eqref{eq:genForwardRate},  is given by 
\begin{align*}
I^{k_0}_{T_{m-1}}(T_{m-1}, T_m, T_m)=\frac{\Excond{\QQ^{k_0}}{\frac{B^{c,k_0}_{T_{m-1}}}{B^{c,k_0}_{T_m}}I^{k_0}_{T_m}(T_{m-1}, T_m, T_m)}{\cG_{T_{m-1}}}}{\Excond{\QQ^{k_0}}{\frac{B^{c,k_0}_{T_{m-1}}}{B^{c,k_0}_{T_m}}}{\cG_{T_{m-1}}}}=\frac{1}{\delta_m}\left(\frac{1}{B^{k_0,k_0}(T_{m-1},T_m)}-1\right).
\end{align*}
This is a $\cG_{T_{m-1}}$-measurable spot index.

\subsubsection{Backward-looking in-arrears forward rate}
\label{ex:FmLyaMer}
The backward-looking in-arrears forward rate, in the notation of \cite{LyasMer2019} $R_m(t)$, is the $t$-time fixed fair value $K_R$ of a FRA with payoff $\delta_m\left(R(T_{m-1}, T_m) - K_R\right)$. In this case, we have a genuine forward rate which can be mapped in the general definition. In particular, we have $T^s=T_{m-1}$, $T^f = T_m$ and $T^p = T_m$.  We then identify this forward rate as a general forward in currency $k_2=k_0$ with collateralization in currency $k_3=k_0$, by setting $I^{k_0,k_0,k_0}_t(T_{m-1},T_m ,T_{m}) = R_m(t)$. Equation \eqref{eq:genForwardRate} takes the form
    \begin{align*}
   I^{k_0,k_0,k_0}_t(T_{m-1},T_m ,T_{m}) =\frac{\Excond{\QQ^{k_0}}{\frac{B^{c,k_0}_t}{B^{c,k_0}_{T_m}}I^{k_0}_{T_m}(T_{m-1}, T_m, T_m)}{\cG_t}}{\Excond{\QQ^{k_0}}{\frac{B^{c,k_0}_t}{B^{c,k_0}_{T_m}}}{\cG_t}}=\frac{1}{\delta_m}\left(\frac{B^{k_0,k_0}(t,T_{m-1})}{B^{k_0,k_0}(t,T_m)}-1\right).
\end{align*}
It is clear that $R_m(T_{m-1})=F(T_{m-1},T_m)$. 

\subsubsection{Forward-looking forward rate}
The forward-looking forward rate, $F_m(t)$ in the notation of \cite{LyasMer2019}, is the $t$-time fixed fair value $K_F$ of a FRA with payoff $\delta_m\left(F(T_{m-1}, T_m) - K_F\right)$. Also in this case we have a genuine forward rate which can be mapped in our general definition by setting $T^s=T_{m-1}$, $T^f = T_{m-1}$ and $T^p = T_m$. We then identify this forward rate as a general forward in currency $k_2=k_0$ with collateralization in currency $k_3=k_0$, by setting $I^{k_0,k_0,k_0}_t(T_{m-1},T_{m-1},T_m)=F_m(t)$. Equation \eqref{eq:genForwardRate} takes then the form
\begin{align*}
I^{k_0,k_0,k_0}_t(T_{m-1},T_{m-1},T_m) =\frac{\Excond{\QQ^{k_0}}{\frac{B^{c,k_0}_t}{B^{c,k_0}_{T_m}}I^{k_0}_{T_{m-1}}(T_{m-1}, T_m, T_m)}{\cG_t}}{\Excond{\QQ^{k_0}}{\frac{B^{c,k_0}_t}{B^{c,k_0}_{T_m}}}{\cG_t}}
=\frac{1}{\delta_m}\left(\frac{B^{k_0,k_0}(t,T_{m-1})}{B^{k_0,k_0}(t,T_m)}-1\right).
\end{align*}
As already observed by \cite{LyasMer2019}, for $t>T_{m-1}$, we have $F_m(t)=F\left(T_{m-1}, T_m\right)$. Moreover, for $t \leq T_{m-1}$, $R_m(t)=F_m(t)$, $d\PP\otimes dt$-a.s., and for $t=T_{m-1}$, we have $R_m\left(T_{m-1}\right)=F_m\left(T_{m-1}\right)=F\left(T_{m-1}, T_m\right)$  .

\subsubsection{Forward-looking inter-bank offered rate}
An example of forward is a contract written on a forward-looking inter-bank offered rate (IBOR), such as LIBOR, EURIBOR, TIBOR or AMERIBOR. We denote by $\mathcal{I}^{k_0}(T^s,T^p)$ the value of the IBOR index for the currency $k_0$. The IBOR index is fixed at the beginning of the period, namely $T^f = T^s$, hence for $t=T^s$, the index $\mathcal{I}_{T^s}^{k_0}(T^s,T^p)$ is $\cG_{T^s}$-measurable. For $k_2=k_0$, we then set $T^f = T^s = T_{m-1}$ and $T^p = T_m$, and equation \eqref{eq:genForwardRate} takes the form
    \begin{align*}
    \begin{aligned}
   I^{k_0,k_0,k_0}_t(T_{m-1},T_{m-1},T_m) &=\frac{\Excond{\QQ^{k_0}}{\frac{B^{c,k_0}_t}{B^{c,k_0}_{T_m}}\mathcal{I}^{k_0}_{T_{m-1}}(T_{m-1},T_m)}{\cG_t}}{\Excond{\QQ^{k_0}}{\frac{B^{c,k_0}_t}{B^{c,k_0}_{T_m}}}{\cG_t}}\\
   &=\Excond{\QQ^{T_m,k_0,k_0}}{\mathcal{I}^{k_0}_{T_{m-1}}(T_{m-1},T_m)}{\cG_t}=:\mathcal{I}^{k_0,k_0}_t(T_{m-1},T_{m-1},T_m),
   \end{aligned}
   \end{align*}
for $m=1, \dots, M$, which serves a definition for the IBOR forward rate for currency $k_0$ collateralized in units of currency $k_0$. This coincides with the definition of forward LIBOR rate originally introduced in \cite{mer09} and then employed in the literature on interest rate modeling in the multiple curve framework.

\subsubsection{Commodity forwards}
Definition \ref{def:abstractForwardIndex} covers also examples in the commodity markets. Let $I^{k_0}(T_1, T_2)$ be an index whose values depend on an underlying process observed over the interval $[T_1, T_2]$ with $T_1\le T_2$, hence in this case $T^s=T_1$ and $T^f = T^p = T_2$. The underlying process could be, e.g., the spot price of electricity or a temperature index, such as the cumulative average temperature (CAT) index. In the notation of \cite{benth2008}, we would have $I^{k_0}(T_1, T_2) = \frac{1}{T_2-T_1}\int_{T_1}^{T_2}\mathcal{S}_udu$ for the electricity index, and $I^{k_0}(T_1, T_2) = \int_{T_1}^{T_2}\mathcal{T}_udu$ for the temperature index, with $\{\mathcal{S}_u\}_{u\ge 0}$ and $\{\mathcal{T}_u\}_{u\ge 0}$, respectively, the spot price of electricity and the instantaneous temperature. For $k_2=k_3 = k_0$, the forward written on $I^{k_0}(T_1, T_2)$ is then, in the notation of \cite{benth2008}, $I_t^{k_0, k_0, k_0}(T_1, T_2, T_2)=F(t; T_1, T_2)$. Equation \eqref{eq:genForwardRate} becomes
\begin{align*}
I_t^{k_0, k_0, k_0}(T_1, T_2, T_2)=\frac{\Excond{\QQ^{k_0}}{\frac{B^{c,k_0}_{t}}{B^{c,k_0}_{T_2}}I^{k_0}_{T_2}(T_1, T_2)}{\cG_{t}}}{\Excond{\QQ^{k_0}}{\frac{B^{c,k_0}_t}{B^{c,k_0}_{T_2}}}{\cG_{t}}} = \Excond{\QQ^{T_2, k_0, k_0}}{I^{k_0}_{T_2}(T_1, T_2)}{\cG_t}.
\end{align*}
This is a $\cG_{T_2}$-measurable contract which is referred to as a forward with delivery period $[T_1, T_2]$. Here $\QQ^{T_2, k_0, k_0}$ is the pricing measure and $I_t^{k_0, k_0, k_0}(T_1, T_2, T_2)$ is a $(\QQ^{T_2, k_0, k_0}, \GG)$-martingale. This is a typical kind of contracts in the electricity markets where the underlying (electricity) must be delivered over a period of time. For those types of commodities with instantaneous delivery, we have $T_1 = T_2$ instead, and we may simply write $I^{k_0}(T_1) = I^{k_0}(T_1, T_1)$ for the index and $I_t^{k_0, k_0, k_0}(T_1) = I_t^{k_0, k_0, k_0}(T_1, T_1, T_1) = \Excond{\QQ^{T_1, k_0, k_0}}{I^{k_0}_{T_1}(T_1)}{\cG_t}$ for the forward contract.

\subsection{HJM framework for abstract indices} \label{sec:abstractIndeces}
The aim of the present section is to study an HJM-type of framework for forward contracts on the abstract index from Definition \ref{def:abstractForwardIndex}. As for discount curves, we can choose between modeling directly the forward contracts, or we can introduce appropriate spread models. We shall follow the approach of \cite{Cuchiero2016}, where the authors modelled the multiplicative spreads between LIBOR FRA rates and OIS FRA rates. In particular, we shall extend their approach in at least two directions. In fact, OIS FRA rates correspond to the forwards on the future performance of the collateral account. However, we have seen that in our setting we have multiple collateral accounts corresponding to multiple collateral currencies. Hence, on the one hand, we will allow for multiple collateral currencies. On the other hand, we will work in a general setting allowing for both forward- and backward-looking rates, namely we will extend the definition of OIS FRA rates to include also backward-looking rates.  

Before doing that, let us discuss the role of the native currency of denomination $k_2$ for the spot index $\left(I^{k_2}_t(t-\delta^f, t, t+\delta^p)\right)_{t\ge \delta^f}$ which appears in Definition \ref{def:abstractForwardIndex}. Starting from equation \eqref{eq:genForwardRate} and switching to the $\QQ^{k_2}$ spot measure as in Definition \ref{def:foreignMeasure}, we obtain that
\begin{align}
\notag 
I^{k_0,k_2,k_3}_t(T-\delta^f, T, T+\delta^p)&=\frac{\frac{\mathcal{X}^{k_0,k_2}_{t}B^{k_2}_{t}}{B^{k_0}_{t}}\Excond{\QQ^{k_2}}{\frac{B^{k_0}_{T+\delta^p}}{\mathcal{X}^{k_0,k_2}_{T+\delta^p}B^{k_2}_{T+\delta^p}}\frac{B^{c,k_0,k_3}_t}{B^{c,k_0,k_3}_{T+\delta^p}}I^{k_2}_{T}(T-\delta^f, T, T+\delta^p)\mathcal{X}^{k_0,k_2}_{T+\delta^p}}{\cG_t}}{\frac{\mathcal{X}^{k_0,k_2}_{t}B^{k_2}_{t}}{B^{k_0}_{t}}\Excond{\QQ^{k_2}}{\frac{B^{k_0}_{T+\delta^p}}{\mathcal{X}^{k_0,k_2}_{T+\delta^p}B^{k_2}_{T+\delta^p}}\frac{B^{c,k_0,k_3}_t}{B^{c,k_0,k_3}_{T+\delta^p}}}{\cG_t}}\\
&=\frac{\mathcal{X}^{k_0,k_2}_{t}\Excond{\QQ^{k_2}}{\frac{B^{c,k_2,k_3}_t}{B^{c,k_2,k_3}_{T+\delta^p}}I^{k_2}_{T}(T-\delta^f, T, T+\delta^p)}{\cG_t}}{\Excond{\QQ^{k_2}}{\frac{\mathcal{X}^{k_0,k_2}_{t}B^{c,k_2,k_3}_t}{\mathcal{X}^{k_0,k_2}_{T+\delta^p}B^{c,k_2,k_3}_{T+\delta^p}}}{\cG_t}},\label{eq:changeI}
\end{align}
where we used that $\frac{B^{c,k_0,k_3}B^{k_2}}{B^{k_0}} = B^{c,k_2,k_3}$ due to the relations
\begin{equation*}
    r^{c, k_0}+ q^{k_0, k_3} + r^{k_2} - r^{k_0} = r^{k_2}-r^{k_3}+r^{c, k_3} = r^{c, k_2}+ q^{k_2, k_3}.
\end{equation*}
Equation \eqref{eq:changeI} shows that we can conveniently postulate a model for the generic forward under the native currency of denomination of the index $k_2$ and then perform a measure change to recover the formulation of Definition \ref{def:abstractForwardIndex}. For this reason, we will limit ourselves to model $I^{k_0,k_2, k_3}(T-\delta^f, T, T+\delta^p)$ for $k_2=k_0$, which will allow us to further simplify the notation and write $I^{k_0,k_3}(T-\delta^f, T, T+\delta^p):=I^{k_0,k_0, k_3}(T-\delta^f, T, T+\delta^p)$. 

Our proposed generalization of OIS FRA rate, or, equivalently, of forward performance of the collateral rate, is the following: 
\begin{definition}\label{def:simpleFwdColl}
For any $1\leq k_0,k_3\leq L$, any $\delta^f, \delta^p\ge 0$ and $T\ge \delta^f$, we define the $k_0$-$k_3$ \emph{simple forward collateral rate} $\left(I^{k_0,k_3,D}_t(T-\delta^f, T,T+\delta^p)\right)_{\delta^f\le t\le T+\delta^p}$ related to the term structure of discount factors $\{(B^{k_0,k_3}(t,\tau))_{t\in[0,\tau]}, \ \tau\geq 0\}$ by
\begin{align*}
   I^{k_0,k_3,D}_t&(T-\delta^f, T,T+\delta^p) \\&:=
   \begin{cases}
        \frac{1}{\delta^f}\left(\frac{1}{B^{k_0,k_3}(t,T+\delta^p)}\EE^{\QQ^{k_0}}\Bigl[\Bigl.e^{-\int_t^{T-\delta^f}r_u^{c, k_0, k_3}du-\int_{T}^{T+\delta^p}r_u^{c, k_0, k_3}du
}\Bigr|\cG_{t}\Bigr]-1\right), & \delta^f \le t\le T-\delta^f,\\
       \frac{1}{\delta^f}\left(\frac{B_t^{c, k_0, k_3}}{B_{T-\delta^f}^{c, k_0, k_3}}\frac{1}{B^{k_0,k_3}(t,T+\delta^p)}\EE^{\QQ^{k_0}}\Bigl[\Bigl.e^{-\int_{T}^{T+\delta^p}r_u^{c, k_0, k_3}du
}\Bigr|\cG_{t}\Bigr]-1\right), & T-\delta^f < t \le  T,\\
       \frac{1}{\delta^f}\left(\frac{B_{T}^{c, k_0, k_3}}{B_{T-\delta^f}^{c, k_0, k_3}}-1\right), & T< t\le T+\delta^p.
   \end{cases}
\end{align*}
\end{definition}

The definition above is very general since it does not restrict the fixing time $T$ to coincide with the start of the period or with the payment date. 
In particular, this latter case is obtained by letting $\delta^p=0$. This leads to
\begin{equation}\label{eq:ID0}
   I^{k_0,k_3,D}_t(T-\delta^f, T,T) =
   \begin{cases}
        \frac{1}{\delta^f}\left(\frac{B^{k_0,k_3}(t,T-\delta^f)}{B^{k_0,k_3}(t,T)}-1\right), & \delta^f \le t\le  T-\delta^f,\\
       \frac{1}{\delta^f}\left(\frac{B_t^{c, k_0, k_3}}{B_{T-\delta^f}^{c, k_0, k_3}} \frac{1}{B^{k_0,k_3}(t,T)}-1\right), & T-\delta^f < t \le T.
   \end{cases}
\end{equation}
Notice that, in the single-currency setting, for $\delta^f \le t\le  T-\delta^f$,  this corresponds to the OIS FRA rate of \cite{Cuchiero2016} with $k_3 = k_0$, namely $I^{k_0,k_0,D}_t(T-  \delta^f,T, T) = L^D_t(T-\delta^f,T)$ in the notation of \cite{Cuchiero2016}. Moreover, in Section \ref{sec:forCollDiscCurv} we derived the HJM dynamics of $B^{k_0,k_3}(t,\cdot)$ as a consequence of the frameworks postulated for $B^{k_0,k_0}(t,\cdot)$ and $Q^{k_0,k_3}(t,\cdot)$. 
It is then immediate to link the $k_0$-$k_3$ simple forward collateral rate with the $k_0$-$k_3$ cross-currency spread bond since, by definition, we have
\begin{equation*}
   I^{k_0,k_3,D}_t(T-\delta^f, T,T) =
   \begin{cases}
        \frac{1}{\delta^f}\left(\frac{B^{k_0,k_0}(t,T-\delta^f)Q^{k_0,k_3}(t,T-\delta^f)}{B^{k_0,k_0}(t,T)Q^{k_0,k_3}(t,T)}-1\right), & \delta^f \le t\le  T-\delta^f,\\
       \frac{1}{\delta^f}\left(\frac{B_t^{c, k_0, k_0}Q_t^{c, k_0, k_3}}{B_{T-\delta^f}^{c, k_0, k_0}Q_{T-\delta^f}^{c, k_0, k_3}}\frac{1}{B^{k_0,k_0}(t,T)Q^{k_0,k_3}(t,T)}-1\right), & T-\delta^f < t \le T.
   \end{cases}
\end{equation*} 
We then see that the dynamics model for $I^{k_0,k_3,D}_t(T-\delta^f, T,T)$ (but also the one for $I^{k_0,k_3,D}_t(T-\delta^f, T,T+\delta^p)$) is fully characterized by the HJM models studied in Section \ref{sec:HJMs}.  

The next step is to introduce appropriate HJM frameworks for the spreads with respect to the forward of a generic index. 
It is obvious that for modelling the multiplicative spread between the forward $I^{k_0, k_3}$ on a generic index and the discount curve $I^{k_0, k_3, D}$, the start of the period for $I^{k_0, k_3}$ must coincide with the start of the period for $I^{k_0, k_3, D}$, say $T-\delta^f$ for some $T\ge \delta^f\ge0$. Similarly, the payment date for $I^{k_0, k_3}$ must coincide with the payment date for $I^{k_0, k_3, D}$, say $T+\delta^p$ for a certain $\delta^p\ge0$. However, we observe at this point that for any fixed start of period $T-\delta^f$ and any fixed payment date $T+\delta^p$, the two quantities $I^{k_0, k_3}(T-\delta^f, \cdot, T+\delta^p)$ and $I^{k_0, k_3, D}(T-\delta^f, \cdot, T+\delta^p)$ may have different fixing dates, both varying between $T-\delta^f$ and $T+\delta^p$. As this would lead to work with four different date indices, we shall let free the fixing date in $I^{k_0, k_3}(T-\delta^f, \cdot, T+\delta^p)$, namely we set it to $T$, and we shall set the fixing date in $I^{k_0, k_3, D}(T-\delta^f, \cdot, T+\delta^p)$ to coincide with the payment date, namely $T+\delta^p$, similarly as in \eqref{eq:ID0}. Under this simplifying assumption for the discount curve index, we generalize \cite{Cuchiero2016} by introducing the following.
\begin{definition}\label{def:multSpreadIndex}
Fix $1\leq k_0,k_3\leq L$ and $\delta^f, \delta^p\ge 0$, and let $\delta:=\delta^f+\delta^p$. For every $T\ge \delta^f$, the \emph{forward index spread} at time $\delta^f \le t\le T+\delta^p$ for the time period $[T-\delta^f,T+\delta^p]$ over the index $I^{k_0}(T-\delta^f,  T, T+\delta^p)$ collateralized according to $B^{c,k_0,k_3}$ and paid at time $T+\delta^p$ is defined by
\begin{equation}\label{def:forwardspread}
\mathcal{S}^{k_0,k_3}_t(T-\delta^f, T,T + \delta^{p}):=\frac{1+\delta I_t^{k_0,k_3}(T-\delta^f, T,T + \delta^{p})}{1+\delta I_t^{k_0,k_3, D}(T-\delta^f, T+ \delta^{p},T+ \delta^{p})}.
\end{equation}
\end{definition}
Notice that, strictly speaking, the quantities $\delta^f$ and $\delta^p$ in \eqref{def:forwardspread} refer to the fixing and payment adjustments for the numerator $I_t^{k_0,k_3}(T-\delta^f, T,T + \delta^{p})$. Indeed, the fixing and payment adjustments for the denominator $I_t^{k_0,k_3, D}(T-\delta^f, T+ \delta^{p},T+ \delta^{p})$ are $\delta = \delta^f+\delta^p$ and $0$, respectively. Hence \eqref{eq:ID0} in this case becomes
\begin{equation*}
   I_t^{k_0,k_3, D}(T-\delta^f, T+ \delta^{p},T+ \delta^{p}) =
   \begin{cases}
        \frac{1}{\delta}\left(\frac{B^{k_0,k_3}(t,T-\delta^f)}{B^{k_0,k_3}(t,T+\delta^p)}-1\right), & \delta^f \le t\le  T-\delta^f,\\
       \frac{1}{\delta}\left(\frac{B_t^{c, k_0, k_3}}{B_{T-\delta^f}^{c, k_0, k_3}} \frac{1}{B^{k_0,k_3}(t,T+\delta^p )}-1\right), & T-\delta^f < t \le T+\delta^p,
   \end{cases}
\end{equation*}
and the forward spread in \eqref{def:forwardspread} can be rewritten by
\begin{equation}\label{eq:spreadexplicit}
\begin{aligned}
   \mathcal{S}^{k_0,k_3}_t&(T-\delta^f, T,T + \delta^{p}) \\&=
   \begin{cases}
        \left(1+\delta I_t^{k_0,k_3}(T-\delta^f, T,T + \delta^{p})\right) \frac{B^{k_0,k_3}(t,T+\delta^p)}{B^{k_0,k_3}(t,T-\delta^f)}, & \delta^f \le t\le  T-\delta^f,\\
       \left(1+\delta I_t^{k_0,k_3}(T-\delta^f, T,T + \delta^{p})\right) \frac{B_{T-\delta^f}^{c, k_0, k_3}}{B_t^{c, k_0, k_3}} B^{k_0,k_3}(t,T+\delta^p ), & T-\delta^f < t \le T,\\
       \left(1+\delta I_T^{k_0,k_3}(T-\delta^f, T,T + \delta^{p})\right) \frac{B_{T-\delta^f}^{c, k_0, k_3}}{B_t^{c, k_0, k_3}} B^{k_0,k_3}(t,T+\delta^p ), & T< t \le T+\delta^p .
   \end{cases}
\end{aligned}
\end{equation}
This definition of spread, explicitly featuring the $k_0$-$k_3$ cross-currency spread bonds, highlights the role that cross-currency convexity adjustments will play in the dynamics of the generalized forward.

In particular, we observe that for $\delta^f \le t\leq T$, both numerator and denominator in the definition of spread \eqref{def:forwardspread} are random quantities. However, for $T-\delta^f<t\leq T$ we are in the monitoring period of both numerator and denominator, hence we expect the volatility to be decreasing. Finally, for $T<t\leq T+\delta^p$ the numerator is no longer random since it has been fixed in $t=T$, while we still observe fluctuations of the denominator with decreasing volatility up to time $t=T+\delta^p$, where the volatility becomes zero. From \eqref{eq:spreadexplicit}, we then deduce that for $T<t\leq T+\delta^p$ the multiplicative spread is of the form
\begin{equation*}
    \mathcal{S}^{k_0,k_3}_t(T-\delta^f, T,T + \delta^{p}) = \mathcal{C}\,\frac{B^{k_0,k_3}(t,T+\delta^p )}{B_t^{c, k_0, k_3}},
\end{equation*}
with $\mathcal{C}:= B_{T-\delta^f}^{c, k_0, k_3}\left(1+\delta I_T^{k_0,k_3}(T-\delta^f, T,T + \delta^{p})\right)\in \RR$ being a $\cG_T$-measurable random variable. Hence for $T<t\leq T+\delta^p$ the model for the multiplicative spread is given by the HJM framework for the foreign-collateral discount curves in Section \ref{sec:forCollDiscCurv}, since the only fluctuating components in the spread are $B^{k_0,k_3}(\cdot,T+\delta^p)$ and $B_\cdot^{c, k_0, k_3}$.

Notice further that for $T=t$, we have that $I^{k_0, k_3}_t(t-\delta^f, t, t+\delta^p)=I_t^{k_0}(t-\delta^f, t,t + \delta^{p})$, hence from \eqref{def:forwardspread} the spot index spread is given by
\begin{equation*}
\mathcal{S}^{k_0, k_3}_t(t-\delta^f, t,t + \delta^{p})=\frac{1+\delta I_t^{k_0}(t-\delta^f, t,t + \delta^{p})}{1+\delta I_t^{k_0, k_3, D}(t-\delta^f, t+ \delta^{p},t+ \delta^{p})}, \quad \mbox{for any } t\ge \delta^f.
\end{equation*}

\begin{remark}
    In some applications, it is desirable to guarantee that the forward index spread is larger than one. These situations arise for example when modeling a forward risky inter-bank rate. This property can be achieved by imposing restrictions on the dynamics of the spread as in \cite{Cuchiero2016} and \cite{Cuchiero2019}. Alternatively, one can define the forward index spread by first modelling
    \begin{align*}
    \begin{aligned}
        1+\delta\tilde{\mathcal{S}}^{k_0,k_3}_t(T-\delta^f, T,T+\delta^{p})&:=\frac{1+\delta I_t^{k_0,k_3}(T-\delta^f, T,T + \delta^{p})}{1+\delta I_t^{k_0,k_3, D}(T-\delta^f, T+ \delta^{p},T+ \delta^{p})},
    \end{aligned}
    \end{align*}
    so that $\mathcal{S}^{k_0,k_3}_t(T-\delta^f, T,T + \delta^{p})=1+\delta\tilde{\mathcal{S}}^{k_0,k_3}_t(T-\delta^f, T,T+\delta^{p})$ is larger than one as soon as $\tilde{\mathcal{S}}^{k_0,k_3}_t(T-\delta^f, T,T+\delta^{p})$ is positive. The drawback of this approach is that it typically leads to more complicated pricing formulas. For example, for a caplet written on the IBOR rate, this approach leads in general to the pricing formula of a two-dimensional basket option.
\end{remark}

The next lemma characterizes the martingale property of the forward index spreads. 
\begin{lemma}\label{lem:bayesforI}
    Let $1\leq k_0,k_3\leq L$ and  $\delta^f, \delta^p\ge 0$. Then the forward index spread $(\mathcal{S}^{k_0,k_3}_t(T-\delta^f, T,T + \delta^{p}))_{t\in[\delta^f,T+\delta^p]}$ is a $(\QQ^{T-\delta^f,k_0,k_3},\GG)$-martingale for every $T\ge \delta^f$.    
\end{lemma}    
\begin{proof}
    The proof generalises Lemma 3.11 of \cite{Cuchiero2016}. 
    Fix any schedule $\delta^f, \delta^p\ge 0$.  
    We know from \eqref{eq:fwdForColl} that $({I}^{k_0,k_3}_t(T-\delta^f, T,T+\delta^p))_{t\in[\delta^f,T+\delta^p]}$ is a martingale under $\QQ^{T+\delta^p,k_0,k_3}$. 
    Moreover, using Bayes' formula, the spread process $(\mathcal{S}^{k_0,k_3}_t(T-\delta^f, T,T+\delta^p))_{t\in[\delta^f,T+\delta^p]}$ is a $(\QQ^{T-\delta^f,k_0,k_3},\GG)$-martingale if and only if the process
    \begin{align*}
        \mathcal{S}^{k_0,k_3}_t(T-\delta^f, T,T+\delta^p)\left.\frac{\partial \QQ^{T-\delta^f,k_0,k_3}}{\partial \QQ^{k_0}}\right|_{\cG_{t}}
    \end{align*}
    is a $(\QQ^{k_0},\GG)$-martingale. From Definition \ref{def:forwardmeasure}, we can compactly write that 
    \begin{equation*}
        \left.\frac{\partial \QQ^{T-\delta^f,k_0,k_3}}{\partial \QQ^{k_0}}\right|_{\cG_{t}}=\frac{B^{c,k_0,k_3}_{\delta^f}}{B^{c,k_0,k_3}_{t\land(T-\delta^f)}}\frac{B^{k_0,k_3}(t\land (T-\delta^f),T-\delta^f)}{B^{k_0,k_3}(\delta^f,T-\delta^f)}, \qquad \mbox{for all }\delta^f \le t\le T+\delta^p,
    \end{equation*}
    where $t\land (T-\delta^f)$ is the minimum between $t$ and $T-\delta^f$, and we consider $\delta^f$ as the start of the running time. Similarly, from \eqref{eq:spreadexplicit} we can rewrite compactly the spread as
    \begin{equation*}
   \mathcal{S}^{k_0,k_3}_t(T-\delta^f, T,T + \delta^{p}) =
   \left(1+\delta I_{t}^{k_0,k_3}(T-\delta^f, T,T + \delta^{p})\right) \frac{B_{t\land(T-\delta^f)}^{c, k_0, k_3}}{B_t^{c, k_0, k_3}}\frac{B^{k_0,k_3}(t,T+\delta^p)}{B^{k_0,k_3}(t\land (T-\delta^f),T-\delta^f)}.
\end{equation*}
    By combining the last three equations, we then obtain
    \begin{align*}
    \left(1+\delta I_t^{k_0,k_3}(T-\delta^f, T,T+\delta^p)\right)\frac{B^{c,k_0,k_3}_{\delta^f}B^{k_0,k_3}(t,T+\delta^p)}{B^{c,k_0,k_3}_tB^{k_0,k_3}(\delta^f,T+\delta^p)}\frac{B^{k_0,k_3}(\delta^f,T+\delta^p)}{B^{k_0,k_3}(\delta^f,T-\delta^f)},
    \end{align*}
    which is indeed a $(\QQ^{k_0},\GG)$-martingale because $(I_t^{k_0,k_3}(T-\delta^f, T,T+\delta^p))_{t\in[\delta^f,T+\delta^p]}$ is a $(\QQ^{T+\delta^p,k_0,k_3},\GG)$-martingale and $\left.\frac{\partial \QQ^{T+\delta^p,k_0,k_3}}{\partial \QQ^{k_0}}\right|_{\cG_t}=\frac{B^{c,k_0,k_3}_{\delta^f}B^{k_0,k_3}(t,T+\delta^p)}{B^{c,k_0,k_3}_tB^{k_0,k_3}(\delta^f,T+\delta^p)}$, for all $t\ge \delta^f$. 
\end{proof}
We point out that Lemma \ref{lem:bayesforI} generalises Lemma 3.11 of \cite{Cuchiero2016} in the sense that in our setting the fixing date $T$ may differ from the start-of-period date, $T-\delta^f$. Moreover, since we model the spread up to the payment date, $T+\delta^p$, we obtain a more general result, stating that the multiplicative forward index spread is a martingale under the forward measure $\QQ^{T-\delta^f,k_0,k_3}$ for any $\delta^f\le t\le T+\delta^p$. Notice, however, that for $t\ge T-\delta^f$ the Radon-Nikodym derivative $\left.\frac{\partial \QQ^{T-\delta^f,k_0,k_3}}{\partial \QQ^{k_0}}\right|_{\cG_{t}}$ is a known quantity, hence the two measures $\QQ^{T-\delta^f,k_0,k_3}$ and $\QQ^{k_0}$ are equivalent for any $t\ge T-\delta^f$ up to a multiplicative factor. This means that, in practice, the multiplicative spread is a $\QQ^{k_0}$-martingale after the start of the monitoring period, $T-\delta^f$.

We proceed now to introduce the modeling framework for the forward index spreads. The following definition captures the complexity given by the fact that the forward on the index, for any given schedule $\delta^f, \delta^p\ge 0$ and any currency $k_0$, can be collateralized in any currency $k_3=1,\ldots, L$.

\begin{definition}\label{def:fwdIndSprHJM}
Let $1\leq k_0\leq L$ and $\delta^f, \delta^p\ge 0$ be fixed.  We call a model consisting of
\begin{enumerate}
    \item[I.] An extended bond-price model for the currency $k_0$
    \begin{align*}
         \Bigl(&X,Q^{k_0,1},\ldots,Q^{k_0,k_0-1},Q^{k_0,k_0+1} \ldots,Q^{k_0,L},B^{c,k_0},f^{c,k_0}_{\delta^f}, q^{k_0,1}_{\delta^f}, \dots, q^{k_0,k_0-1}_{\delta^f}, q^{k_0,k_0+1}_{\delta^f}, \dots, q^{k_0,L}_{\delta^f},\\
         &\alpha^{c,k_0}, \alpha^{k_0,1}, \dots, \alpha^{k_0,k_0-1}, \alpha^{k_0,k_0+1}, \dots, \alpha^{k_0,L},\sigma^{c,k_0},\sigma^{k_0,1}, \dots, \sigma^{k_0,k_0-1}, \sigma^{k_0,k_0+1}, \dots, \sigma^{k_0,L}\Bigr)
    \end{align*} 
    in the sense of Definition \ref{def:extendedbondpricemodel};
    \item[II.] The $\RR^L$-valued It\^o semimartingale $\left(\mathcal{S}^{k_0, 1}_t(t-\delta^f,t, t+\delta^p), \dots, \mathcal{S}^{k_0, L}_t(t-\delta^f,t, t+\delta^p)\right)_{t\ge \delta^f}$;
    \item[III.] The functions $h^{\delta^f, \delta^p,k_0,1}_{\delta^f},\ldots, h^{\delta^f, \delta^p, k_0,L}_{\delta^f}$;
    \item[IV.] The processes 
    \begin{equation*}
        \alpha^{\delta^f, \delta^p,k_0,1},\ldots, \alpha^{\delta^f, \delta^p,k_0,L},
    \end{equation*}
    and
    \begin{equation*}
        \sigma^{\delta^f, \delta^p,k_0,1},\ldots, \sigma^{\delta^f, \delta^p,k_0,L};
    \end{equation*}
\end{enumerate}
a \emph{multiplicative spread model} for the currency $k_0$, 
if  for every $1\leq k_3\leq L$ the following conditions are satisfied:
\begin{enumerate}
    \item The spot spread index $\mathcal{S}^{k_0, k_3}_t(t-\delta^f,t, t+\delta^p)$ is absolutely continuous with respect to the Lebsegue measure and satisfies $\mathcal{S}^{k_0, k_3}_t(t-\delta^f,t, t+\delta^p)=e^{-\int_{\delta^f}^t h^{\delta^f, \delta^p,k_0,k_3}_s ds}$ with multiplicative spread short rate $h^{\delta^f, \delta^p,k_0,k_3} = (h^{\delta^f, \delta^p,k_0,k_3}_t )_{t\ge \delta^f}$;
    \item The triple $(h^{\delta^f, \delta^p, k_0,k_3}_{\delta^f}, \alpha^{\delta^f, \delta^p, k_0,k_3}, \sigma^{\delta^f, \delta^p,k_0,k_3})$  satisfies the HJM-basic condition in Definition \ref{def:basiccond};
    \item For every $\tau\ge \delta^f$, the instantaneous multiplicative spread forward rate $(h^{\delta^f, \delta^p, k_0,k_3}_t(\tau))_{t\in[\delta^f,\tau]}$ is given by
    \begin{equation}\label{eq:hjmspread}
        h^{\delta^f, \delta^p, k_0,k_3}_t(\tau)=h^{\delta^f, \delta^p, k_0,k_3}_{\delta^f}(\tau)+\int_{\delta^f}^t\alpha^{\delta^f, \delta^p,k_0,k_3}_s(\tau)ds+\int_{\delta^f}^t\sigma^{\delta^f, \delta^p,k_0,k_3}_s(\tau)dX_s;
    \end{equation}
    \item For every $T\ge \delta^f$, the instantaneous multiplicative spread forward rate satisfies   $$h^{\delta^f, \delta^p, k_0,k_3}_t(\tau) = f^{c, k_0,k_3}_t(\tau), \qquad \mbox{for every } T< t\le \tau\le T+\delta^p;$$
    \item The forward index spread $(\mathcal{S}^{k_0,k_3}_t(T-\delta^f,T, T+\delta^p))_{t\in[\delta^f,T+\delta^p]}$ satisfies
    \begin{equation}\label{eq:hjmspreadS}
        \mathcal{S}^{k_0,k_3}_t(T-\delta^f,T, T+\delta^p)=e^{-\int_{\delta^f}^t h^{\delta^f, \delta^p, k_0,k_3}_s ds-\int_t^{T+\delta^p} h^{\delta^f, \delta^p, k_0,k_3}_t(u)du}.
    \end{equation}
\end{enumerate}
\end{definition}

The next definition naturally collects the martingale conditions that are relevant in the current setting.
\begin{definition}
\label{def:spreadmartingale}
    Let $1\leq k_0\leq L$  and $\delta^f, \delta^p\ge 0$. We say that the multiplicative spread model for the  currency $k_0$ is \emph{risk neutral} if the following conditions hold:
    \begin{enumerate}
        \item The extended bond-price model for the currency $k_0$
    \begin{align*}
         \Bigl(&X,Q^{k_0,1},\ldots,Q^{k_0,k_0-1},Q^{k_0,k_0+1} \ldots,Q^{k_0,L},B^{c,k_0},f^{c,k_0}_{\delta^f}, q^{k_0,1}_{\delta^f}, \dots, q^{k_0,k_0-1}_{\delta^f}, q^{k_0,k_0+1}_{\delta^f}, \dots, q^{k_0,L}_{\delta^f},\\
         &\alpha^{c,k_0}, \alpha^{k_0,1}, \dots, \alpha^{k_0,k_0-1}, \alpha^{k_0,k_0+1}, \dots, \alpha^{k_0,L},\sigma^{c,k_0},\sigma^{k_0,1}, \dots, \sigma^{k_0,k_0-1}, \sigma^{k_0,k_0+1}, \dots, \sigma^{k_0,L}\Bigr)
    \end{align*} is risk neutral in the sense of Definition \ref{def:combinedmartingale};
        \item For each $1\leq k_3 \leq L$, the forward index spreads
        \begin{align*}
        \left\{\left(\mathcal{S}^{k_0,k_3}_t(T-\delta^f, T,T + \delta^{p}\right)_{t\in[\delta^f,T+\delta^p]},\ T\geq \delta^f\right\}
    \end{align*}
    are $(\QQ^{T-\delta^f,k_0,k_3},\GG)$-martingales.
    \end{enumerate}
\end{definition}

For every $\tau\ge \delta^f$, we further define
\begin{equation*}
    \Sigma^{\delta^f, \delta^p, k_0, k_3}_t(\tau):= \int_t^{\tau} \sigma_t^{\delta^f, \delta^p, k_0, k_3}(u)du,
\end{equation*}
and state the following result characterizing condition (ii) of Definition \ref{def:spreadmartingale}.
\begin{theorem}\label{prop:driftcondspreadmodel_h}
    Let $1\leq k_0\leq L$ and $\delta^f, \delta^p\ge 0$. For a multiplicative spread model for the currency $k_0$, the followings are equivalent:
    \begin{enumerate}
        \item The multiplicative spread model satisfies condition \emph{(ii)} of Definition \ref{def:spreadmartingale};
        \item For every $T\ge \delta^f$ and every $k_3$, the conditional expectation hypothesis holds, namely
        \begin{equation*}
            \Excond{\QQ^{T-\delta^f, k_0, k_3}}{\mathcal{S}^{k_0,k_3}_{T}(T-\delta^f,T, T+\delta^p)}{\cG_t}=e^{-\int_{\delta^f}^th^{\delta^f, \delta^p, k_0,k_3}_s ds-\int_t^{T+\delta^p}h^{\delta^f, \delta^p, k_0,k_3}_t(u) du},
        \end{equation*}
        for every $t\in [\delta^f, T+\delta^p]$;
        \item For every $T\ge \delta^f$ and every $k_3$, $-\Sigma^{\delta^f, \delta^p, k_0, k_3}(T+\delta^p)\in\mathcal{U}^{\QQ^{k_0},X}$ and $$-\left(\Sigma^{\delta^f, \delta^p, k_0, k_3}(T+\delta^p)+\Sigma^{c, k_0}(T-\delta^f)+\Sigma^{k_0, k_3}(T-\delta^f)\right)\in\mathcal{U}^{\QQ^{k_0},X},$$
        and the following conditions are satisfied:
    \begin{enumerate}
        \item The process
        \begin{multline}\label{eq:martexp_h}
            \Biggl(\exp\left\{-\int_{\delta^f}^t\left(\Sigma^{\delta^f, \delta^p, k_0, k_3}_s(T+\delta^p)+\Sigma^{c, k_0}_{s\land (T-\delta^f)}(T-\delta^f)+\Sigma^{k_0,k_3}_{s\land (T-\delta^f)}(T-\delta^f)\right)dX_s\Biggr.\right.\\\left.\Biggl.-\int_{\delta^f}^t\Psi^{\QQ^{k_0},X}_s(-\Sigma^{\delta^f, \delta^p, k_0, k_3}_s(T+\delta^p)-\Sigma^{c, k_0}_{s\land (T-\delta^f)}(T-\delta^f)-\Sigma^{k_0,k_3}_{s\land (T-\delta^f)}(T-\delta^f))ds\right\}\Biggr)_{t\in[\delta^f,T+\delta^p]}
        \end{multline}
        is a $\QQ^{k_0}$-martingale;
        \item The consistency condition holds, meaning that
        \begin{align}\label{eq:consistency_h}
            \Psi_t^{\QQ^{k_0},-\int_{\delta^f}^\cdot h^{\delta^f, \delta^p,k_0,k_3}_sds}(1)=-h^{\delta^f, \delta^p,k_0,k_3}_{t-}=-h^{\delta^f, \delta^p,k_0,k_3}_{t-}(t)
        \end{align}
        for all $t\geq \delta^f$;
        \item For every $T\ge \delta^f$ and every $k_3$, the HJM drift condition 
        \begin{equation}\label{eq:driftcondadeltak3k0}
        \begin{aligned}
            &\int_t^{T+\delta^p}\alpha^{\delta^f, \delta^p, k_0, k_3}_t(u)du\\&= \Psi^{\QQ^{k_0},X}_t(-\Sigma^{\delta^f, \delta^p, k_0, k_3}_t(T+\delta^p)-\Sigma^{c, k_0}_{t\land (T-\delta^f)}(T-\delta^f)-\Sigma^{k_0,k_3}_{t\land (T-\delta^f)}(T-\delta^f))\\&- \Psi^{\QQ^{k_0},X}_t(-\Sigma^{c,k_0}_{t\land (T-\delta^f)}(T-\delta^f)-\Sigma^{k_0, k_3}_{t\land (T-\delta^f)}(T-\delta^f))
        \end{aligned}
        \end{equation}
        holds for every $t\in [\delta^f, T+\delta^p]$.
    \end{enumerate}
    \end{enumerate}
\end{theorem}
\begin{proof} 
The proof goes in parallel to the proof of Theorem \ref{prop:driftcondspreadmodel}.
In the following, let $T\ge \delta^f$ and $1\leq k_3\leq L$ be fixed.
    \begin{itemize}
        \item[] (i) $\Rightarrow$ (ii) Since the process $\left(\mathcal{S}^{k_0,k_3}_t(T-\delta^f,T,T+\delta^p)\right)_{t\in[\delta^f,T+\delta^p]}$ is a $(\QQ^{T-\delta^f,k_0,k_3},\GG)$-martingale, it follows that
        \begin{align*}
            \Excond{\QQ^{T-\delta^f, k_0, k_3}}{\mathcal{S}^{k_0,k_3}_T(T-\delta^f, T,T+\delta^p)}{\cG_t} &= \mathcal{S}^{k_0,k_3}_t(T-\delta^f,T,T+\delta^p) \\&= e^{-\int_{\delta^f}^t h^{\delta^f, \delta^p,k_0,k_3}_s ds-\int_t^{T+\delta^p}h^{\delta^f, \delta^p,k_0,k_3}_t(u)du}.
        \end{align*}
        \item[] (i) $\Rightarrow$ (iii) Since the process $\left(\mathcal{S}^{k_0,k_3}_t(T-\delta^f,T,T+\delta^p)\right)_{t\in[\delta^f,T+\delta^p]}$ is a $(\QQ^{T-\delta^f,k_0,k_3},\GG)$-martingale, by Bayes' formula, the process
    \begin{equation}
    \begin{aligned}
 \label{eq:Qk0mart_h}
        &\left(\mathcal{S}^{k_0,k_3}_t(T-\delta^f,T,T+\delta^p)\frac{B^{c,k_0,k_3}_{\delta^f}}{B^{c,k_0,k_3}_{t\land(T-\delta^f)}}\frac{B^{k_0,k_3}(t\land (T-\delta^f),T-\delta^f)}{B^{k_0,k_3}(\delta^f,T-\delta^f)}\right)_{t\in[\delta^f,T+\delta^p]}\\
        &= \left(\frac{ e^{-\int_{\delta^f}^t h^{\delta^f, \delta^p,k_0,k_3}_s ds-\int_t^{T+\delta^p}h^{\delta^f, \delta^p,k_0,k_3}_t(u)du -\int_{t\land (T-\delta^f)}^{T-\delta^f}f_{t\land (T-\delta^f)}^{c, k_0, k_3}(u)du }}{e^{\int_{\delta^f}^{t\land (T-\delta^f)}r_s^{c, k_0, k_3} ds-\int_{\delta^f}^{T-\delta^f}f_{\delta^f}^{c, k_0, k_3}(u) du }}\right)_{t\in[\delta^f,T+\delta^p]}
    \end{aligned}
    \end{equation}
    is a $(\QQ^{k_0},\GG)$-martingale.  
    Let \begin{align*}
        R_t :&= - \int_{\delta^f}^th^{\delta^f, \delta^p,k_0,k_3}_sds-\int_t^{T+\delta^p}h^{\delta^f, \delta^p,k_0,k_3}_t(u)du -\int_{t\land (T-\delta^f)}^{T-\delta^f}f_{t\land (T-\delta^f)}^{c, k_0, k_3}(u)du  \\& -\int_{\delta^f}^{t\land (T-\delta^f)}r_s^{c, k_0, k_3}ds +\int_{\delta^f}^{T-\delta^f}f_{\delta^f}^{c, k_0, k_3}(u) du .
    \end{align*}
    Then the martingale property of \eqref{eq:Qk0mart_h} is equivalent to the martingale property of $\exp\left(R\right)$, which implies that $1\in \mathcal{U}^{\QQ^{k_0},R}$ and $\Psi_t^{\QQ^{k_0},R}(1) = 0$. Due to the integrability conditions on $\alpha^{c, k_0}$ and $\sigma^{c, k_0}$ in Definition \ref{def:bondpricemodel}, on $\alpha^{k_0, k_3}$ and $\sigma^{k_0, k_3}$ in Definition \ref{def:extendedbondpricemodel}, and $\alpha^{\delta^f, \delta^p, k_0, k_3}$ and $\sigma^{\delta^f, \delta^p, k_0, k_3}$ in Definition \ref {def:fwdIndSprHJM}, we can apply the classical and the stochastic Fubini theorem, which yield
        \begin{equation}\label{eq:conscondproof_h}
        \begin{aligned}
            \int_t^{T+\delta^p} h^{\delta^f, \delta^p,k_0,k_3}_t(u)du &= \int_{\delta^f}^{T+\delta^p} h^{\delta^f, \delta^p,k_0,k_3}_{\delta^f}(u)du+  \int_{\delta^f}^t\int_s^{T+\delta^p}\alpha^{\delta^f, \delta^p, k_0, k_3}_s(u)duds \\&+ \int_{\delta^f}^t\Sigma^{\delta^f, \delta^p, k_0, k_3}_s(T+\delta^p) dX_s - \int_{\delta^f}^th^{\delta^f, \delta^p,k_0,k_3}_u(u)du,
        \end{aligned}
        \end{equation}
        and, similarly, starting from \eqref{eq:hjmck0k3},
        \begin{equation}\label{eq:conscondproof_f}
        \begin{aligned}
            &\int_{t\land (T-\delta^f)}^{T-\delta^f} f_{t\land (T-\delta^f)}^{c, k_0, k_3}(u) du \\&= \int_{\delta^f}^{T-\delta^f} f_{\delta^f}^{c, k_0, k_3}(u)du+  \int_{\delta^f}^{t\land (T-\delta^f)}\int_s^{T-\delta^f}(\alpha_s^{c, k_0}(u)+\alpha_s^{k_0, k_3}(u)) duds \\&+ \int_{\delta^f}^{t\land (T-\delta^f)}(\Sigma^{c, k_0}_s(T-\delta^f) + \Sigma^{k_0, k_3}_s(T-\delta^f)) dX_s - \int_{\delta^f}^{t\land (T-\delta^f)}f_u^{c, k_0, k_3}(u)du.
        \end{aligned}
        \end{equation}
        By applying \cite[Lemma A.13]{KK13} and using the equality $$\Psi_{t}^{\QQ^{k_0},-\int_{\delta^f}^{\cdot\land(T-\delta^f)} r_s^{c, k_0, k_3}ds}\left(1\right)=\Psi_{t\land(T-\delta^f)}^{\QQ^{k_0},-\int_{\delta^f}^{\cdot\land(T-\delta^f)} r_s^{c, k_0, k_3}ds}\left(1\right),$$ we then obtain that
        \begin{equation}\label{eq:conscondproof2_h}
        \begin{aligned}
            0 &= \Psi^{\QQ^{k_0},R}_t(1) \\&= \Psi_t^{\QQ^{k_0},-\int_{\delta^f}^\cdot h^{\delta^f, \delta^p,k_0,k_3}_sds}\left(1\right)+\Psi_{t\land(T-\delta^f)}^{\QQ^{k_0},-\int_{\delta^f}^{\cdot\land(T-\delta^f)} r_s^{c, k_0, k_3}ds}\left(1\right) \\&+ \Psi_t^{\QQ^{k_0}, X}\left( -\Sigma^{\delta^f, \delta^p, k_0, k_3}_t(T+\delta^p)-\Sigma_{t\land(T-\delta^f)}^{c, k_0}(T-\delta^f) - \Sigma_{t\land(T-\delta^f)}^{k_0, k_3}(T-\delta^f)\right)\\&
            -\int_t^{T+\delta^p}\alpha^{\delta^f, \delta^p, k_0, k_3}_t(u)du -\int_{t\land(T-\delta^f)}^{T-\delta^f}(\alpha_{t\land(T-\delta^f)}^{c, k_0}(u)+\alpha_{t\land(T-\delta^f)}^{k_0, k_3}(u))du \\&+ h^{\delta^f, \delta^p,k_0,k_3}_{t-}(t)+f^{c, k_0, k_3}_{(t\land(T-\delta^f))-}(t\land(T-\delta^f)),
        \end{aligned}
        \end{equation}
        where $\Psi_{t\land(T-\delta^f)}^{\QQ^{k_0},-\int_{\delta^f}^{\cdot\land(T-\delta^f)} r_s^{c, k_0, k_3}ds}\left(1\right) + f^{c, k_0, k_3}_{(t\land(T-\delta^f))-}(t\land(T-\delta^f))=0$ because of the consistency condition \eqref{eq:consistencyk0k3}. 
        Set now $t=T+\delta^p$ in \eqref{eq:conscondproof2_h}. Since $\Sigma^{\delta^f, \delta^p, k_0, k_3}_{T+\delta^p}(T+\delta^p)=\Sigma_{T-\delta^f}^{c, k_0}(T-\delta^f)=\Sigma_{T-\delta^f}^{k_0, k_3}(T-\delta^f)=0$, we get 
        \begin{equation*}
            0 = \Psi_t^{\QQ^{k_0},-\int_{\delta^f}^\cdot h^{\delta^f, \delta^p,k_0,k_3}_sds}\left(1\right)+ h^{\delta^f, \delta^p,k_0,k_3}_{t-}(t),
        \end{equation*}
        hence \eqref{eq:consistency_h}. Moreover, substituting the two consistency conditions into \eqref{eq:conscondproof2_h} yields the following drift condition:
    \begin{align*}
        &\int_t^{T+\delta^p}\alpha^{\delta^f, \delta^p, k_0, k_3}_t(u)du +\int_{t\land(T-\delta^f)}^{T-\delta^f}(\alpha_{t\land(T-\delta^f)}^{c, k_0}(u)+\alpha_{t\land(T-\delta^f)}^{k_0, k_3}(u))du \\&= \Psi_t^{\QQ^{k_0}, X}\left( -\Sigma^{\delta^f, \delta^p, k_0, k_3}_t(T+\delta^p)-\Sigma_{t\land(T-\delta^f)}^{c, k_0}(T-\delta^f) - \Sigma_{t\land(T-\delta^f)}^{k_0, k_3}(T-\delta^f)\right),
    \end{align*}
    hence, by the drift condition \eqref{eq:driftcondak3k0f},
    \begin{align*}
    &\int_t^{T+\delta^p}\alpha^{\delta^f, \delta^p, k_0, k_3}_t(u)du\\&=\Psi_t^{\QQ^{k_0}, X}\left( -\Sigma^{\delta^f, \delta^p, k_0, k_3}_t(T+\delta^p)-\Sigma_{t\land(T-\delta^f)}^{c, k_0}(T-\delta^f) - \Sigma_{t\land(T-\delta^f)}^{k_0, k_3}(T-\delta^f)\right)\\&- \Psi^{\QQ^{k_0},X}_{t\land(T-\delta^f)}(-\Sigma^{c,k_0}_{t\land(T-\delta^f)}(T-\delta^f)-\Sigma^{k_0, k_3}_{t\land(T-\delta^f)}(T-\delta^f)).
    \end{align*}
    We now have both the consistency condition and the drift condition. By substituting them into \eqref{eq:conscondproof_h}, and then together with \eqref{eq:conscondproof_f} into \eqref{eq:Qk0mart_h} we write that 
    \begin{equation}\label{eq:conscondproof3_h}
    \begin{aligned}
        &\mathcal{S}^{k_0,k_3}_t(T-\delta^f,T,T+\delta^p)\frac{B^{c,k_0,k_3}_{\delta^f}}{B^{c,k_0,k_3}_{t\land(T-\delta^f)}}\frac{B^{k_0,k_3}(t\land (T-\delta^f),T-\delta^f)}{B^{k_0,k_3}(\delta^f,T-\delta^f)} \\
        & = \exp \Biggl\{-\int_{\delta^f}^{T+\delta^p}h^{\delta^f, \delta^p,k_0,k_3}_{\delta^f}(u)du \Biggr.\\&\qquad \quad  -\int_{\delta^f}^t\Psi^{\QQ^{k_0},X}_s\left(-\Sigma^{\delta^f, \delta^p, k_0, k_3}_s(T+\delta^p) -\Sigma^{c, k_0}_{s\land(T-\delta^f)}(T-\delta^f)-\Sigma^{k_0,k_3}_{s\land(T-\delta^f)}(T-\delta^f)\right)ds \\&\qquad \quad - \int_{\delta^f}^t\left(\Sigma^{\delta^f, \delta^p, k_0, k_3}_s(T+\delta^p)+\Sigma^{c, k_0}_{s\land(T-\delta^f)}(T-\delta^f) + \Sigma^{k_0, k_3}_{s\land(T-\delta^f)}(T-\delta^f)\right) dX_s \Biggr\},
    \end{aligned}
    \end{equation}
    from which we deduce that the process \eqref{eq:martexp_h} is a $\QQ^{k_0}$-martingale for every $T\geq \delta^f$.
    \item[] (iii) $\Rightarrow$ (i) The consistency condition and the drift condition yield again equation \eqref{eq:conscondproof3_h}. The martingale property of \eqref{eq:martexp_h} implies then that the forward index spread $(\mathcal{S}^{k_0,k_3}_t(T-\delta^f,T, T+\delta^p))_{t\in[\delta^f,T+\delta^p]}$ is a $(\QQ^{T-\delta^f,k_0,k_3},\GG)$-martingale.
    \end{itemize}
\end{proof}

The drift condition \eqref{eq:driftcondadeltak3k0} allows us to explicitly observe the interplay between the different risk factors that drive the multiplicative spread. In particular, we observe that the drift depends of course on the integrated volatility of the spread $\Sigma^{\delta^f, \delta^p, k_0, k_3}_t(T+\delta^p)$ up to the payment date $T+\delta^p$, but also on the integrated volatility of the instantaneous collateral  forward rate $\Sigma^{c, k_0}_t(T-\delta^f)$ and on that of the instantaneous cross-currency basis spread $\Sigma^{k_0, k_3}_t(T-\delta^f)$, which is an interesting feature from an economic perspective.

\begin{remark}
    Similarly as in Remark \ref{rmk:changepsi}, we notice that the HJM drift condition \eqref{eq:driftcondadeltak3k0} is given in terms of the local exponent $\Psi^{\QQ^{k_0},X}$. However, one can show that for any $\delta^f \le t\le T-\delta^f$ the local exponent $\Psi^{\QQ^{T-\delta^f, k_0, k_3},X}$ under the measure $\QQ^{T-\delta^f, k_0, k_3}$ is obtained from  $\Psi^{\QQ^{k_0},X}$ by
    \begin{equation*}
        \Psi^{\QQ^{T-\delta^f, k_0, k_3},X}(\beta) = \Psi^{\QQ^{k_0},X}(\beta - \Sigma^{c, k_0}(T-\delta^f)- \Sigma^{k_0, k_3}(T-\delta^f))- \Psi^{\QQ^{k_0},X}(- \Sigma^{c, k_0}(T-\delta^f)- \Sigma^{k_0, k_3}(T-\delta^f)),
    \end{equation*}
    for any $\RR^d$-valued predictable and $X$-integrable process $\beta = (\beta_t)_{t\ge 0}$. We further observe from Definition \ref{def:forwardmeasure}, that for $T-\delta^f< t \le T+\delta^p$ the two measures $\QQ^{k_0}$ and $\QQ^{T-\delta^f, k_0, k_3}$ are equivalent up to a multiplicative constant. Hence, basically, we have that
      \begin{equation*}
        \Psi^{\QQ^{T-\delta^f, k_0, k_3},X}(\beta) = \Psi^{\QQ^{k_0},X}(\beta).
    \end{equation*}  
    The HJM drift condition \eqref{eq:driftcondadeltak3k0} can then be rewritten under the forward measure $\QQ^{T-\delta^f, k_0, k_3}$ as
    \begin{align*}
    \int_t^{T+\delta^p}\alpha^{\delta^f, \delta^p, k_0,k_3}_t(u)du=\Psi_t^{\QQ^{T-\delta^f, k_0, k_3},X}(-\Sigma^{\delta^f, \delta^p, k_0,k_3}_t(T+\delta^p)),
    \end{align*}
    for every $\delta^f\le t\le T+\delta^p$. We further notice that, since  $h^{\delta^f, \delta^p, k_0,k_3}_t(\tau) = f^{c, k_0,k_3}_t(\tau)$ for every $T< t\le \tau\le T+\delta^p$ by definition, then $\alpha_t^{\delta^f, \delta^p, k_0,k_3}(\tau) = \alpha_t^{c, k_0}(\tau) + \alpha_t^{k_0, k_3}(\tau)$ and $\Sigma_t^{\delta^f, \delta^p, k_0,k_3}(\tau) = \Sigma_t^{c, k_0}(\tau)+\Sigma_t^{k_0,k_3}(\tau)$ for every $T< t\le \tau\le T+\delta^p$, and the drift condition \eqref{eq:driftcondadeltak3k0} coincides with the drift condition \eqref{eq:driftcondak3k0f}, namely
    \begin{align*}
     \int_t^{T+\delta^p} (\alpha_t^{c, k_0}(u) + \alpha_t^{k_0, k_3}(u))du&= \int_t^{T+\delta^p}\alpha^{\delta^f, \delta^p, k_0,k_3}_t(u)du\\&=\Psi_t^{\QQ^{k_0},X}(-\Sigma^{\delta^f, \delta^p, k_0,k_3}_t(T+\delta^p))\\&= \Psi^{\QQ^{k_0},X}_t(-\Sigma^{c, k_0}_t(T+\delta^p)-\Sigma^{k_0,k_3}_t(T+\delta^p)),
        \end{align*}
        for every $T< t\le  T+\delta^p$.
\end{remark}

\begin{remark}
    For $k_3 = k_0$, $\delta^f = 0$ and $\delta^p = \delta$, the drift condition \eqref{eq:driftcondadeltak3k0} becomes 
\begin{equation}\label{eq:driftcond0delta}
\int_t^{T+\delta}\alpha^{0, \delta, k_0, k_0}_t(u)du= \Psi^{\QQ^{k_0},X}_t(-\Sigma^{0, \delta, k_0, k_0}_t(T+\delta)-\Sigma^{c, k_0}_{t\land T}(T))- \Psi^{\QQ^{k_0},X}_t(-\Sigma^{c,k_0}_{t\land T}(T)).
\end{equation}
Notice that this extends the results of  \cite{Cuchiero2016}, because we model the spread beyond the fixing date, in this case $T$, namely we model the spread up to the payment date, $T+\delta$. This is important since, as discussed before, despite for $T< t \le T+\delta^p$ the index has already been fixed, one still observes fluctuations in the spread due to fluctuations in the foreign-collateral discount curves.  Instead, in \cite{Cuchiero2016}, these fluctuations are implicitly ignored since the spread is only modelled up to the fixing date $T$. In particular, this would mean to set
\begin{equation*}
    \int_T^{T+\delta^p} \alpha_t^{0, \delta, k_0, k_0}(u)du=\int_T^{T+\delta^p} \sigma_t^{0, \delta, k_0, k_0}(u)du = 0,
\end{equation*}
hence $\Sigma^{0, \delta, k_0, k_0}_t(T+\delta) = \Sigma^{0, \delta, k_0, k_0}_t(T)$. Under these assumptions, we observe that the drift condition \eqref{eq:driftcond0delta} coincides indeed with the drift condition found in \cite[Theorem 3.15]{Cuchiero2016}.
\end{remark}

\begin{corollary}
    Let $1\leq k_0\leq L$ and $\delta^f, \delta^p\ge 0$. If the multiplicative spread model for the currency $k_0$ is risk neutral, then for every $1\le k_3\le L$ we have that: 
    \begin{enumerate}
        \item For every $t\ge  \delta^f$ and $\tau\ge \delta=\delta^f+\delta^p$, the instantaneous multiplicative spread forward rate is given by
        \begin{equation*}
        \begin{aligned}
        &h^{\delta^f, \delta^p, k_0,k_3}_t(\tau)=h^{\delta^f, \delta^p, k_0,k_3}_{\delta^f}(\tau)\\&-\int_{\delta^f}^t
        \Biggl(\left(\sigma^{\delta^f, \delta^p,k_0,k_3}_s(\tau)+\sigma^{c, k_0}_{s\land (\tau-\delta)}(\tau-\delta) + \sigma^{k_0, k_3}_{s\land (\tau-\delta)}(\tau-\delta)\right)\cdot
        \Biggr.\\
        &\qquad   \cdot\nabla\Psi_s^{\QQ^{k_0}, X}\left(-\Sigma^{\delta^f, \delta^p,k_0,k_3}_s(\tau)-\Sigma^{c, k_0}_{s\land (\tau-\delta)}(\tau-\delta)-\Sigma^{k_0, k_3}_{s\land (\tau-\delta)}(\tau-\delta)\right)\\
   &\qquad \Biggl.-\left(\sigma^{c, k_0}_{s\land (\tau-\delta)}(\tau-\delta) + \sigma^{k_0, k_3}_{s\land (\tau-\delta)}(\tau-\delta)\right)\nabla\Psi_s^{\QQ^{k_0}, X}\left(-\Sigma^{c, k_0}_{s\land (\tau-\delta)}(\tau-\delta)-\Sigma^{k_0, k_3}_{s\land (\tau-\delta)}(\tau-\delta)\right)\Biggr)ds\\&+\int_{\delta^f}^t\sigma^{\delta^f, \delta^p,k_0,k_3}_s(\tau)dX_s;
        \end{aligned}
    \end{equation*}
        \item For every $t\ge\delta$, the multiplicative spread short rate $h^{\delta^f, \delta^p, k_0,k_3}_t$ at time $t$ is given by
         \begin{equation*}
        \begin{aligned}
        &h^{\delta^f, \delta^p, k_0,k_3}_t=h^{\delta^f, \delta^p, k_0,k_3}_t(t)=h^{\delta^f, \delta^p, k_0,k_3}_{\delta^f}(t)\\&-\int_{\delta^f}^t
        \Biggl(\left(\sigma^{\delta^f, \delta^p,k_0,k_3}_s(t)+\sigma^{c, k_0}_{s\land (t-\delta)}(t-\delta) + \sigma^{k_0, k_3}_{s\land (t-\delta)}(t-\delta)\right)\cdot
        \Biggr.\\
        &\qquad   \cdot\nabla\Psi_s^{\QQ^{k_0}, X}\left(-\Sigma^{\delta^f, \delta^p,k_0,k_3}_s(t)-\Sigma^{c, k_0}_{s\land (t-\delta)}(t-\delta)-\Sigma^{k_0, k_3}_{s\land (t-\delta)}(t-\delta)\right)\\
   &\qquad \Biggl.-\left(\sigma^{c, k_0}_{s\land (t-\delta)}(t-\delta) + \sigma^{k_0, k_3}_{s\land (t-\delta)}(t-\delta)\right)\nabla\Psi_s^{\QQ^{k_0}, X}\left(-\Sigma^{c, k_0}_{s\land (\tau-\delta)}(t-\delta)-\Sigma^{k_0, k_3}_{s\land (t-\delta)}(t-\delta)\right)\Biggr)ds\\&+\int_{\delta^f}^t\sigma^{\delta^f, \delta^p,k_0,k_3}_s(t)dX_s.
        \end{aligned}
    \end{equation*}
    \end{enumerate}
\end{corollary}
\begin{proof}
    The proof proceeds similarly to the proof of Corollary \ref{cor:driftcondk0k0}, starting from \eqref{eq:hjmspread} and using the drift condition \eqref{eq:driftcondadeltak3k0} for $\tau :=T+\delta^p$. For the short rate, we let $t\to \tau$.
\end{proof}

We conclude this section by deriving  the HJM framework for an abstract index by combining the multiplicative spread model for the currency $k_0$ with the HJM framework for the extended bond-price model in Section \ref{sec:forCollDiscCurv}. In particular, from Definition \ref{def:multSpreadIndex}, for any fixed $1\leq k_0,k_3\leq L$, any $\delta^f, \delta^p\ge 0$ and for any $T\ge \delta^f$, the forward on the abstract index  $I^{k_0}(T-\delta^f,  T, T+\delta^p)$ at time $\delta^f \le t\le T+\delta^p$ for the time period $[T-\delta^f,T+\delta^p]$ collateralized according to $B^{c,k_0,k_3}$ and payed at time $T+\delta^p$ is the product between the multiplicative spread and the discount curve index, namely
\begin{equation*}
I_t^{k_0,k_3}(T-\delta^f, T,T + \delta^{p})=\frac{1}{\delta}\left(\mathcal{S}^{k_0,k_3}_t(T-\delta^f, T,T + \delta^{p})\bigl(1+\delta I_t^{k_0,k_3, D}(T-\delta^f, T+ \delta^{p},T+ \delta^{p})\bigr)-1\right),
\end{equation*}
where $\delta=\delta^f+\delta^p$. From \eqref{eq:spreadexplicit} we further write that
\begin{equation}\label{eq:indexexplicit}
\begin{aligned}
   I_t^{k_0,k_3}&(T-\delta^f, T,T + \delta^{p}) \\&=
   \begin{cases}
         \frac{1}{\delta}\left(\mathcal{S}^{k_0,k_3}_t(T-\delta^f, T,T + \delta^{p})\frac{B^{k_0,k_3}(t,T-\delta^f)}{B^{k_0,k_3}(t,T+\delta^p)}-1\right), & \delta^f \le t\le  T-\delta^f,\\
       \frac{1}{\delta}\left(\mathcal{S}^{k_0,k_3}_t(T-\delta^f, T,T + \delta^{p}) \frac{B_t^{c, k_0, k_3}}{B_{T-\delta^f}^{c, k_0, k_3}} \frac{1}{B^{k_0,k_3}(t,T+\delta^p )}-1\right), & T-\delta^f < t \le T,\\
       I_T^{k_0}(T-\delta^f, T,T + \delta^{p}), & T < t \le T+\delta^p.
   \end{cases}
\end{aligned}
\end{equation}
By combining equations \eqref{eq:bondspread2}, \eqref{eq:bondspread3}, \eqref{eq:Bck0k3} and \eqref{eq:hjmspreadS}, we can also rewrite the index in terms of the instantaneous collateral  forward rate, the cross-currency basis spread and of the instantaneous multiplicative spread forward rate, namely
\begin{equation*}
I_t^{k_0,k_3}(T-\delta^f, T,T + \delta^{p}) =\frac{1}{\delta}\left(e^{-\int_{\delta^f}^t h^{\delta^f, \delta^p, k_0,k_3}_s ds-\int_t^{T+\delta^p} h^{\delta^f, \delta^p, k_0,k_3}_t(u)du+\int_{T-\delta^f}^{T+\delta^p}\left(f_t^{c, k_0}(u)+q_t^{k_0,k_3}(u)\right)du}-1\right),
\end{equation*}
for $\delta^f \le t\le  T-\delta^f$, and
\begin{equation*}
\begin{aligned}
  & I_t^{k_0,k_3}(T-\delta^f, T,T + \delta^{p}) \\&=
       \frac{1}{\delta}\left(e^{-\int_{\delta^f}^t h^{\delta^f, \delta^p, k_0,k_3}_s ds-\int_t^{T+\delta^p} h^{\delta^f, \delta^p, k_0,k_3}_t(u)du+\int_{T-\delta^f}^t\left(r_s^{c, k_0}+q_s^{k_0,k_3}\right)ds -\int_{t}^{T+\delta^p}\left(f_t^{c, k_0}(u)+q_t^{k_0,k_3}(u)\right)du} -1\right),
\end{aligned}
\end{equation*}
for $T-\delta^f < t \le T$. Notice that for $T<t\le T+\delta^p$ the forward equals the index at time $T$, namely $I_t^{k_0,k_3}(T-\delta^f, T,T + \delta^{p})= I_T^{k_0}(T-\delta^f, T,T + \delta^{p})$. Hence, in particular, it is constant.

\section{Application: cross-currency swaps pricing}\label{sec:example}
We provide an in-depth study of cross-currency swap contracts. The motivation for studying these instruments is twofold: on the one hand, cross-currency swap contracts offer the perfect example of instruments which depend on all the sources of risk that we have described in the previous sections. On the other hand, studying cross-currency swap contracts serves as a starting point for analysing 
the benchmark transition, which is still not treated in the literature. In particular, in view of the benchmark reform, we shall describe the legs of the contracts in terms of some abstract indices, in line with the approach adopted in Section \ref{sec:abstractIndeces}. We proceed as follows: we consider first the case of constant-notional cross-currency swaps and provide the corresponding pricing formulas under different collateral currencies. Then, we consider resetting cross-currency swaps, and, finally, we study potential leg asymmetries which are originated by the LIBOR transition.

We consider two generic currencies $1\leq k_0, k\leq L$ and a time interval $[\tau^s, \tau^e]$, with $\tau^e\ge \tau^s \ge 0$, representing the monitoring period of the cross-currency swap contract which is stipulated by two agents at time $0$.  In general, the number of legs for a swap can be arbitrary, most typically for bespoke over-the-counter (OTC) contracts. However, for simplicity, we consider contracts with only two legs and we use the index $k_0$ and the index $k$ to denote, respectively, the domestic and the foreign leg. Let then $N^{k_0}\in\NN$ and $N^{k}\in\NN$ be the numbers of time intervals in the two schedules of cash flows happening within the interval $[\tau^s, \tau^e]$ and referring, respectively, to $k_0$ and to $k$. Notice that, in general, the number of cash flows can be different among different legs, typically due to different payment frequencies. We then denote the time instants for the domestic leg by $t^{k_0}_n$, with $0\leq n\leq N^{k_0}$, and the time instants for the foreign leg by $t^{k}_n$, with $0\leq n\leq N^{k}$. In particular, we set $t_0^{k_0}=t_0^{k}=\tau^s$ and $t_{N^{k_0}}^{k_0}=t_{N^{k}}^{k}=\tau^e$. For each interval of the form $(t^{k_0}_{n-1}, t^{k_0}_n]$, we define $\delta^{k_0}_n:= t^{k_0}_n-t^{k_0}_{n-1}$, for $1\leq n\leq N^{k_0}$. Similarly, we define $\delta^{k}_n:= t^{k}_n-t^{k}_{n-1}$, for $1\leq n\leq N^{k}$. We then introduce $\delta^{k_0,f}_n$, $\delta^{k_0,p}_n$, $\delta^{k,f}_n$ and $\delta^{k,p}_n$ such that $\delta^{k_0,f}_n+\delta^{k_0,p}_n=\delta^{k_0}_n$ and  $\delta^{k,f}_n+\delta^{k,p}_n=\delta^{k}_n$, respectively.

We further introduce $\phi\in\{+1,-1\}$ as the indicator for long and short positions, respectively, and $\cN^{k_0}$ and $\cN^{k}$ for the constant notional of the contract in domestic and foreign currency, respectively. Notice that a schedule of time-varying notional could be included by writing $\cN^{k_0}_{t_n^{k_0}}$ and $\cN^{k}_{t_n^{k}}$, instead. Spreads can also be included in either leg by means of the quantities $\mathcal{S}^{k_0}_0(\tau^e)$ and $\mathcal{S}^{k}_0(\tau^e)$. The standard market practice is to quote cross-currency swaps against USD and to add the spread to the non-USD leg. Notice also that the two spreads are fixed at the moment of the stipulation of the contract, namely in $0$, and depend on the final horizon of the monitoring period, namely $\tau^e$. As in the previous sections, we denote with $k_3$ the currency of denomination of the collaterals. 

\subsection{Constant-notional cross-currency swaps}\label{sec:CCS}
The simplest form of cross-currency swap involves two agents who lend to each other notional amounts in two different currencies. The notional is swapped at the initial time, $\tau^s$, and then swapped back at the terminal time, $\tau^e$. We denote by $CCS^{k_0, k_3}$ the value in domestic currency of the constant-notional cross-currency swap collateralized in currency $k_3$, which at time $t\ge 0$ is given by
\begin{align*}
    CCS_t^{k_0, k_3}=\phi\left(S_t^{k_0}(A_{CCS}^{k_0},C^{k_3})-\mathcal{X}^{k_0,k}_tS_t^{k}(A_{CCS}^{k},C^{k_3})\right), 
\end{align*}
with $A_{CCS}^{k_0}$ and $A_{CCS}^{k}$ representing the cash flow of the domestic and foreign legs associated to two generic market indices $I^{k_0}$ and $I^{k}$. In particular, for each $ t \le \tau^e$, the leg $\ell\in\{k_0, k\}$ is evaluated via
\begin{equation}\label{eq:CCYk_0}
\begin{aligned}
    S_t^{\ell}(A_{CCS}^{\ell},C^{k_3})&= \cN^{\ell}\Biggl(-B^{\ell,k_3}(t,\tau^s)\mathbb{I}_{\left\{t\le \tau^s\right\}}+B^{\ell,k_3}(t,\tau^e)\Biggr.\\&\Biggl.+\sum_{n=1}^{N^{\ell}}\delta_n^{\ell}\Excond{\QQ^{\ell}}{\frac{B^{c,\ell, k_3}_t}{B^{c,\ell, k_3}_{t_n^{\ell}}}\left(I^{\ell}_{t_{n-1}^{\ell}+\delta_n^{\ell,f}}(t_{n-1}^{\ell}, t_{n-1}^{\ell}+\delta_n^{\ell,f}, t_n^{\ell})+\mathcal{S}^{\ell}_0(\tau^e)\right)}{\cG_t}\mathbb{I}_{\left\{t\le t_n^{\ell}\right\}}\Biggr),
     \end{aligned}
\end{equation}  
where the indicator function $\mathbb{I}$ means that each term in the summations exists until its payment date, namely the $n$-th term in the domestic leg disappears for $t> t_n^{k_0}$, and the $n$-th term in the foreign leg disappears for $t> t_n^{k}$. Notice that the same formulas allow also to treat the case of fixed-versus-fixed and fixed-versus-floating cross-currency swaps by suitably setting the desired indices to zero and interpreting the spreads $\mathcal{S}_0^{k_0}(\tau^e)$ and $\mathcal{S}_0^{k}(\tau^e)$ as fixed rates. 

We now illustrate how the modeling quantities analyzed in the present work are crucial in the evaluation of the formula \eqref{eq:CCYk_0}. 
By Definition \ref{def:abstractForwardIndex}, we express the leg $\ell\in\{k_0, k\}$ in terms of the forward contract $I^{\ell, k_3}$ written on the index $I^{\ell}$ with collateralization in currency $k_3$, namely
\begin{align}\label{eq:ccs}
\begin{aligned}
   S_t^{\ell}(A_{CCS}^{\ell},C^{k_3})&= \cN^{\ell}\Biggl(-B^{\ell,k_3}(t,\tau^s)\mathbb{I}_{\left\{t\le \tau^s\right\}}+B^{\ell,k_3}(t,\tau^e)\Biggr.\\&\Biggl.+\sum_{n=1}^{N^{\ell}}\delta_n^{\ell}\left(I^{\ell, k_3}_{t}(t_{n-1}^{\ell}, t_{n-1}^{\ell}+\delta_n^{\ell,f}, t_n^{\ell})+\mathcal{S}_0^{\ell}(\tau^e)\right)B^{\ell,k_3}(t,t_n^{\ell})\,\mathbb{I}_{\left\{t\le t_n^{\ell}\right\}}\Biggr).
\end{aligned}
\end{align}
In particular, for every $t\ge 0$, the forward contract $I^{\ell, k_3}$ can be written in compact form starting from equation \eqref{eq:indexexplicit} as
\begin{equation}\label{eq:indexexplicit2}
    I^{\ell, k_3}_{t}(t_{n-1}^{\ell}, t_{n-1}^{\ell}+\delta_n^{\ell,f}, t_n^{\ell}) = \frac{1}{\delta_n^\ell}\left(\frac{B^{c,\ell,k_3}_t}{B^{c, \ell,k_3}_{t\land t_{n-1}^{\ell}}}\frac{B^{\ell,k_3}(t\land t_{n-1}^{\ell},t_{n-1}^{\ell})}{B^{\ell,k_3}(t,t_{n}^{\ell})}\mathcal{S}^{\ell,k_3}_t(t_{n-1}^{\ell}, t_{n-1}^{\ell}+\delta_n^{\ell,f}, t_n^{\ell})-1\right).
\end{equation}
By substituting  \eqref{eq:indexexplicit2} into \eqref{eq:ccs}, we get
\begin{align*}
\begin{aligned}
   &S_t^{\ell}(A_{CCS}^{\ell},C^{k_3})=\cN^{\ell}\Biggl(-B^{\ell,k_3}(t,\tau^s)\mathbb{I}_{\left\{t\le \tau^s\right\}}+B^{\ell,k_3}(t,\tau^e) \Biggr.\\ \Biggl.&+\sum_{n=1}^{N^{\ell}}\left(\frac{B^{c,\ell,k_3}_t}{B^{c, \ell,k_3}_{t\land t_{n-1}^{\ell}}}\frac{B^{\ell,k_3}(t\land t_{n-1}^{\ell},t_{n-1}^{\ell})}{B^{\ell,k_3}(t,t_{n}^{\ell})}\mathcal{S}^{\ell,k_3}_t(t_{n-1}^{\ell}, t_{n-1}^{\ell}+\delta_n^{\ell,f}, t_n^{\ell})-1+\delta_n^{\ell}\mathcal{S}_0^{\ell}(\tau^e)\right) B^{\ell,k_3}(t,t_n^{\ell})\,\mathbb{I}_{\left\{t\le t_n^{\ell}\right\}}\Biggr).
\end{aligned}
\end{align*}
Finally, we use the definition of the foreign-collateral discount curves as product between the domestic bond with domestic collateral and the $\ell$-$k_3$ cross-currency bond spread in equation \eqref{eq:bondspread2} and obtain
\begin{align*}
\begin{aligned}
   S_t^{\ell}(A_{CCS}^{\ell},C^{k_3})=&\cN^{\ell}\Biggl(-B^{\ell,\ell}(t,\tau^s)Q^{\ell,k_3}(t,\tau^s)\mathbb{I}_{\left\{t\le \tau^s\right\}}+B^{\ell,\ell}(t,\tau^e)Q^{\ell,k_3}(t,\tau^e) \Biggr.\\ &+\sum_{n=1}^{N^{\ell}}\Bigl(\,\frac{B^{c,\ell}_tQ^{\ell,k_3}_t}{B^{c, \ell}_{t\land t_{n-1}^{\ell}}Q^{\ell,k_3}_{t\land t_{n-1}^{\ell}}}\frac{B^{\ell,\ell}(t\land t_{n-1}^{\ell},t_{n-1}^{\ell})Q^{k_0,k_3}(t\land t_{n-1}^{\ell},t_{n-1}^{\ell})}{B^{\ell,\ell}(t,t_{n}^{\ell})Q^{\ell,k_3}(t,t_{n}^{k_0})}\cdot \Bigr.\\&\hspace{1.5cm} \Biggr.\cdot\Bigl.   
   \mathcal{S}^{\ell,k_3}_t(t_{n-1}^{\ell}, t_{n-1}^{\ell}+\delta_n^{\ell,f}, t_n^{\ell})-1+\delta_n^{\ell}\mathcal{S}_0^{\ell}(\tau^e)\Bigr) B^{\ell,\ell}(t,t_n^{\ell})Q^{\ell,k_3}(t,t_n^{\ell})\,\mathbb{I}_{\left\{t\le t_n^{\ell}\right\}}\Biggr).
\end{aligned}
\end{align*}
This shows how all the quantities modelled in the previous sections enter into play in the evaluation of cross-currency swap contracts. 
In particular, the formulas obtained also highights that constant-notional cross-currency swaps are linear with respect to all the quantities introduced in our cross-currency HJM framework.

\subsection{Resetting cross-currency swaps}\label{sec:MtMCCS} A resetting (or marked-to-market) cross-currency swap (MtMCCS) is constructed via a sequence of one-period cross-currency swaps. Every single cross-currency swap in the sequence involves an exchange of notional, so that the contract can be interpreted as a rolling strategy on loans with varying notional. The rolling mechanism implies exchanges of notional at every payment date, and it reduces the outstanding counterparty risk exposure. This version of the instrument is cheaper if one takes into account the possibility of default of the agents which are involved in the transaction.

Typically, quoted instruments feature notional resets on the leg indexed to the \emph{stronger} currency (e.g., USD for the EURUSD pair), and a constant notional for the \emph{weaker} currency (i.e., EUR in the EURUSD pair example). In this case, the spread is introduced only on the weaker leg and it is fixed in such a way that the value of the contract at initiation is zero. This means in particular that the spread is the market quote. For our example, we shall analyze the case when the notional resets affect either the $k_0$- or the $k$-denominated leg. The collateralization is in currency $k_3$.

We then denote by $MtMCCS^{k_0, k_3}$ the value in domestic currency of the resetting cross-currency swap collateralized in currency $k_3$. With the notional resets affecting the $k_0$-denominated leg, for $t \le \tau^e$ we have that
\begin{align*}   
MtMCCS^{k_0, k_3}=\phi\left(S_t^{k_0}(A_{MtMCCS}^{k_0},C^{k_3})-\mathcal{X}^{k_0,k}_tS^k_t(A_{CCS}^{k},C^{k_3})\right),
\end{align*}
where $A_{MtMCCS}^{k_0}$ represents the cash flow of the domestic leg with resetting. Similarly, if the notional resets are affecting the foreign leg instead, the valuation formula is
\begin{align*}
MtMCCS^{k_0, k_3}=\phi\left(S^{k_0}_t(A_{CCS}^{k_0},C^{k_3})-\mathcal{X}^{k_0,k}_tS^k_t(A_{MtMCCS}^{k},C^{k_3})\right),
\end{align*}
with $A_{MtMCCS}^{k}$ being the cash flow of the foreign leg with resetting. 
In particular, then the value of the leg with resetting is given for $\ell\in\{k_0, k\}$ by
\begin{equation}\label{eq:domlegres}
\begin{aligned}
    &S_t^{\ell}(A_{MtMCCS}^{\ell},C^{k_3})=\\&\cN^{\kappa}\left(\sum_{n=1}^{N^{\ell}}\Excond{\QQ^{\ell}}{\frac{B^{c,\ell, k_3}_t}{B^{c,\ell, k_3}_{t_n^{\ell}}}\mathcal{X}^{\ell, \kappa}_{t_{n-1}^{\ell}}\left(1+\delta_n^{\ell}\left(I^{\ell}_{t_{n-1}^{\ell}+\delta_n^{\ell,f}}(t_{n-1}^{\ell}, t_{n-1}^{\ell}+\delta_n^{\ell,f}, t_{n}^{\ell})+\mathcal{S}_0^{\ell}(\tau^e)\right)\right)}{\cG_t}\mathbb{I}_{\left\{t\le t_n^{\ell}\right\}}\right.\\
    &\left.-\sum_{n=1}^{N^{\ell}}\Excond{\QQ^{\ell}}{\frac{B^{c,\ell, k_3}_t}{B^{c,\ell, k_3}_{t_{n-1}^{\ell}}}\mathcal{X}^{\ell,\kappa}_{t_{n-1}^{\ell}}}{\cG_t}\mathbb{I}_{\left\{t\le t_n^{\ell}\right\}}\right),
\end{aligned}
\end{equation}
where we notice that the notional for the leg $\ell\in\{k_0, k\}$ is now denominated in the other currency, namely in currency $\kappa\in\{k_0, k\}$ with $\kappa \ne \ell$. 
The leg with no resetting, $S^{\kappa}_t(A_{CCS}^{\kappa},C^{k_3})$, is evaluated as in \eqref{eq:CCYk_0} for $\kappa\in\{k_0, k\}$. Starting from \eqref{eq:domlegres}, with similar steps as in Section \ref{sec:CCS}, one recovers all the modelling quantities studied in the paper, including the dynamics for the FX rate in Section \ref{sec:FXrate}.

\subsection{The impact of the LIBOR transition}
For concreteness, we take the perspective of a USD agent who entered at time $t=0$ before any benchmark reform into an OTC EURUSD cross-currency swap with notional resets and fixed spread applied on the EUR-denominated leg. In this case, we then have $k_0=\$$ and $k=\text{\euro}$. As a legacy product, we assume that the swap was entered when the USD 3M LIBOR and the 3M EURIBOR were the market standard floating rates applied to the two legs in this kind of contracts. Both the floating rates are fixed at the beginning of the period and paid at the end, so in this case $\delta^{\text{\euro},f}_n=\delta^{\$,f}_n=0$ for all $n$, which allows us to use the more familiar notation
\begin{align*}
    I^{\ell}_{t_{n-1}^{\ell}+\delta_n^{\ell,f}}(t_{n-1}^{\ell}, t_{n-1}^{\ell}+\delta_n^{\ell,f}, t_{n}^{\ell})&=I^{\ell}_{t_{n-1}^{\ell}}(t_{n-1}^{\ell}, t_{n-1}^{\ell}, t_{n}^{\ell})=\mathcal{I}^{\ell}_{t_{n-1}^{\ell}}(t_{n-1}^{\ell}, t_{n}^{\ell}), \ \mbox{for }\ell\in\{\text{\euro},\$\},
\end{align*}
to denote the IBOR rate of the two currencies which is fixed in $t_{n-1}^{\ell}$ for the interval $(t_{n-1}^{\ell},t_{n}^{\ell}]$, and is paid in $t_{n}^{\ell}$. In line with the prevailing market standard, we assume that the collateral is exchanged in USD, namely $r^{c,k_3}=r^{c,k_0}=r^{c,\$}$, with $r^{c,\$}$ being initially the Fed Fund rate which is an unsecured overnight rate. The value of the resetting cross-currency swap is then
\begin{align*}
MtMCCS^{\$, \$}=\phi\left(S^{\$}_t(A_{CCS}^{\$},C^{\$})-\mathcal{X}^{\$,\text{\euro}}_tS^{\text{\euro}}_t(A_{MtMCCS}^{\text{\euro}},C^{\$})\right),
\end{align*}
where, before the LIBOR discontinuation, the USD leg is evaluated as in \eqref{eq:CCYk_0}, namely
\begin{align} \label{eq:usdleg}
\begin{aligned}
S^{\$}_t(A_{CCS}^{\$},C^{\$})&=\cN^{\$}\Biggl(-B^{\$,\$}(t,\tau^s)\mathbb{I}_{\left\{t\le \tau^s\right\}}+B^{\$,\$}(t,\tau^e)\Biggr.\\&\Biggl.+\sum_{n=1}^{N^{\$}}\delta_n^{\$}\Excond{\QQ^{\$}}{\frac{B^{c,\$}_t}{B^{c,\$}_{t_n^{\$}}}\mathcal{I}^{\$}_{t_{n-1}^{\$}}(t_{n-1}^{\$}, t_{n}^{\$})}{\cG_t}\mathbb{I}_{\left\{t\le t_n^{\$}\right\}}\Biggr)\\
&=\cN^{\$}\Biggl(-B^{\$,\$}(t,\tau^s)\mathbb{I}_{\left\{t\le \tau^s\right\}}+B^{\$,\$}(t,\tau^e)\Biggr.\\&\Biggl.+\sum_{n=1}^{N^{\$}}\delta_n^{\$}B^{\$,\$}(t,t_n^{\$})\Excond{\QQ^{t_n^{\$},\$,\$}}{\mathcal{I}^{\$}_{t_{n-1}^{\$}}(t_{n-1}^{\$}, t_{n}^{\$})}{\cG_t}\mathbb{I}_{\left\{t\le t_n^{\$}\right\}}\Biggr).
\end{aligned}
\end{align}
Here the spread $\mathcal{S}_0^{\text{\$}}(\tau^e)=0$ because the USD currency is the strong one in the USD-EUR pair. The EUR leg is instead evaluated as in \eqref{eq:domlegres} by
\begin{align*}
    \begin{aligned}
    S_t^{\text{\euro}}(A_{MtMCCS}^{\text{\euro}},C^{\$})&=\cN^{\$}\left(\sum_{n=1}^{N^{\text{\euro}}}\Excond{\QQ^{\text{\euro}}}{\frac{B^{c,\text{\euro}, \$}_t}{B^{c,\text{\euro}, \$}_{t_n^{\text{\euro}}}}\mathcal{X}^{\text{\euro}, \$}_{t_{n-1}^{\text{\euro}}}\left(1+\delta_n^{\text{\euro}}\left(\mathcal{I}^{\text{\euro}}_{t_{n-1}^{\text{\euro}}}(t_{n-1}^{\text{\euro}}, t_{n}^{\text{\euro}})+\mathcal{S}_0^{\text{\euro}}(\tau^e)\right)\right)}{\cG_t}\mathbb{I}_{\left\{t\le t_n^{\text{\euro}}\right\}}\right.\\
    &\left.-\sum_{n=1}^{N^{\text{\euro}}}\Excond{\QQ^{\text{\euro}}}{\frac{B^{c,\text{\euro}, \$}_t}{B^{c,\text{\euro}, \$}_{t_{n-1}^{\text{\euro}}}}\mathcal{X}^{\text{\euro},\$}_{t_{n-1}^{\text{\euro}}}}{\cG_t}\mathbb{I}_{\left\{t\le t_n^{\text{\euro}}\right\}}\right),
\end{aligned}
\end{align*}
where the rate $r^{c,\text{\euro}}$ was initially represented by the EONIA rate. We remark that the expression above depends also on the infinitesimal cross-currency basis $q^{\$,\text{\euro}}$. 

The benchmark reform had multiple impacts on the valuation of this type of products. The first one is a switch of the collateral rate by central counterparties  for both the currency areas, which also became the market standard for OTC markets. On 27 July 2020, the central counterparties switched $r^{c,\text{\euro}}$ from EONIA to ESTR. Subsequently, on 16 October 2020, the rate $r^{c,\$}$ has been switched from the Fed Fund rate to SOFR. Both the switches of discount rates triggered exchanges of cash flows among agents. Even more relevant is the discontinuation of the USD 3M LIBOR rate, which is impacting directly the USD leg. In this case, the agents can agree on different alternatives, all of which are captured by our framework.

A first choice is to replace the USD 3M LIBOR by means of AMERIBOR T90. In this case, the new benchmark is still an unsecured forward-looking rate that embeds a dynamic credit component, hence the valuation formula \eqref{eq:usdleg} is virtually unchanged. A second choice is to adopt the ISDA fallback protocol\footnote{\url{https://assets.isda.org/media/3062e7b4/23aa1658-pdf/}}, where USD 3M LIBOR is replaced by the sum of the backward-looking rate and a fixed credit spread. This means that in \eqref{eq:usdleg} we replace the spot index $\mathcal{I}^{\$}_{t_{n-1}^{\$}}(t_{n-1}^{\$}, t_{n}^{\$})$ by\footnote{Concretely, the formula involves a daily compounding of rates over the whole interval $(t_{n-1}^{\$},t_{n}^{\$}]$ which we do not present here to avoid the introduction of further notations.}
\begin{align*}
  I^{\$}_{t_{n-1}^{\$}}(t_{n-1}^{\$}, t_{n}^{\$}, t_{n}^{\$}):=\frac{1}{t_{n}^{\$}-t_{n-1}^{\$}}\left(e^{\int_{t_{n-1}^{\$}}^{t_{n}^{\$}} r^{c,\$}_u du}-1\right)+ \mathcal{CS}=\frac{1}{t_{n}^{\$}-t_{n-1}^{\$}}\left(\frac{B^{c,\$}_{t_{n}^{\$}}}{B^{c,\$}_{t_{n-1}^{\$}}}-1\right) + \mathcal{CS},
\end{align*}
where $\mathcal{CS}$ denotes a credit spread between the USD 3M LIBOR and the backward-looking rate. This is usually estimated as a mean or a median on the basis of historical data. For an in-depth analysis of LIBOR fallbacks from a quantitative perspective, we refer to \cite{henrard2019}. We limit ourselves to notice that a constant credit spread is a questionable choice, since it is clearly unable to capture the dynamic nature of inter-bank risk, \cite{fitr12}.

In the aftermath of the benchmark reform, cross-currency swaps still fall within our HJM framework: the market practice has moved to the use of the backward-looking rate on both legs.

\section{Conclusions and final remarks}
We introduced a general Heath-Jarrow-Morton framework that allows to simultaneously model multiple-curve interest rate markets denominated in different currencies. The framework is flexible enough to capture the stylized facts illustrated in the introduction, including the recent transition away from unsecured interbank rates in the US and the UK markets. The framework's inherent flexibility and generality provide ground for exploring diverse directions for future research.

In the present work we made the choice not to consider options, but to focus on cross-currency swaps as a minimal example of a contract that highlights the relevance of our framework. Having established such sound foundation, a possible natural direction for future research involves the study of further non-linear European derivatives, also in view of model calibration. Recalling that market quotes refer to fully collateralized instruments, the pricing formula \eqref{eq:fullCollateral} provides the tool to derive the evaluation of any fully collateralized European claim relevant for the calibration procedure. Given initial interest rate curves obtained via a bootstrap procedure, \eqref{eq:fullCollateral} can indeed be used to compute the price of futures/caps/floors/swaptions for the base currency $k_0$ and for all other foreign interest rate markets. In a subsequent step, one can compute the price of FX futures and options relevant for the calibration of the spot FX models. In line with market practice, a cascade procedure can be devised where first all single-currency interest rate models are calibrated and, subsequently, the obtained interest rate models are fed as input for determining the drifts in the calibration of FX models. Such a high-dimensional calibration procedure is beyond the scope of the present work.

A further potential direction for future investigations would involve the inclusion of stochastic discontinuities as in \cite{fontana2023b}.

\appendix
\section{Multi-currency trading in the basic market}\label{sec:MultiCurrTrading}
The objective of the present and subsequent sections is to construct a general arbitrage-free multiple-curve cross-currency market with an arbitrary number of currencies and risky assets in each currency area. On such typically incomplete market, we study the pricing of contingent claims when collateral is exchanged in an arbitrary currency denomination. When the pricing formula that we derive here is applied to ZCBs or forward on indices, we obtain the basic objects that are modelled under the HJM framework in the main body of the paper and that are used to construct useful forward measures.

We follow the notation of \cite{gnoSei2021}. Let $T>0$ be a fixed time horizon and $\left(\Omega,\cG,\GG,\PP\right)$ be a filtered probability space  with the filtration $\GG=\left(\cG_t\right)_{t\in [0,T]}$ satisfying the usual conditions. Here $\cG_0$ is assumed to be trivial, and all the processes to be introduced in the sequel are assumed to be $\GG$-adapted right-continuous with left limits (RCLL) semimartingales. 

We then postulate the existence of $L\in\mathbb{N}$ economies, and we introduce the following indices ranging from $1$ to $L$ in order to distinguish between all the possible scenarios:
\begin{enumerate}
\item $k_0$ denotes the currency of denomination of the portfolio, hence represents the \emph{domestic currency};
\item $k_1$ denotes the currency of denomination of the risky assets, the associated repo cash accounts, and of the unsecured funding accounts;
\item $k_2$ denotes the currency of denomination for the contractual cash-flows;
\item $k_3$ denotes the currency of denomination of the collateral;
\item $k$ will be used to denote a general currency.
\end{enumerate}
We set to $d_{k_1}$ the number of risky assets which are traded in terms of the currency with index $k_1$. Then $S^{i,k_1}$ denotes the ex-dividend price of the $i$-th risky asset traded in units of currency $k_1$, and $D^{i,k_1}$ is the corresponding cumulative dividend stream, for every $i=1,\ldots, d_{k_1}$. 

The trading desk uses different sources of funding, each being represented by a suitable family of cash accounts.  In particular, for each risky asset there is an asset-specific funding account, which we call the \emph{repo account}. We denote by $B^{i,k_1}$ the funding account associated to the asset $S^{i,k_1}$. Positive or negative dividends from the risky asset $S^{i, k_1}$ are invested in the corresponding funding account $B^{i,k_1}$. For unsecured funding, we assume that the trading desk can fund its activity by unsecured borrowing and lending in different currencies. We then introduce the cash accounts $B^{k_1}:=B^{0,k_1}$ with unsecured rates $r^{k_1}$, for $k_1=1,\ldots,L$. Moreover, we use the symbol $\,\widehat{\cdot}\,$ to denote quantities which are discounted by means of their corresponding repo account,  and the symbol $\,\widetilde{\cdot}\,$ for quantities which are discounted by means of the corresponding unsecured account. For example, for the risky asset $S^{i, k_1}$ we write $$\hat{S}^{i, k_1} := \frac{S^{i, k_1}}{B^{i, k_1}}, \quad \mbox{ and }\quad  \tilde{S}^{i, k_1} := \frac{S^{i, k_1}}{B^{k_1}}.$$  Let us recall that $\mathcal{X}^{k_0,k}$ is the price of one unit of currency $k$ in terms of currency $k_0$, for every $k\ne k_0$. 

We work under the following assumptions as in \cite[Assumption 2.1]{gnoSei2021}.
\begin{assumption} We assume that:
\begin{enumerate}
\item For all $i=1,\ldots,d_{k_1}$, and all $k_1$, the ex-dividend price processes $S^{i,k_1}$ are real-valued RCLL semimartingales;
\item For all $i=1,\ldots,d_{k_1}$, and all $k_1$, the cumulative dividend streams $D^{i,k_1}$ are processes of finite variation with $D^{i,k_1}_0=0$;
\item For all $j=0,\ldots,d_{k_1}$, and all $k_1$, the funding accounts $B^{j,k_1}$ are strictly positive and continuous processes of finite variation with $B^{j,k_1}_0=1$;
\item For all $k_0$ and all $k\ne k_0$, the exchange rate processes $\mathcal{X}^{k_0,k}$ are positive-valued RCLL semimartingales.
\end{enumerate}
\end{assumption}
Following \cite{gnoSei2021}, we shall first characterize the absence of arbitrage in a market consisting solely of basic traded assets. A trading portfolio in this market is defined as follows.
\begin{definition}
Let $N_S:= \sum_{k_1=1}^L d_{k_1}$ be the total number of traded assets in all currencies. A \emph{dynamic portfolio} consisting of risky securities and funding accounts is denoted by $\varphi=(\xi, \psi)$, where:
\begin{enumerate}
    \item $\xi \in \RR^{N_S}$ with $\GG$-predictable components, $\xi^{i,k_1}$, representing the number of shares owned on the risky asset $S^{i,k_1}$, for $i=1,\ldots, d_{k_1}$, and $k_1=1,\ldots, L$;
    \item $\psi \in \RR^{N_S+L}$ with $\GG$-predictable components, $\psi^{i,k_1}$, representing the units of cash account on $B^{i,k_1}$, for $i=0,\ldots, d_{k_1}$, and $k_1=1,\ldots, L$. For $i=0$ we use the shorthand $\psi^{k_1}:=\psi^{0,k_1}$. 
\end{enumerate} 
\end{definition}

We denote by $V(\varphi)$ the wealth process of the trading strategy $\varphi$ expressed in currency $k_0$, where we omit the index $k_0$ to simplify the notation. From the definition of trading strategy, it is easy to see that
\begin{align*}
V_t(\varphi)=\sum_{k_1=1}^L\mathcal{X}^{k_0,k_1}_t\left(\sum_{i=1}^{d_{k_1}} \xi_{t}^{i,k_1} S_{t}^{i,k_1}+\sum_{j=0}^{d_{k_1}} \psi_{t}^{j,k_1} B_{t}^{j,k_1}\right).
\end{align*}
The discounted wealth process is $\tilde V(\varphi) := \frac{V(\varphi)}{B^{k_0}}$.

We now introduce the concept of self-financing trading strategy.
\begin{definition}\label{def:firstSelfFinancing} A trading strategy $\varphi$ is \emph{self financing} whenever the wealth process $V(\varphi)$ satisfies
\begin{align*}
\begin{aligned}
V_t(\varphi)&=\sum_{k_1=1}^L\left\{\sum_{i=1}^{d_{k_1}}\left(\int_{(0, t]} \mathcal{X}_{s-}^{k_0, k_{1}} \xi_{s}^{i, k_{1}}d S_{s}^{i, k_{1}}+\int_{(0, t]} \mathcal{X}_{s}^{k_0, k_{1}} \xi_{s}^{i, k_{1}}d D_{s}^{i, k_{1}}\right.\right.\\
&\left.+\int_{(0, t]} \xi_{s}^{i, k_{1}} S_{s-}^{i, k_{1}} d \mathcal{X}_{s}^{k_0, k_{1}}+\int_{(0, t]} \xi_{s}^{i, k_{1}} d\left[S^{i, k_{1}}, \mathcal{X}^{k_0, k_{1}}\right]_s\right)\\
&\left.+\sum_{j=0}^{d_{k_1}}\left(\int_{(0,t]} \mathcal{X}_{s}^{k_0, k_{1}} \psi_{s}^{j, k_{1}} d B_{s}^{j, k_{1}}+\int_{(0,t]}\psi_{s}^{j, k_{1}} B_{s}^{j, k_{1}} d \mathcal{X}_{s}^{k_0, k_{1}}\right)\right\}.
\end{aligned}
\end{align*}
\end{definition}

Our first task is to provide conditions guaranteeing absence of arbitrage in the  basic market consisting only of risky assets and cash account positions.
The concepts of admissibility and of arbitrage opportunity that we consider are the standard ones.

\begin{definition}\label{def:Admissible}
A self-financing trading strategy $\varphi$ is \emph{admissible} for the trader whenever the discounted wealth $\tilde{V}(\varphi)$ is bounded from below by a constant. 
An admissible trading strategy $\varphi$ is an \emph{arbitrage opportunity} whenever $\PP\left(\tilde{V}_T(\varphi)\geq 0\right)=1$ and
 $\PP\left(\tilde{V}_T(\varphi)> 0\right)>0$, for $T>0$.
\end{definition}

In a market with two risk-free assets growing at two different rates it is straightforward to construct an arbitrage opportunity via a long/short strategy. We preclude such arbitrage opportunities via the following condition, which we call the repo constraint
\begin{align}
\label{eq:repoConstraint}
\psi^{i,k_1}_tB^{i,k_1}_t+\xi^{i,k_1}_tS^{i,k_1}_t=0,\quad \mbox{ for every }t\in[0,T], \ i=1,\ldots d_{k_1}, \mbox{ and } 1\leq k_1\leq L.
\end{align}
The repo constraint reflects the realistic situation where the holdings on every risky asset are financed by a position on the asset-specific cash account, and it is not possible to create long-short positions on different cash accounts to produce risk-less profits. 

\begin{lemma}\label{lem:V0dinrepo}
    Under the repo constraint \eqref{eq:repoConstraint}, the discounted portfolio dynamics is
\begin{align*}
\begin{aligned}
d\tilde{V}_t(\varphi) &= \frac{1}{B_t^{k_0}}\sum_{k_1=1}^L  \sum_{i=1}^{d_{k_1}} \xi_t^{i,k_1}\left(dK_t^{i,k_0,k_1}-S^{i,k_1}_{t-}d\mathcal{X}^{k_0,k_1}_t\right)+\sum_{k_1=1,k_1 \neq k_0}^L\psi_t^{k_1}d\left(\frac{B^{k_1}\mathcal{X}^{k_0,k_1}}{B^{k_0}}\right)_t,
\end{aligned}
\end{align*}
where for every $i=1,\ldots d_{k_1}$, and $k_1=1,\ldots,L$, the processes
\begin{align*}
\begin{aligned}
K^{i,k_0,k_1}_t&:=\int_{(0,t]}\left(S^{i,k_1}_{s-}d\mathcal{X}^{k_0,k_1}_s-\frac{S^{i,k_1}_s\mathcal{X}^{k_0,k_1}_s}{B^{i,k_1}_s}dB^{i,k_1}_s+\mathcal{X}^{k_0,k_1}_{s-}dS^{i,k_1}_s+d\left[\mathcal{X}^{k_0,k_1},S^{i,k_1}\right]_s+\mathcal{X}^{k_0,k_1}_sdD^{i,k_1}_s\right)
\end{aligned}
\end{align*}
represent the wealth, denominated in units of currency $k_0$ and discounted by the funding account $B^{i,k_1}$, of a self-financing trading strategy that invests in the asset $S^{i,k_1}$.
\end{lemma}
\begin{proof}
\cite[Corollary 2.15]{gnoSei2021} shows that the portfolio dynamics are of the form
\begin{align*}
\begin{aligned}
d\tilde{V}_t(\varphi)&=\sum_{k_1=1}^L\sum_{i=1}^{d_{k_1}}\frac{1}{B_t^{k_0}}\xi^{i,k_1}_tdK^{i,k_0,k_1}_t+\sum_{k_1=1}^L\sum_{i=1}^{d_{k_1}}\frac{1}{B^{i,k_1}}\left(\psi^{i,k_1}_t B^{i,k_1}_t+\xi^{i,k_1}_t S^{i,k_1}_t\right)\mathcal{X}^{k_0,k_1}_td\left(\frac{B^{i,k_1}}{B^{k_0}}\right)_t\\
&+\sum_{k_1=1}^L\sum_{i=1}^{d_{k_1}}\frac{B^{i,k_1}_t}{B_t^{k_0}}\psi^{i,k_1}_td\mathcal{X}^{k_0,k_1}_t+\sum_{k_1=1}^L \psi^{k_1}_td\left(\frac{\mathcal{X}^{k_0,k_1}B^{k_1}}{B^{k_0}}\right)_t.
\end{aligned}
\end{align*}
Hence substituting the repo constraints \eqref{eq:repoConstraint} and regrouping terms, we get the result.
\end{proof}

The absence of arbitrage in the basic model is characterized by \cite[Proposition 3.4]{gnoSei2021} as follows.

\begin{proposition}\label{prop:aoabasemkt}
Assume that all the strategies available are admissible and satisfy the repo constraint \eqref{eq:repoConstraint}. Then the  multi-currency model is arbitrage free if there exists a probability measure $\QQ^{k_0}\sim\PP$ on $(\Omega,\GG)$, such that the processes
\begin{align}
\label{eq:firstMartingale}
&\left(\int_{(0,t]}\left(\mathcal{X}^{k_0,k_1}_{s-}d\left(\frac{S^{i,k_1}}{B^{i,k_1}}\right)_s+\frac{\mathcal{X}^{k_0,k_1}_s}{B^{i,k_1}_s}dD^{i,k_1}_s+d\left[\frac{S^{i,k_1}}{B^{i,k_1}},\mathcal{X}^{k_0,k_1}\right]_s\right)\right)_{0\leq t\leq T} \mbox{ and}\\
\label{eq:secondMartingale}
&\left(\frac{\mathcal{X}^{k_0,k_1}_tB^{k_1}_t}{B^{k_0}_t}\right)_{0\leq t\leq T}
\end{align}
are $(\QQ^{k_0}, \GG)$-local martingales, for all $i=1,\ldots,d_{k_1}$, and all $k_1=1\ldots,L$.
\end{proposition}

Thanks to Proposition \ref{prop:aoabasemkt} we have the recipe to construct arbitrage-free models. In the next section we shall introduce the concept of collateralization and extend the market model with collateralized contracts.

\section{Pricing under funding costs and collateralization}\label{sec:pricing}
The approach of \cite{gnoSei2021} is based on the assumption that it is possible to replicate contracts with cash flow streams by means of collateralized trading strategies. However, interest rate markets are intrinsically incomplete since interest rates are not traded assets. We then need to adapt the approach of \cite{gnoSei2021} to the setting of martingale modeling. The idea is that introducing a contingent claim with a given and yet-to-be-determined price process into an arbitrage-free market model does not introduce arbitrage opportunities. Another aspect is that the formulas of \cite{gnoSei2021} were derived under a diffusive setting, namely without jumps. We consider a more general setting and work with semimartingales as driving processes for the market.

The approach is as follows: starting from the trading portfolio $V(\varphi)$, we include in the market a collateralized contingent claim with dividend flow. The inclusion of this additional asset must be done in a coherent manner, namely without breaking the martingale property of the market. We start with the following assumption.
\begin{assumption}\label{assu:martingales}
The processes \eqref{eq:firstMartingale} and \eqref{eq:secondMartingale} and stochastic integrals with respect to \eqref{eq:firstMartingale} and \eqref{eq:secondMartingale} are true $(\QQ^{k_0},\GG)$-martingales.
\end{assumption}

Thanks to Lemma \ref{lem:V0dinrepo}, Assumption \ref{assu:martingales} guarantees that $\tilde{V}(\varphi)$ is a true $(\QQ^{k_0},\GG)$-martingale. As a consequence, the basic market consisting only of the primary assets is free of arbitrage opportunities. We now proceed to extend the market by introducing the contingent claim. In particular, we define the dividend flow of a financial contract as in \cite[Definition 2.5]{gnoSei2021}, where we exclude possible cash flows occurring at time zero as these would only represent a shift in the value of the contract.
\begin{definition}\label{def:bilContr}
We define a \emph{financial contract} as an arbitrary RCLL process $A^{k_2}$ of finite variation representing the cumulative cash flows paid by the contract in currency $k_2$ from time $0$ until the maturity date $T$. By convention, we set $A^{k_2}_{0}=0$.
\end{definition}

We now introduce collateralization. This is a legally regulated procedure, according to which the two parties agree to mutually exchange assets (most typically cash) as a guarantee against default risk. Collateralization is in principle possible also between two default-free agents, and represents a standard procedure in the context of futures trading. Moreover, current market quotes are published under the  assumption that instruments are collateralized without referring to any particular economic agent. Hence it is important to study collateralization without introducing a credit model for the agents in the transaction. We refer to \cite{BieRut15} for the single-currency case and to \cite{gnoSei2021} for the cross-currency setting, where collateralization is treated under different legal assumptions, also covering the case of units of risky assets as collateral. We consider the cash-collateral convention, meaning that the collateral is exchanged in cash, i.e., in units of an arbitrary currency $k_3$, and not in terms of units of a risky security. Rehypothecation is also allowed, meaning that the trader can use the cash he/she receives to fund the trading activity. This constitutes the most adopted convention for variation margin. 

We represent the collateral by a right-continuous $\GG$-adapted process $C^{k_3}$ which is received or posted by the trader in units of currency $k_3$, with $k_3=1,\ldots,L$. In particular, for every time instant $t\in[0, T]$, we denote with $C^{k_3,+}_t$ the value of collateral that is received by the trader from the counterparty at time $t$, and with $C^{k_3,-}_t$ the value of collateral that is posted by the trader to the counterparty at time $t$.
We assume that the collateral account satisfies $C^{k_3}_T=0$, meaning that the collateral is returned to its legal owner at the terminal time $T$. We also assume that the agent receives or pays interest contingent on being the poster or the receiver of collateral: the trader receives interest payments based on the rate $r^{c,k_3,l}$ or pays interests based on the rate $r^{c,k_3,b}$. 

We then introduce a contingent claim with dividend process $A^{k_2}$ collateralized by means of $C^{k_3}$, and with price process in domestic currency $S^{k_0}(A^{k_2},C^{k_3})$. We work under the following assumption.
\begin{assumption}\label{martingaleModeling}
The price of the contingent claim $S^{k_0}(A^{k_2},C^{k_3})$ depends only on the yet-to-be-paid cash-flows, i.e. $S_T^{k_0}(A^{k_2},C^{k_3})=0$, $\QQ^{k_0}$-a.s.. Moreover, for every $0\leq t\leq T$, we assume that $S_t^{k_0}(A^{k_2},C^{k_3})$ is integrable with respect to $(\QQ^{k_0},\cG_t)$.
\end{assumption}

We finally define the full-discounted value process of the claim including the evolution of the mark-to-market and the collateralization procedure.
\begin{definition}
The \emph{discounted full-value process} of the collateralized contingent claim $S^{k_0}(A^{k_2},C^{k_3})$ is defined by
\begin{align}
\label{eq:mathscriptM}
\begin{aligned}
\tilde{\mathscr{M}}_t:&=\tilde{S}_t^{k_0}(A^{k_2},C^{k_3})+\int_{(0,t]}\frac{\mathcal{X}^{k_0,k_2}_s}{B^{k_0}_s}dA^{k_2}_s\\
&+ \int_{(0,t]}\left[ \left( r_s^{k_0} - r_s^{c,k_3,b} \right)(C_s^{k_3})^+- \left(r^{k_0}_s - r_s^{c,k_3,l}  \right)(C_s^{k_3})^- \right] \frac{\mathcal{X}_s^{k_0,k_3}}{B_s^{k_0}} ds \\
&- \int_{(0,t]} \frac{C_s^{k_3}}{B_s^{k_0}}\mathcal{X}_s^{k_0,k_3}(r^{k_0}_s-r^{k_3}_s)ds.
\end{aligned}
\end{align}
\end{definition}

The various terms appearing in \eqref{eq:mathscriptM} have the following interpretation: $\tilde{S}^{k_0}_t(A^{k_2},C^{k_3})$ captures the fluctuations of the mark-to-market of the contract that is held by the trader, who also receives dividends during the lifetime of the contract. These dividends are reinvested in the unsecured cash account which produces $\int_{(0,t]}\frac{\mathcal{X}^{k_0,k_2}_s}{B^{k_0}_s}dA^{k_2}_s$. Moreover, the transaction is collateralized. If the trader receives $(C^{k_3})^+$, then he/she reinvests this amount at the unsecured rate $r^{k_0}$ and pays interests to the counterparty at the rate $r^{c,k_3,b}$, thus giving rise to the funding spread $r^{k_0} - r^{c,k_3,b}$. Similar considerations hold for the case when the trader posts the collateral amount $(C^{k_3})^-$. Finally, the last term in \eqref{eq:mathscriptM} captures the fluctuations of the collateral amount due to changes in the foreign exchange rate. 

The next assumption is crucial in order to preserve absence of arbitrage.
\begin{assumption}\label{assu:processMtilde} The process $\tilde{\mathscr{M}}$ and stochastic integrals with respect to $\tilde{\mathscr{M}}$ are true $(\QQ^{k_0},\GG)$-martingales. 
\end{assumption}

We now denote by $\varphi^{ex}=(\varphi,1)$ the self-financing portfolio that invests in the basic traded assets according to $\varphi$, and invests additionally one unit in the claim with discounted full-value process $\tilde{\mathscr{M}}$. The portfolio is admissible in the sense of Definition \ref{def:Admissible}, hence self financing, meaning that
\begin{equation*}
    d\tilde V_t(\varphi^{ex})= d\tilde V_t(\varphi) + d\tilde{\mathscr{M}}_t.
\end{equation*}

Starting from an extended portfolio $\varphi^{ex}$, we derive in the following theorem the price formula for the contingent claim ${S}^{k_0}(A^{k_2},C^{k_3})$.
\begin{theorem}\label{th:priceformula} Let Assumption \ref{assu:martingales}, \ref{martingaleModeling}, \ref{assu:processMtilde} and the repo constraint \eqref{eq:repoConstraint} hold.
Then the price in units of currency $k_0$ of a contingent claim with cash flows $A^{k_2}$ and with collateral $C^{k_3}$ is
\begin{align*}
\begin{aligned}
{S}^{k_0}_t(A^{k_2},C^{k_3})&=B^{k_0}_t\mathbb{E}^{\QQ^{k_0}}\left[\left.\int_{(t,T]}\frac{\mathcal{X}^{k_0,k_2}_s}{B^{k_0}_s}dA^{k_2}_s\right.\right.\\
&\quad\quad\quad\quad\quad +\int_{(t,T]}\left[ \left( r_s^{k_0} - r_s^{c,k_3,b} \right)(C_s^{k_3})^+- \left(r^{k_0}_s - r_s^{c,k_3,l}  \right)(C_s^{k_3})^- \right] \frac{\mathcal{X}_s^{k_0,k_3}}{B_s^{k_0}} ds\\
&\quad\quad\quad\quad\quad\left.\left. - \int_{(t,T]} \frac{C_s^{k_3}}{B_s^{k_0}}\mathcal{X}_s^{k_0,k_3}(r^{k_0}_s-r^{k_3}_s)ds\right|\mathcal{G}_t\right].
\end{aligned}
\end{align*}
\end{theorem}
\begin{proof}
Combining Assumption \ref{assu:martingales} and Assumption \ref{assu:processMtilde}, we deduce that $\tilde{V}_t(\varphi^{ex})$ is a true $(\QQ^{k_0},\GG)$-martingale. 
From the martingale property of $\tilde{V}_t(\varphi^{ex})$ we then get that
\begin{align*}
\begin{aligned}
0=\mathbb{E}^{\QQ^{k_0}}\left[\left.\tilde{V}_T(\varphi^{ex})-\tilde{V}_t(\varphi^{ex})\right|\mathcal{G}_t\right]=\mathbb{E}^{\QQ^{k_0}}\left[\left.\tilde{\mathscr{M}}_T-\tilde{\mathscr{M}}_t\right|\mathcal{G}_t\right],
\end{aligned}
\end{align*}
from which, by Assumption \ref{martingaleModeling}, we deduce the expression for the price of the contingent claim.
\end{proof}

Notice that, modulo the different sign convention, the pricing formula derived in Theorem \ref{th:priceformula} is equivalent to the pricing equation (6.10) in \cite{gnoSei2021}. We stress however, that the present derivation does not assume a diffusive setting (in fact, it does not rely on any explicit dynamics), nor relies on the concept of replication.

We now simplify the setting by introducing the following assumption.
\begin{assumption}
    \label{ass:same_rc}
    We shall assume that $r^{c, k_3, b}= r^{c, k_3, l}$, $\QQ^{k_0}$-a.s. for all $k_3=1, \dots, L$. 
\end{assumption}
We then set $r^{c, k_3}:=r^{c, k_3, b}= r^{c, k_3, l}$, and let $B^{c, k_3}$ be the collateral cash account with interest rate $r^{c, k_3}$, namely 
    \begin{equation}
    \label{eq:Bck3def}
        B^{c,k_3}_t:=\exp\left\{\int_{(0,t]} r^{c, k_3}_sds\right\}.
    \end{equation}

We further say that the contingent claim with discounted full-value process \eqref{eq:mathscriptM} is perfectly or fully collateralized if
\begin{align}
\label{eq:fullcoll}
C^{k_3}=\frac{S^{k_0}(A^{k_2},C^{k_3})}{\mathcal{X}^{k_0,k_3}}, \ d\PP\otimes dt\mbox{-a.s.}.
\end{align}
In this case, the dynamics of $\tilde{\mathscr{M}}$ in \eqref{eq:mathscriptM} simplifies to
\begin{align}\label{eq:mathscriptM2}
d\tilde{\mathscr{M}}_t=\frac{dS^{k_0}_t(A^{k_2},C^{k_3})}{B^{k_0}_t}+\frac{\mathcal{X}^{k_0, k_2}_t}{B^{k_0}_t}dA^{k_2}_t-\frac{S_t^{k_0}(A^{k_2},C^{k_3})}{B^{k_0}_t}\left(r^{c,k_0}_t+q^{k_0,k_3}_t\right)dt,
\end{align}
and we obtain the following corollary.
\begin{corollary}\label{prop:priceformulafullcollateral}
Let Assumption \ref{assu:martingales}, \ref{martingaleModeling}, \ref{assu:processMtilde} and the repo constraint \eqref{eq:repoConstraint} hold. Then the price of a fully collateralized contingent claim with cash-flows $A^{k_2}$ and collateral $C^{k_3}$ is
\begin{align}
\label{eq:fullCollateral}
S_t^{k_0}(A^{k_2},C^{k_3})=B^{c,k_0,k_3}_t\mathbb{E}^{\QQ^{k_0}}\left[\left.\int_{(t,T]}\frac{\mathcal{X}^{k_0,k_2}_s}{B^{c,k_0,k_3}_s}dA^{k_2}_s\right|\mathcal{G}_t\right].
\end{align}
\end{corollary}
\begin{proof}
From Assumption \ref{assu:processMtilde}, we have that
\begin{align*}
0&=\mathbb{E}^{\QQ^{k_0}}\left[\left.\int_{(t,T]}\frac{B^{k_0}_s}{B^{c,k_0,k_3}_s} d\tilde{\mathscr{M}}_s\right|\mathcal{G}_t\right]\\
&=\mathbb{E}^{\QQ^{k_0}}\left[\int_{(t,T]}\frac{B^{k_0}_s}{B^{c,k_0,k_3}_s} \left(\frac{dS^{k_0}_s(A^{k_2},C^{k_3})}{B^{k_0}_s}\left.+\frac{\mathcal{X}^{k_0,k_2}_s}{B^{k_0}_s}dA^{k_2}_s-\frac{S_s^{k_0}(A^{k_2},C^{k_3})}{B^{k_0}_s}\left(r^{c,k_0}_s+q^{k_0,k_3}_s\right)ds\right)\right|\mathcal{G}_t\right]\\
&=\mathbb{E}^{\QQ^{k_0}}\left[\left.\frac{S_T^{k_0}(A^{k_2},C^{k_3})}{B^{c,k_0,k_3}_T}-\frac{S_t^{k_0}(A^{k_2},C^{k_3})}{B^{c,k_0,k_3}_t}+\int_{(t,T]}\frac{\mathcal{X}^{k_0,k_2}_s}{B^{c,k_0,k_3}_s}dA^{k_2}_s\right|\mathcal{G}_t\right].
\end{align*}
From Assumption \ref{martingaleModeling}, we have that $S_T^{k_0}(A^{k_2},C^{k_3})=0$, hence we get the claim.
\end{proof}

The formula obtained in Corollary \ref{prop:priceformulafullcollateral} generalizes (6.23) in \cite{gnoSei2021} to incomplete markets possibly driven by jump-diffusion processes. From \eqref{eq:fullCollateral} we can also obtain generalizations of formulas (6.24)-(6.26) of \cite{gnoSei2021}.
\begin{corollary}
The following pricing formulas can be derived:
\begin{enumerate}
\item $k_0$ cash-flows collateralized in currency $k_0$: this corresponds to $k_2 = k_3 = k_0$ and we obtain
\begin{align*}
\begin{aligned}
S_t^{k_0}(A^{k_0},C^{k_0}) &= \EE^{\QQ^{k_0}} \left[ \left.\int_{(t,T]}e^{-\int_t^s r^{c,k_0}_u du} dA_s^{k_0} \right|\mathcal{G}_t\right] ,
\end{aligned}
\end{align*}
so we discount using the domestic collateral rate. This is the valuation formula employed in the whole literature on single-currency multiple-curve interest rate models.
\item $k_0$ cash-flows collateralized in a currency $k_3$: this corresponds to $k_2=k_0$, $k_3 \neq k_0$ and we obtain
\begin{align*}
\begin{aligned}
S_t^{k_0}(A^{k_0},C^{k_3}) &= \EE^{\QQ^{k_0}} \left[ \left.\int_{(t,T]} e^{-\int_t^s \left(r^{c,k_0}_u+q^{k_0,k_3}_u\right) du} dA_s^{k_0}\right|\mathcal{G}_t \right] ,
\end{aligned}
\end{align*}
so that the foreign collateralization results in the appearance of the cross-currency basis in the discount factor.
\item $k_2$ cash-flows collateralized in currency $k_0$: this corresponds to $k_2\neq k_0$, $k_3 = k_0$ and we obtain
\begin{align*}
\begin{aligned}
S_t^{k_0}(A^{k_2},C^{k_0}) &= \EE^{\QQ^{k_0}} \left[ \left.\int_{(t,T]}e^{-\int_t^s r^{c,k_0}_u du} \mathcal{X}^{k_0,k_2}_sdA_s^{k_2} \right|\mathcal{G}_t\right].
\end{aligned}
\end{align*}
\end{enumerate}
\end{corollary}

We conclude with a remark.
\begin{remark}\label{rem:TerminalPayoff}
For deriving the pricing formula \eqref{eq:fullCollateral}, we assumed that $S_T^{k_0}(A^{k_2},C^{k_3})=0$ $\PP$-a.s.. If the contract pays a single cash-flow at the terminal time $T$, then one obtains the same pricing formulas by alternatively postulating that $A^{k_2}=0$, $d\PP\otimes dt$-a.s., and by treating $S_T^{k_0}(A^{k_2},C^{k_3})\neq 0$ as the random terminal payoff of the contract. For such a simple instrument the two assumptions are equivalent. This alternative viewpoint will be relevant in the next Section \ref{sec:ZCBs} when considering forward measures: in this case we will make use of the well-known property of zero-coupon bonds being equal to one at maturity.
\end{remark}

\subsection{Pricing of zero-coupon bonds}\label{sec:ZCBs}
We have obtained pricing formulas for contingent claims under arbitrary currency configurations for the promised cash flows and for the collateralization agreement. We proceed now to treat zero-coupon bonds as a special case of this setting. This will serve as basis for term-structure models in Sections \ref{sec:HJMs} and \ref{sec:hjmforward}. We shall use the shorthand ZCB for zero-coupon bonds.

\begin{remark}
Due to the martingale properties postulated in Section \ref{sec:pricing}, we find it convenient to work directly with the process $\tilde{\mathscr{M}}$ instead of working with the whole extended portfolio. The derivation of equivalent formulas with the extended portfolio is left to the reader.
\end{remark}

\subsubsection{Domestic ZCB with domestic collateral}\label{sec:domBonddomColl}
Let $T\geq0$ and denote the price process of a domestic ZCB collateralized in domestic currency by $\left\{B^{k_0,k_0}(t,T), \ 0\leq t\leq T\right\}$. In the notation of Section \ref{sec:pricing}, this corresponds to $k_2=k_3=k_0$. Moreover,
\begin{align*}
S_t^{k_0}(A^{k_2},C^{k_3})=B^{k_0,k_0}(t,T), \quad \mbox{ and } \quad  A^{k_2}_t=A^{k_0}_t=\Ind{t=T}.
\end{align*}
Since $\mathcal{X}^{k_0,k_0}=1$, $d\PP\otimes dt$-a.s., full collateralization takes from equation \eqref{eq:fullcoll} the simpler form
\begin{align*}
C^{k_3}_t=C^{k_0}_t=
\frac{B^{k_0,k_0}(t,T)}{\mathcal{X}^{k_0,k_0}_t}=B^{k_0,k_0}(t,T).
\end{align*}
The process $\tilde{\mathscr{M}}$ in \eqref{eq:mathscriptM2} simplifies then to
\begin{align*}
d\tilde{\mathscr{M}}_t=\frac{dB^{k_0,k_0}(t,T)}{B^{k_0}_t}+\frac{1}{B^{k_0}_t}d\Ind{t=T}-\frac{B^{k_0,k_0}(t,T)}{B^{k_0}_t}r^{c,k_0}_tdt,
\end{align*}
and from the pricing formula \eqref{eq:fullCollateral} we get that
\begin{align}
\label{eq:priceBk0k0}
B^{k_0,k_0}(t,T)=B^{c,k_0}_t\mathbb{E}^{\QQ^{k_0}}\left[\left.\frac{1}{B^{c,k_0}_T}\right|\mathcal{G}_t\right],
\end{align}
with $B^{c,k_0}$ the bank account defined in \eqref{eq:Bck3def}.
Equation \eqref{eq:priceBk0k0} represents the pricing formula for a so-called \emph{OIS bond} as in, e.g., \cite{Cuchiero2016} and \cite{Cuchiero2019}. Holdings in these bonds are funded by holdings in the asset-specific cash account $B^{c,k_0}$.

\subsubsection{Domestic ZCB with foreign collateral}\label{sec:case2}
Let $T\geq0$ and denote the price process of a domestic ZCB collateralized in foreign currency by $\left\{B^{k_0,k_3}(t,T), \ 0\leq t\leq T\right\}$.  In the notation of Section \ref{sec:pricing}, this corresponds to $k_2=k_0$ and $k_3\ne k_0$. Moreover,
\begin{align*}
S_t^{k_0}(A^{k_2},C^{k_3})=B^{k_0,k_3}(t,T), \quad \mbox{ and } \quad A^{k_2}_t=A^{k_0}_t=\Ind{t=T}.
\end{align*}
Full collateralization from equation \eqref{eq:fullcoll} means that
\begin{align*}
C^{k_3}_t=
\frac{B^{k_0,k_3}(t,T)}{\mathcal{X}^{k_0,k_3}_t},
\end{align*}
and the process $\tilde{\mathscr{M}}$ in \eqref{eq:mathscriptM2} takes the form
\begin{align*}
d\tilde{\mathscr{M}}_t=\frac{dB^{k_0,k_3}(t,T)}{B^{k_0}_t}+\frac{1}{B^{k_0}_t}d\Ind{t=T}-\frac{B^{k_0,k_3}(t,T)}{B^{k_0}_t}(r^{c,k_0}_t+q^{k_0,k_3}_t)dt.
\end{align*}
From the pricing formula \eqref{eq:fullCollateral} we get that
\begin{align}\label{eq:priceBk0k3}
B^{k_0,k_3}(t,T)=B^{c,k_0,k_3}_t\mathbb{E}^{\QQ^{k_0}}\left[\left.\frac{1}{B^{c,k_0,k_3}_T}\right|\mathcal{G}_t\right],
\end{align}
where $B^{c,k_0,k_3}$ is the bank account defined in \eqref{eq:Bck0k3def}. Notice that the currency dislocation in the collateralization schemes results in the presence of a second term structure of zero-coupon bonds that are funded by means of the newly introduced asset-specific cash account $B^{c,k_0,k_3}$.

\subsubsection{Foreign ZCB with domestic collateral}\label{sec:ZCBk2k0}
Let $T\geq0$ and denote the price process of a foreign ZCB collateralized in domestic currency by $\left\{B^{k_2,k_0}(t,T), \ 0\leq t\leq T\right\}$. In the notation of Section \ref{sec:pricing}, this corresponds to $k_3=k_0$ and $k_2\ne k_0$. Moreover,
\begin{align*}
S_t^{k_0}(A^{k_2},C^{k_3})=\mathcal{X}^{k_0,k_2}_tB^{k_2,k_0}(t,T), \quad \mbox{ and }\quad A^{k_2}_t=\Ind{t=T}.
\end{align*}
Full collateralization  from equation \eqref{eq:fullcoll} means that
\begin{align*}
C_t^{k_3}=
\frac{\mathcal{X}^{k_0,k_2}_tB^{k_2,k_0}(t,T)}{\mathcal{X}^{k_0,k_0}_t}
=\mathcal{X}^{k_0,k_2}_tB^{k_2,k_0}(t,T),
\end{align*}
and the process $\tilde{\mathscr{M}}$  in \eqref{eq:mathscriptM2}  takes the form
\begin{align*}
d\tilde{\mathscr{M}}_t=\frac{d\left(\mathcal{X}^{k_0,k_2}_\cdot B^{k_2,k_0}(\cdot,T)\right)_t}{B^{k_0}_t}+\frac{\mathcal{X}^{k_0,k_2}_t}{B^{k_0}_t}d\Ind{t=T}-\frac{\mathcal{X}^{k_0,k_2}_t B^{k_2,k_0}(t,T)}{B^{k_0}_t}r^{c,k_0}_tdt.
\end{align*}
From the pricing formula \eqref{eq:fullCollateral} we then get that
\begin{align}
\label{eq:domesticZCBforeignColl}
\mathcal{X}^{k_0,k_2}_t B^{k_2,k_0}(t,T)=B^{c,k_0}_t\mathbb{E}^{\QQ^{k_0}}\left[\left.\frac{\mathcal{X}^{k_0,k_2}_T}{B^{c,k_0}_T}\right|\mathcal{G}_t\right],
\end{align}
with $B^{c,k_0}$ the bank account defined in \eqref{eq:Bck3def}.

\subsubsection{Domestic ZCB without collateral}
Let $T\geq0$ and denote the price of a fully unsecured ZCB in domestic currency by $\left\{B^{k_0}(t,T), \ 0\leq t\leq T\right\}$. In the notation of Section \ref{sec:pricing}, this corresponds to $k_2=k_0$ and $C^{k_3}=0$, $d\PP\otimes dt$-a.s.. Then
\begin{align*}
S_t^{k_0}(A^{k_2},C^{k_3})=B^{k_0}(t,T), \quad A^{k_2}_t=A^{k_0}_t=\Ind{t=T},
\end{align*}
and the process $\tilde{\mathscr{M}}$ in \eqref{eq:mathscriptM2}   simplifies significantly to
\begin{align*}
d\tilde{\mathscr{M}}_t=\frac{dB^{k_0}(t,T)}{B^{k_0}_t}+\frac{1}{B^{k_0}_t}d\Ind{t=T}-\frac{B^{k_0}(t,T)}{B^{k_0}_t}r^{k_0}_tdt,
\end{align*}
with $r^{k_0}$ being the unsecured rate. 
From the pricing formula \eqref{eq:fullCollateral} we then get that
\begin{align*}
B^{k_0}(t,T)=B^{k_0}_t\mathbb{E}^{\QQ^{k_0}}\left[\left.\frac{1}{B^{k_0}_T}\right|\mathcal{G}_t\right].
\end{align*}
This corresponds to a textbook pre-financial-crisis ZCB linked to the unsecured bank account $B^{k_0}$.

\section{Measure changes}\label{sec:measureChanges}
We considered so far the domestic risk-neutral measure $\QQ^{k_0}$. This is the measure such that the processes \eqref{eq:firstMartingale} and \eqref{eq:secondMartingale} are martingales, and we have used $\QQ^{k_0}$ as the pricing measure to obtain pricing formulas from the point of view of an agent in the economy $k_0$. We shall introduce in this sections new measures which naturally arise in our framework and which will be crucial for the HJM modelling in Sections \ref{sec:HJMs} and \ref{sec:hjmforward}. More specifically, we shall introduce spot-foreign measures and forward measures.

\subsection{Spot-foreign measures}
Under Assumption \ref{assu:martingales}, we introduce spot-foreign risk-neutral measures as follows.

\begin{definition} \label{def:foreignMeasure} Under Assumption \ref{assu:martingales}, let $1\leq k_2\leq L$ with $k_2\neq k_0$. We define the $\QQ^{k_2}$ \emph{(spot)-foreign risk-neutral measure} $\QQ^{k_2}\sim \QQ^{k_0}$ on $(\Omega, \mathbb{G})$ by
\begin{align*}
\frac{\partial \QQ^{k_2}}{\partial \QQ^{k_0}}:=\frac{B^{k_2}_T\mathcal{X}^{k_0,k_2}_T}{B^{k_0}_T}\frac{B^{k_0}_0}{B^{k_2}_0\mathcal{X}^{k_0,k_2}_0}.
\end{align*}
Due to the martingale property of \eqref{eq:secondMartingale} we have that
\begin{align}\label{eq:changek0k2}
\left.\frac{\partial \QQ^{k_2}}{\partial \QQ^{k_0}}\right|_{\cG_t}=\Excond{\QQ^{k_0}}{\frac{\partial \QQ^{k_2}}{\partial \QQ^{k_0}}}{\cG_t}=\frac{B^{k_2}_t\mathcal{X}^{k_0,k_2}_t}{B^{k_0}_t}\frac{B^{k_0}_0}{B^{k_2}_0\mathcal{X}^{k_0,k_2}_0}, \quad \mbox{ for all } t\le T.
\end{align}
\end{definition}
This family of measures allows to price ZCBs from different point of views, as illustrated in the following two examples.

\subsubsection{Foreign ZCB with foreign collateral under the domestic measure}
We consider a $k_2$-ZCB collateralized in currency $k_2$ from the point of view of the domestic measure $\QQ^{k_0}$. Performing a measure change from $\QQ^{k_0}$ to $\QQ^{k_2}$ allows to obtain the dual formula to \eqref{eq:priceBk0k0}, namely the dual formula to the price of a domestic ZCB with domestic collateral under the domestic measure. In this case, we have $k_3=k_2$ and $k_0\ne k_2$. In the notation of Section \ref{sec:pricing} we then have 
\begin{align*}
S_t^{k_0}(A^{k_2},C^{k_3})=\mathcal{X}^{k_0,k_2}_tB^{k_2,k_2}(t,T), \quad \mbox{ and } \quad A^{k_2}_t=\Ind{t=T}.
\end{align*}
Notice that full collateralization in equation \eqref{eq:fullcoll} takes now the form
\begin{align*}
C^{k_3}_t
=\frac{\mathcal{X}^{k_0,k_2}_tB^{k_2,k_2}(t,T)}{\mathcal{X}^{k_0,k_2}_t}=B^{k_2,k_2}(t,T),
\end{align*}
and the process $\tilde{\mathscr{M}}$ in \eqref{eq:mathscriptM2} satisfies
\begin{align*}
d\tilde{\mathscr{M}}_t=\frac{d\left(\mathcal{X}^{k_0,k_2}_\cdot B^{k_2,k_2}(\cdot,T)\right)_t}{B^{k_0}_t}+\frac{\mathcal{X}^{k_0,k_2}_t}{B^{k_0}_t}dA^{k_2}_t-\frac{\mathcal{X}^{k_0,k_2}_t B^{k_2,k_2}(t,T)}{B^{k_0}_t}\left(r^{c,k_0}_t+q^{k_0,k_2}_t\right)dt.
\end{align*}
From the pricing formula \eqref{eq:fullCollateral} we finally get that
\begin{align}\label{eq:Bk2k2_1}
\mathcal{X}^{k_0,k_2}_t B^{k_2,k_2}(t,T)=B^{c,k_0,k_2}_t\Excond{\QQ^{k_0}}{\frac{\mathcal{X}^{k_0,k_2}_T}{B^{c,k_0,k_2}_T}}{\cG_t},
\end{align}
which provides us with the pricing formula for a $k_2$-foreign bond collateralized in the currency $k_2$ from the point of view of a $k_0$-based agent. By equation \eqref{eq:changek0k2}, we now perform a change of measure on the right-hand side of \eqref{eq:Bk2k2_1} and change to the pricing measure $\QQ^{k_2}$:
\begin{align*}
\mathcal{X}^{k_0,k_2}_t B^{k_2,k_2}(t,T)=B^{c,k_0,k_2}_t\Excond{\QQ^{k_2}}{\frac{\mathcal{X}^{k_0,k_2}_T}{B^{c,k_0,k_2}_T}\frac{B^{k_0}_T}{\mathcal{X}^{k_0,k_2}_TB^{k_2}_T}}{\cG_t}\frac{\mathcal{X}^{k_0,k_2}_t B^{k_2}_t}{B^{k_0}_t}.
\end{align*}
Notice that
\begin{align*}
\frac{B^{k_0}_T}{B^{c,k_0,k_2}_TB^{k_2}_T}=\exp\left\{\int_0^T\left(r^{k_0}_s-r^{k_2}_s-(r^{c,k_0}_s+q^{k_0,k_2}_s)\right)ds\right\}=\exp\left\{-\int_0^Tr^{c,k_2}_sds\right\}=\frac{1}{B^{c,k_2}_T},
\end{align*}
hence equation \eqref{eq:Bk2k2_1} simplifies to
\begin{align*}
B^{k_2,k_2}(t,T)=B^{c,k_2}_t\Excond{\QQ^{k_2}}{\frac{1}{B^{c,k_2}_T}}{\cG_t}.
\end{align*}
This is the formula for the same contract under the dual measure $\QQ^{k_2}$, which is the same as the previously obtained formula \eqref{eq:priceBk0k0} in Section \ref{sec:domBonddomColl} for $k_2=k_0$.

\subsubsection{Foreign ZCB with domestic collateral under the foreign measure}
We consider a $k_2$-ZCB with domestic collateral as in Section \ref{sec:ZCBk2k0}. With a measure change from $\QQ^{k_0}$ to $\QQ^{k_2}$ we will obtain a dual valuation formula under the foreign measure which is consistent with equation \eqref{eq:priceBk0k3} of Section \ref{sec:case2}, namely with the pricing formula of a domestic ZCB with foreign collateralization. 

Starting from equation \eqref{eq:domesticZCBforeignColl}, we perform a change of measure to $\QQ^{k_2}$ accordingly to equation \eqref{eq:changek0k2}:
\begin{align*}
\mathcal{X}^{k_0,k_2}_t B^{k_2,k_0}(t,T)=B^{c,k_0}_t\mathbb{E}^{\QQ^{k_0}}\left[\left.\frac{\mathcal{X}^{k_0,k_2}_T}{B^{c,k_0}_T}\right|\mathcal{G}_t\right]
=B^{c,k_0}_t\Excond{\QQ^{k_2}}{\frac{\mathcal{X}^{k_0,k_2}_T}{B^{c,k_0}_T}\frac{B^{k_0}_T}{B^{k_2}_T\mathcal{X}^{k_0,k_2}_T}}{\cG_t}\frac{B^{k_2}_t\mathcal{X}^{k_0,k_2}_t}{B^{k_0}_t},
\end{align*}
which leads to
\begin{align}\label{eq:changek0k2_1}
B^{k_2,k_0}(t,T)=\frac{B^{c,k_0}_tB^{k_2}_t}{B^{k_0}_t}\Excond{\QQ^{k_2}}{\frac{B^{k_0}_T}{B^{c,k_0}_TB^{k_2}_T}}{\cG_t}.
\end{align}
Notice that
\begin{align*}
\frac{B^{k_0}_T}{B^{c,k_0}_TB^{k_2}_T}=\exp\left\{\int_0^T \left(r^{k_0}_s-r^{c,k_0}_s-r^{k_2}_s\right)ds\right\}
=\exp\left\{-\int_0^T \left(r^{c,k_2}_s+q^{k_2,k_0}_s\right)ds\right\} = \frac{1}{B_T^{c, k_2, k_0}(t, T)},
\end{align*}
hence \eqref{eq:changek0k2_1} becomes
\begin{align*}
B^{k_2,k_0}(t,T)
=B^{c,k_2,k_0}_t\Excond{\QQ^{k_2}}{\frac{1}{B^{c,k_2,k_0}_T}}{\cG_t}.
\end{align*}
This corresponds to \eqref{eq:priceBk0k3} with flipped currency indices.

\subsection{Forward measures}
Recall from Remark \ref{rem:TerminalPayoff}, that ZCBs are contracts with zero dividend process and a terminal price of one unit of currency. Then, from Section \ref{sec:ZCBs}, we obtain that, under Assumption \ref{assu:martingales}, \ref{martingaleModeling}, \ref{assu:processMtilde} and under the repo constraint \eqref{eq:repoConstraint}, the processes
\begin{align}\label{eq:martingales}
\left(\frac{B^{k_0,k_0}(t,T)}{B^{c,k_0}_t}\right)_{0\leq t\leq T}, \, \left(\frac{B^{k_0,k_3}(t,T)}{B^{c,k_0,k_3}_t}\right)_{0\leq t\leq T}, \, 
\mbox{ and } \,\left(\frac{B^{k_0}(t,T)}{B^{k_0}_t}\right)_{0\leq t\leq T},
\end{align}
are martingales for every choice of the indices $k_1$, $k_2$, and $k_3$. This shows that including ZCBs with different funding strategies corresponds to including new risky assets together with their asset-specific cash-accounts: each new asset is funded by an associated asset-specific cash account, hence giving rise to further repo constraints of the form \eqref{eq:repoConstraint}.

Furthermore, the martingales in \eqref{eq:martingales} may serve as density processes for new probability measures. We shall introduce some of them in the following definition: the list is not exhaustive, but covers all the essential tools that are needed for cross-currency term-structure modeling. In particular, notice that classical forward measure are defined up to the maturity of the corresponding ZCB. For later use, we adopt the approach of \cite{LyasMer2019} and extend the concept of forward measure by considering as num\'eraire the self-financing strategy that  after the maturity of the corresponding ZCB, say $T$, reinvests the notional into the corresponding cash account.

\begin{definition}\label{def:forwardmeasure}
Let $T\geq 0$ be fixed. We define the following forward measures:
\begin{enumerate}
\item The \emph{domestic-collateralized domestic $T$-forward measure} $\QQ^{T,k_0,k_0}\sim\QQ^{k_0}$ on $(\Omega, \mathbb{G})$ is defined via the Radon-Nikodym derivative 
\begin{align*}
\frac{\partial \QQ^{T,k_0,k_0}}{\partial \QQ^{k_0}}:=\frac{B^{k_0,k_0}(T,T)}{B^{c,k_0}_T}\frac{B^{c,k_0}_0}{B^{k_0,k_0}(0,T)}.
\end{align*}
In particular: 
\begin{enumerate}
    \item For $t\leq T$ we have
    \begin{align*}
\left.\frac{\partial \QQ^{T,k_0,k_0}}{\partial \QQ^{k_0}}\right|_{\cG_t}=\Excond{\QQ^{k_0}}{\frac{\partial \QQ^{T,k_0,k_0}}{\partial \QQ^{k_0}}}{\cG_t}=\frac{B^{k_0,k_0}(t,T)}{B^{c,k_0}_t}\frac{B^{c,k_0}_0}{B^{k_0,k_0}(0,T)};
\end{align*}
\item For $t>T$, since $B^{k_0,k_0}(t,T)=\frac{B_t^{c,k_0}}{B_T^{c,k_0}}$, we define
\begin{align*}
\left.\frac{\partial \QQ^{T,k_0,k_0}}{\partial \QQ^{k_0}}\right|_{\cG_t}:=\frac{1}{B^{c,k_0}_T}\frac{B^{c,k_0}_0}{B^{k_0,k_0}(0,T)},
\end{align*}
hence $\QQ^{T,k_0,k_0}\equiv\QQ^{k_0}$ on $t>T$.
\end{enumerate}
\item The \emph{$k_3$-collateralized domestic $T$-forward measure} $\QQ^{T,k_0,k_3}\sim\QQ^{k_0}$ on $(\Omega, \mathbb{G})$ is defined via the Radon-Nikodym derivative 
\begin{align}\label{eq:QTk0k3}
\frac{\partial \QQ^{T,k_0,k_3}}{\partial \QQ^{k_0}}:=\frac{B^{k_0,k_3}(T,T)}{B^{c,k_0,k_3}_T}\frac{B^{c,k_0,k_3}_0}{B^{k_0,k_3}(0,T)}.
\end{align}
In particular:
\begin{enumerate}
    \item For $t\leq T$ we have
\begin{align*}
\left.\frac{\partial \QQ^{T,k_0,k_3}}{\partial \QQ^{k_0}}\right|_{\cG_t}=\Excond{\QQ^{k_0}}{\frac{\partial \QQ^{T,k_0,k_3}}{\partial \QQ^{k_0}}}{\cG_t}=\frac{B^{k_0,k_3}(t,T)}{B^{c,k_0,k_3}_t}\frac{B^{c,k_0,k_3}_0}{B^{k_0,k_3}(0,T)};
\end{align*}
\item For $t>T$, since $B^{k_0,k_3}(t,T)=\frac{B_t^{c,k_0,k_3}}{B_T^{c,k_0,k_3}}$, we define
\begin{align*}
\left.\frac{\partial \QQ^{T,k_0,k_3}}{\partial \QQ^{k_0}}\right|_{\cG_t}:=\frac{1}{B^{c,k_0,k_3}_T}\frac{B^{c,k_0,k_3}_0}{B^{k_0,k_3}(0,T)},
\end{align*}
hence $\QQ^{T,k_0,k_3}\equiv\QQ^{k_0}$ on $t>T$.
\end{enumerate}
\item The \emph{uncollateralized/unsecured domestic $T$-forward measure} $\QQ^{T,k_0}\sim\QQ^{k_0}$ on $(\Omega, \mathbb{G})$ is defined via the Radon-Nikodym derivative
\begin{align*}
\frac{\partial \QQ^{T,k_0}}{\partial \QQ^{k_0}}:=\frac{B^{k_0}(T,T)}{B^{k_0}_T}\frac{B^{k_0}_0}{B^{k_0}(0,T)}.
\end{align*}
In particular:
\begin{enumerate}
    \item For $t\leq T$ we have
    \begin{align*}
\left.\frac{\partial \QQ^{T,k_0}}{\partial \QQ^{k_0}}\right|_{\cG_t}=\Excond{\QQ^{k_0}}{\frac{\partial \QQ^{T,k_0}}{\partial \QQ^{k_0}}}{\cG_t}=\frac{B^{k_0}(t,T)}{B^{k_0}_t}\frac{B^{k_0}_0}{B^{k_0}(0,T)};
\end{align*}
\item For $t>T$, since $B^{k_0}(t,T)=\frac{B_t^{k_0}}{B_T^{k_0}}$, we define
\begin{align*}
\left.\frac{\partial \QQ^{T,k_0}}{\partial \QQ^{k_0}}\right|_{\cG_t}:=\frac{1}{B^{k_0}_T}\frac{B^{k_0}_0}{B^{k_0}(0,T)},
\end{align*}
hence $\QQ^{T,k_0}\equiv\QQ^{k_0}$ on $t>T$.
\end{enumerate}
\end{enumerate} 
\end{definition}
Figure \ref{fig:pricingMeasures} summarizes some of the relations between the different pricing measures.

\begin{remark} We observe that the domestic-collateralized domestic $T$-forward measure that uses $T$-OIS bonds as num\'eraire is the one typically employed in the literature on single-currency multiple-curve models, such as \cite{Cuchiero2016} and \cite{Cuchiero2019}. Our general setting, however, highlights the fact that there is no need to assume (as in the references above) that the OIS bank account is the num\'eraire of $\QQ^{k_0}$ (in fact, it is not): under $\QQ^{k_0}$, as we have seen, we have that multiple assets with different funding strategies are simultaneously martingales with no cash account playing the role of universal num\'eraire for all the risky assets.
\end{remark}

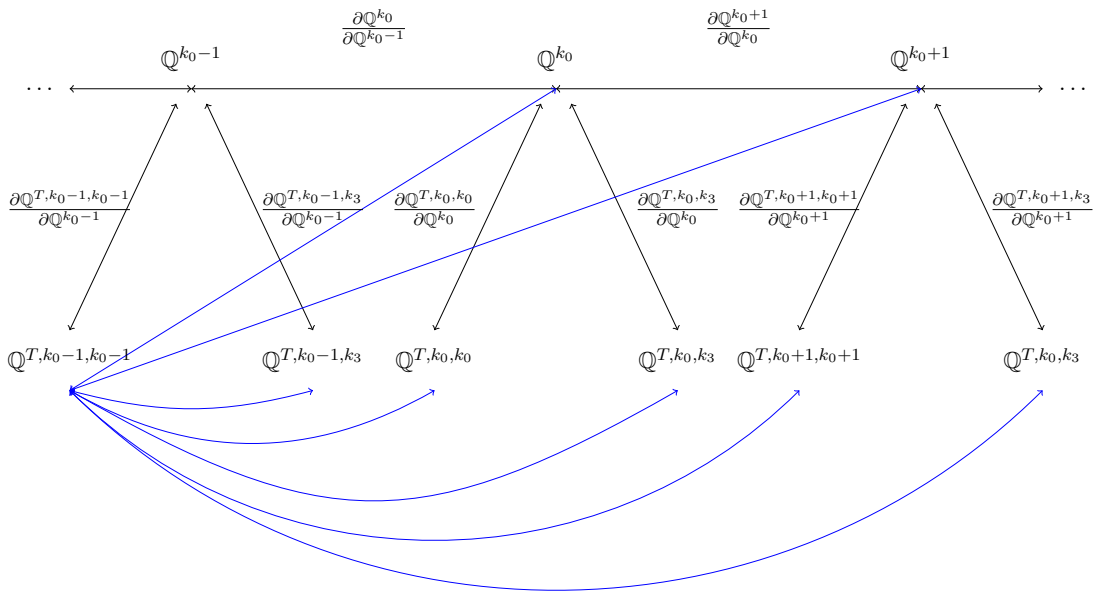
\begin{figure}[t]
    \centering
    \scalebox{0.8}{
   \begin{tikzpicture}
	\begin{pgfonlayer}{nodelayer}
		\node [style=none] (3) at (2, 0) {};
		\node [style=none] (4) at (-2, 0) {};
		\node [style=none] (5) at (4, 0) {};
		\node [style=none] (7) at (8, 0) {};
		\node [style=none] (9) at (-4, 0) {};
		\node [style=none] (11) at (-8, 0) {};
		\node [style=none] (12) at (-6, 4) {};
		\node [style=none] (13) at (-6, 4) {};
		\node [style=none] (14) at (-6, 4) {};
		\node [style=none] (15) at (-6.25, 3.75) {};
		\node [style=none] (17) at (-6, 4.5) {};
		\node [style=none] (18) at (-6, 4.5) {$\mathbb{Q}^{k_0-1}$};
		\node [style=none] (19) at (0, 4.5) {};
		\node [style=none] (20) at (0, 4.5) {$\mathbb{Q}^{k_0}$};
		\node [style=none] (21) at (6, 4.5) {};
		\node [style=none] (22) at (6, 4.5) {$\mathbb{Q}^{k_0+1}$};
		\node [style=none] (23) at (-8, -0.5) {$\mathbb{Q}^{T,k_0-1,k_0-1}$};
		\node [style=none] (24) at (-4, -0.5) {$\mathbb{Q}^{T,k_0-1,k_3}$};
		\node [style=none] (25) at (-2, -0.5) {$\mathbb{Q}^{T,k_0,k_0}$};
		\node [style=none] (26) at (2, -0.5) {$\mathbb{Q}^{T,k_0,k_3}$};
		\node [style=none] (27) at (-5.75, 3.75) {};
		\node [style=none] (28) at (0, 4) {};
		\node [style=none] (29) at (0, 4) {};
		\node [style=none] (30) at (0, 4) {};
		\node [style=none] (31) at (-0.25, 3.75) {};
		\node [style=none] (32) at (0.25, 3.75) {};
		\node [style=none] (35) at (6, 4) {};
		\node [style=none] (36) at (6, 4) {};
		\node [style=none] (37) at (6, 4) {};
		\node [style=none] (38) at (5.75, 3.75) {};
		\node [style=none] (39) at (6.25, 3.75) {};
		\node [style=none] (40) at (4, -0.5) {$\mathbb{Q}^{T,k_0+1,k_0+1}$};
		\node [style=none] (41) at (8, -0.5) {$\mathbb{Q}^{T,k_0,k_3}$};
		\node [style=none] (42) at (4, 2) {$\frac{\partial \mathbb{Q}^{T,k_0+1,k_0+1}}{\partial \mathbb{Q}^{k_0+1}}$};
		\node [style=none] (43) at (7, 2) {};
		\node [style=none] (44) at (3, 4) {};
		\node [style=none] (45) at (-3, 4) {};
		\node [style=none] (46) at (1, 2) {};
		\node [style=none] (47) at (-2, 2) {$\frac{\partial \mathbb{Q}^{T,k_0,k_0}}{\partial \mathbb{Q}^{k_0}}$};
		\node [style=none] (48) at (-5, 2) {};
		\node [style=none] (49) at (-8, 2) {$\frac{\partial \mathbb{Q}^{T,k_0-1,k_0-1}}{\partial \mathbb{Q}^{k_0-1}}$};
		\node [style=none] (50) at (-4, 2) {};
		\node [style=none] (51) at (-4, 2) {$\frac{\partial \mathbb{Q}^{T,k_0-1,k_3}}{\partial \mathbb{Q}^{k_0-1}}$};
		\node [style=none] (52) at (2, 2) {};
		\node [style=none] (53) at (2, 2) {$\frac{\partial \mathbb{Q}^{T,k_0,k_3}}{\partial \mathbb{Q}^{k_0}}$};
		\node [style=none] (54) at (8, 2) {};
		\node [style=none] (55) at (8, 2) {$\frac{\partial \mathbb{Q}^{T,k_0+1,k_3}}{\partial \mathbb{Q}^{k_0+1}}$};
		\node [style=none] (56) at (-3, 5) {$\frac{\partial \mathbb{Q}^{k_0}}{\partial \mathbb{Q}^{k_0-1}}$};
		\node [style=none] (57) at (3, 5) {$\frac{\partial \mathbb{Q}^{k_0+1}}{\partial \mathbb{Q}^{k_0}}$};
		\node [style=none] (58) at (8, 4) {};
		\node [style=none] (59) at (-8, 4) {};
		\node [style=none] (60) at (-8, -1) {};
		\node [style=none] (61) at (-4, -1) {};
		\node [style=none] (62) at (-2, -1) {};
		\node [style=none] (63) at (2, -1) {};
		\node [style=none] (64) at (4, -1) {};
		\node [style=none] (65) at (8, -1) {};
		\node [style=none] (66) at (-8.5, 4) {$\ldots$};
		\node [style=none] (67) at (8.5, 4) {$\ldots$};
	\end{pgfonlayer}
	\begin{pgfonlayer}{edgelayer}
		\draw [style=doubleArrow] (15.center) to (11.center);
		\draw [style=doubleArrow, in=180, out=0] (14.center) to (30.center);
		\draw [style=doubleArrow] (27.center) to (9.center);
		\draw [style=doubleArrow] (31.center) to (4.center);
		\draw [style=doubleArrow] (32.center) to (3.center);
		\draw [style=doubleArrow] (30.center) to (37.center);
		\draw [style=doubleArrow] (38.center) to (5.center);
		\draw [style=doubleArrow] (39.center) to (7.center);
		\draw [style=doubleArrow] (59.center) to (14.center);
		\draw [style=doubleArrow] (58.center) to (37.center);
		\draw [style=new edge style 2, bend right=15] (60.center) to (61.center);
		\draw [style=new edge style 2, bend right] (60.center) to (62.center);
		\draw [style=new edge style 2, bend right, looseness=1.25] (60.center) to (63.center);
		\draw [style=new edge style 2, bend right=45] (60.center) to (64.center);
		\draw [style=new edge style 2, bend right=45] (60.center) to (65.center);
		\draw [style=new edge style 2] (60.center) to (30.center);
		\draw [style=new edge style 2] (60.center) to (37.center);
	\end{pgfonlayer}
\end{tikzpicture}
    }
    \caption{This graph summarizes (a part of) the relations between the different pricing measures. Each node in the graph is a pricing measure and each edge represents the link between two probability measures via a suitable Radon-Nikodym derivative. For the sake of readability, we only plot the Radon-Nikodym derivatives with respect to the spot measures. However, each node of the graph can be linked to all the other nodes. For illustrative purposes, we represent these relations only for the forward measure $\QQ^{T,k_0-1,k_0-1}$ (blue arrows).}
    \label{fig:pricingMeasures}
\end{figure}

\bibliographystyle{plainnat}
\bibliography{references}
 
\end{document}